\documentclass[a4paper,titlepage]{article}
\usepackage[utf8]{inputenc}
\usepackage[T1]{fontenc}
\usepackage{geometry}
\usepackage{float}
\usepackage{subcaption}
\usepackage{dsfont}
\usepackage{bigints}
\usepackage{pdflscape}
\usepackage{pdfpages}
\usepackage{graphicx}
\usepackage{color, colortbl}
\usepackage{amsmath,amsfonts,amssymb,amsthm}
\usepackage{multirow}
\usepackage{tabu}
\usepackage{xcolor}
\numberwithin{equation}{section} 
\usepackage[nodisplayskipstretch]{setspace} 
\newtheorem{proposition}{Proposition}
\newtheorem{lemma}{Lemma}
\newtheorem{problem}{Problem}
\newtheorem{modification}{Modification}
\newtheorem{result}{Result}
 \usepackage{parskip}
 \usepackage{wrapfig}
 \usepackage{rotating}
 \usepackage{listings} 

\begin{document}
\begin{titlepage}

\begin{center}
	\includegraphics[scale=1.5]{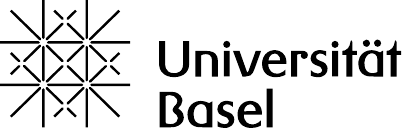}
\end{center}
\vspace{2.5cm}
\begin{center}                     
        {\Large\scshape Master's Thesis}\\*[5mm]
				{\Large\scshape M.Sc. Course ``Actuarial Science"}\\*[5mm]
        {\bf\Large\scshape\, \,       Retirement wealth under fixed limits:\newline the optimal strategy for exponential utility}\\*[12mm]         
\end{center}  
\vspace{4,5cm}
\begin{tabbing}
       submitted by:		\hspace{6.5cm}\=			supervised by:\\*[2mm]
       \bf{Lena Schütte} 								\> \bf{Dr. Catherine Donnelly}  \\*[2mm]		
        Matriculation No.: 08-933-848 		\> \bf{Dr. Michael Schmutz}  \\*[2mm]	
       \\ Email: lena.schuette@stud.unibas.ch									\> Submission date, place: \\*[2mm]
        \textbf{ }		   				\> \textbf{30.01.2017, Edinburgh}\\*[2mm]
				\\*[2mm]
        \\*[2mm]
\end{tabbing}
\end{titlepage}
\newpage

\definecolor{LightMintGreen}{rgb}{0.961,1,0.961}
\definecolor{Lavender}{rgb}{0.9812,0.97,1}
\definecolor{OldLace}{rgb}{0.99916,0.99078,0.931961}
\definecolor{LightCyan}{rgb}{0.93999,1,1}
\definecolor{LightGrey}{rgb}{0.93,0.93,0.93}
\definecolor{DarkBlue}{rgb}{0,0,0.4}

\newpage
\thispagestyle{empty}
\pagenumbering{Roman}
\setcounter{page}{1}

\newpage
\section*{Acknowledgements}

All my gratitude and thanks go to Dr.Catherine Donnelly for accepting me as a visiting scholar at Heriot-Watt University in Edinburgh and without whom this thesis would not have been possible. Her valuable advice and constant guidance as well as the friendly support in all matters related to my stay here, have made the last four months a very exciting and enriching time, both scientifically and personally. \newline
I would also like to thank the Department of Actuarial Mathematics and Statistics at Heriot-Watt for providing great facilities and my collegues from the office for creating a fun and motivating atmosphere. \newline
Further I would like to thank Dr. Michael Schmutz for agreeing to be my supervisor from the part of the University of Basel, Jolanda Bucher for assisting with the administrative procedures and Matthias Kohlbrenner for reviewing the paper.\newline 
A grant from the Swiss-European Mobility Progamme via the Mobility Office of the University of Basel is also gratefully acknowledged.

\newpage

\tableofcontents
\newpage

\addcontentsline{toc}{section}{List of Figures}
\listoffigures
\newpage

\addcontentsline{toc}{section}{List of Tables}
\listoftables
\newpage


\section*{List of Symbols}
\addcontentsline{toc}{section}{List of Symbols}
\renewcommand\arraystretch{1.3}
\begin{tabular}{p{4cm} l }

$X_0$ &initial wealth (independent of strategy) \\
$\hat\pi$ &	 optimal unconstrained strategy \\  
$X^{\hat{\pi}}_t$&wealth process under $\hat\pi$ \\
$\hat{\pi}_m$		& optimal strategy modified by restriction on investment\\
$X^{\hat{\pi}_m}_t$& wealth process under $\hat{\pi}_m$\\
 \\
$K_l$					 &lower bound for terminal wealth \\
$\hat{\pi}_l$ & optimal strategy under a lower bound for terminal wealth \\
$X^{\hat{\pi}_l}_t$&wealth process under $\hat{\pi_l}$\\
$\hat{\pi}_{l,m}$ & $\hat{\pi}_l$  modified by restriction on investment \\
$\tilde{X}_0$&shadow wealth \\
$\tilde{X}^{\hat{\pi}}_t$&shadow process (wealth process under $\hat{\pi}$ for $\tilde{X}_0$\\
$p(t,\tilde{X}^{\hat{\pi}}_t)$&price of put option on shadow wealth (strike price $K_l$) \\
$\tilde{\pi}_p$&strategy for replication of $p(t,\tilde{X}^{\hat{\pi}}_t)$ \\
 \\
$K_u$&upper bound for terminal wealth\\
$\hat{\pi_u}$& optimal strategy under an upper constraint for terminal wealth\\
$X^{\hat{\pi_u}}_t$&wealth process under $\hat{\pi_u}$ \\
$c(t,\tilde{X}^{\hat{\pi}}_t)$ &price of call option on shadow wealth (strike price $K_u$) \\
$\tilde{\pi}_c$&strategy for replication of $c(t,\tilde{X}^{\hat{\pi}}_t)$ \\
\end{tabular}

\newpage

\pagenumbering{arabic}
\setcounter{page}{1}
\section{Introduction}

What is the best way to invest money for retirement? This question might be more relevant than ever, since private and institutional investors face a challenging low-interest market environment and growing retirement needs. At the same time, this question has been widely researched in financial mathematics and economics and offers many interesting approaches. Among these, using the \textit{utility} of wealth for an investor instead of, for example, the simple return, as a criterion seems to best reflect the investor's needs. The strategies that maximize the expected utility are commonly called \textit{optimal} strategies and they can only be derived analytically for few utility functions. The exponential utility function is one of them, and in this thesis it will be used to determine the optimal strategy in a simple Black-Scholes-Setting.\newline
Besides developing a good understanding of the resulting strategy and its effects on the wealth at retirement, we are particularly interested in improving its potential while still taking the investor's needs into consideration. The idea is therefore to introduce upper and lower constraints on the resulting wealth: Utility theory suggests that investors are more sensitive towards lower values of wealth, so they might be ready to give up some investment potential in exchange for a garantee on a minimal return.  From another perspective, it could be favorable to constrain the wealth at retirement to a maximum amount (for example, the present value of annuities) and be compensated by higher probabilites for greater returns on the wealth below this maximum value. \newline
It is one aim of this thesis to explore the consequences on the optimal strategy for exponential utility with retirement wealth facing upper and lower constraints. In this sense, this paper can be seen as a complement to the research done in \cite{Donnelly}, where a power utility function is considered. \newline
The second modification of the optimal strategy that is investigated in this thesis is introduced in order to avoid debts. Since the optimal strategy might involve borrowing money or short-selling, there is a risk that the investor ends up with a negative wealth. This is why we would like to limit the investment to a maximum of 100$\%$ of wealth. This restriction will be implemented on the pre-existing strategies we developed, its consequences will therefore be assessed from empirical results only. 
\newpage We will proceed step by step and gradually adapt the strategy. \newline As a basis, we will derive the optimal unconstrained strategy in Chapter 2, using stochastic optimal control arguments. We will then analyse it briefly with respect to investor-related parameters and focus an the analysis of the optimal strategy with a restriction on investment. \newline \newline In Chapter 3, the optimal strategy where terminal wealth faces only a lower constraint will be developed. This will be done by formulating a dual problem and solving this via risk-neutral valuation. We will then see that the resulting optimal strategy corresponds to the optimal unconstrained strategy combined with a put option. After some qualitative analysis, we will then implement the investment-restriction and see how it affects the strategy. Finally, an emphasis will be put on the theoretical distribution of the resulting wealth (without constraints) and the error produced by the implementation of the modified strategies. \newline\newline
In the last chapter, we will add an upper constraint to the problem. To do this, we will first find the strategy for the isolated case of an upper constraint by the similar methods as used before. Then we will combine it with the results from the previous chapter. It turns out, that in addition to the put option bought, the  'combined' strategy requires to sell a call option. For further analysis of this strategy, we will first briefly investigate the isolated case of an upper constraint and then try to outline how the choice for the upper and lower constraints affect the 'combined' stategies, both qualitatively and with respect to the distribution of terminal wealth.
\newline \newline
Finally, the Appendix is thought to gather background information on the theory used, as well as complementary analysis to validate the empirical results. 


\newpage
\section{An Optimal Strategy for Exponential Utility}
We will start by finding the strategy that maximizes the expected exponential utility of terminal wealth and briefly analyze its results. Further, the strategy is modified by introducing a restriction on the amount invested. The resulting strategy and terminal wealth distribution will then be investigated in more detail. 
\subsection{Derivation of the Optimal Strategy}
In this section, we will  introduce the formal setting and derive the optimal unconstrained strategy by solving an differential equation that characterises optimal strategies. We will see that it requires a deterministic amount to be invested in the risky asset, independently of the investor's wealth or the stock's performance, but growing by the risk-free rate. The resulting terminal wealth is normally distributed.
\subsubsection{Market Model and Hamilton-Jacobi-Bellman Equation}
We assume the Black-Scholes market model consisting of one risky stock and one risk-free bond, available in the continous time interval [0,\textit{T}]. The integer \textit{T} > 0 denotes the \textit{terminal time}, for example the moment of retirement. The price of the bond at time \textit{t} is given by the deterministic price process   \{\textit{B(t), t $\in$} [0,\textit{T}]\} with dynamics 
\begin{equation} dB(t) = rB(t)dt, \end{equation} where \textit{r} >0 is the risk-free interest rate and\textit{ B(0)} = 1 almost surely (abbreviated a.s.). \newline The performance of the risky stock at time \textit{t} is given by the stochastic price process\newline   \{S(t), t $\in$  [0,\textit{T}] \}  with dynamics \begin{equation} dS(t) = \mu S(t)dt + \sigma S(t) dW(t), \end{equation} where $\sigma$ > 0 , \textit{S(0)} = 1 a.s. , $\mu$ > \textit{r} and \textit{W(t)} is the 1-dimensional standard Brownian motion defined on a complete probability space ($\Omega$, \(\mathcal{F}, \mathbb{P}\)).
\newline The information available up to time t is represented by the filtration \begin{center} \(\mathcal{F}_t = \sigma \{ W(s),s \in [0,t]\} \vee \mathcal{N}(P), \)\end{center} where $\mathcal{N}(\mathbb{P})$ denotes the collection of all $\mathbb{P}$-null events in the probability space. Further, call \newline\newline $\pi = \{\pi(t)$ is a $\mathbb{R}$-valued, $\mathcal{F}_t$ -progressively measurable process and $\int_0^t \pi^2(s)ds$<$\infty$ $\forall t \in [0, \textit{T}]\}$\newline \newline a \textit{portfolio}, and $\pi(t)$ is the proportion of wealth invested in the risky asset at time t.\newline
We assume that the investor follows a \textit{self-financing} strategy, which means that wealth gains or losses arise solely from investment gains or losses. Then, the corresponding wealth at time t, $X^{\pi}(t)$, can be described by the dynamics  \begin{center} \( dX^{\pi}(t)=\pi(t) X^{\pi}(t) \dfrac{dS(t)}{S(t)} + (1-\pi(t)) X^{\pi}(t) \dfrac{dB(t)}{B(t)}\). \end{center}
Assuming that the investor starts with a fixed, positive wealth \textit{x} at time 0, and substituting (2.1) and (2.2) into the dynamics, this gives the \textit{wealth process} defined by the \textit{wealth equation}: \begin{equation} dX^{\pi}(t)=(rX^{\pi}(t)+ \pi (t) (\mu -r)X^{\pi}(t))dt + \sigma \pi(t) X^{\pi}(t)dW(t) \ \mathrm {and} \ X^{\pi}(0) = \mathrm{\textit{x} \ \textit{a.s.}}, \end{equation} where \textit{x} $\in$ $\mathbb{R^{+}}$. For better readability let  $X^{\pi}_t$:=$X^{\pi}$(t) and $\pi_t$:=$\pi$(t).
\newline Let the set of \textit{admissible portfolios} be defined as \begin{center} \( \mathcal{A} := \{ \pi: \Omega \times [0,T] \rightarrow \mathbb{R}| X^{\pi}_{0} = x \; \; a.s. \;  and\; \;  \pi \; is \: \; a \;\;  self \text{-}financing \; portfolio \} \) \end{center}  and $\pi$ be called \textit{admissible} if $\pi$ $\in$ $\mathcal{A}$. \newline Also, define the \textit{state price density process} $H(t) := e^{-(r+ \dfrac{\theta^2}{2})t-\theta W_t}$.
\newline
To state the investor's problem, we first let the utility of wealth be described by a function (called $\textit{utility}$ $\textit{function}$) \begin{center} U: $\mathbb{R}^{+}\rightarrow \mathbb{R} , \   x \longmapsto U(x)$.\end{center}  The problem for the investor is then defined as follows:
\begin{problem} \ \ Find a strategy \( \hat{\pi} \in \mathcal{A} \) such that 
\begin{equation}  \mathbb{E}[U(X^{\hat{\pi}}_T)] = \sup\limits_{\pi\in \mathcal{A}} \mathbb{E}[U(X^{\pi}_T)],  \end{equation} holds.
\end{problem} 
It is not a priori clear that the solution $\hat{\pi}$ (called $\textit{optimal investment strategy}$) of (2.4) exists. However, if it does, the problem of finding an optimal investment strategy is equivalent to finding a solution to a stochastic differential equation known as Hamilton-Jacobi-Bellman Equation (HJB). This is shown in detail in \cite{Bjork}. As the outline of its derivation can be found in Appendix A, we will simply state \newline
 \textbf{Hamilton-Jacobi-Bellman Equation (HJB)}\begin{equation} 0= \dfrac{\partial V}{\partial t}(t,x) + \sup\limits_{\pi}\left\lbrace[rx+(\mu-r)\pi x]\dfrac{\partial V}{\partial x}(t,x) + \sigma ^2 \pi^2 x^2 \dfrac{1}{2}\dfrac{\partial^2V}{\partial x^2}(t,x)\right\rbrace   
\end{equation} \begin{equation} \text{with boundary condition} \ \   V(T,x)= U(x). \end{equation} 
For simplification, the notation $\dfrac{\partial V}{\partial t}(t,x)= V_t, \dfrac{\partial V}{\partial x}(t,x)=V_x, \dfrac{\partial^2V}{\partial x^2} (t,x)= V_{xx}$ \newline is introduced.
\subsubsection{The Optimal Strategy for an Exponential Utility Function}
To find the  solution of the HJB, it is necessary to know the utility function as we need to set a boundary $V(T,x) = U(x)$. We will use the exponential utility, which is convenient as it simplifies many calculations and is defined by the function 
\begin{center} U: $\mathbb{R} \rightarrow (-\infty, 0]$, \ 
 $ x \longmapsto U(x) = -e^{-\alpha x}$, \ for a constant $\alpha$ > 0. \end{center} 
U describes how an investor evaluates the wealth $\textit{x}$, given an (individual) parameter $\alpha$ that we call $\textit{risk\ aversion}$. Note that its first derivative $U'(x)$ is converging to zero with increasing wealth, which means that the contribution to utility decreases (i.e. the greater the wealth is, the less an additional increase has an effect on the investor's utility).
This property is called $\textit{decreasing\ utility\ margin}$. Also, the risk aversion is constant, which implies that the investor's attitude to risk is independent of his wealth. The properties of (exponential) utility functions will be further discussed in the Appendix, so for now we will focus on the optimal strategy. \newline
 \newline  Since the supremum in (2.5) is identical to the maximum of the scalar function \newline $f(\pi) = $[$rx + (\mu -r) \pi x] V_x + \dfrac{1}{2}\sigma^2 \pi^2 x^2 V_{xx}$, the optimal value $\hat{\pi}$ needs to satisfy \newline
$0 = f'(\hat{\pi}) = (\mu -r)x V_x + \sigma ^2 \hat{\pi} x^2 V_{xx}$ $\forall$ t $\in$ [0,\textit{T}], hence \begin{equation} \hat{\pi} = - \dfrac{(\mu-r)}{\sigma^2x} \dfrac{V_x}{V_{xx}}. \end{equation} Note that it also needs to be checked that
$f''(\hat{\pi})$ =  $\sigma^2 x^2 V_{xx} < 0$. \newline So with (2.7), (2.5) can be written as \begin{equation}
V_t +rx V_x - \dfrac{1}{2} \dfrac{(\mu-r)^2}{\sigma ^2}\dfrac{(V_x)^2}{V_{xx}} = 0 .
\end{equation}\newline For simplification, the notation $\theta := \dfrac{\mu-r}{\sigma}$ (called the \textit{market price of risk}) is introduced. Based on existing results for similar problems (for example \cite{Tehranchi}) we suggest the following: \begin{proposition}
 The value function 
\begin{equation} V(t,x) = -e^{-\alpha x e^{r(T-t)}-\frac{\theta^2}{2}(T-t)} \text{ is a solution of the HJB.}\end{equation}
\end{proposition}
 \begin{proof}
 For simplification define $A(t,x):= \alpha x e^{r(T-t)} +\dfrac{\theta^2}{2}(T-t)$ and write $V(t,x) = -e^{A(x,t)}$. Then \newline
 $V_t = \dfrac{dA(x,t)}{dt} e^{-A(x,t)} = -[\dfrac{\theta^2}{2} + \alpha x r e^{r(T-t)}] e^{-A(x,t)}, \newline
V_x= \dfrac{dA(x,t)}{dx} e^{-A(x,t)} = \alpha e^{r(T-t)} e^{-A(x,t)} $ \text{and}$ \newline 
V_{xx} = -\alpha ^2 e^{2r(T-t)} e^{-A(x,t)}$ , hence $f''(\hat{\pi}_t) <0$.
\newline \text{Substituting $V_{xx}$ and $V_x$ from above into }  $\dfrac{1}{2} \theta^2 \dfrac{(V_x)^2}{V_{xx}} \text{gives} \newline \dfrac{1}{2} \theta^2  \dfrac{\alpha ^2 e^{2r(T-t)} e^{-2A(x,t)}  } {-\alpha ^2 e^{2r(T-t)} e^{-A(x,t)}} =- \dfrac{\theta^2}{2}e^{-A(x,t)}, \text{hence} \ \ V_t +rx V_x - \dfrac{1}{2} \theta^2 \dfrac{(V_x)^2}{V_{xx}} = 0$
\newline So, \textit{V(t,x)} satisfies equation (2.8) and since $V(T,x) = -e^{-\alpha x}$, it satisfies the HJB (2.5)-(2.6) for exponential utility.\end{proof}\begin{proposition} The optimal investment strategy is given by  \begin{equation}\hat{\pi}_t = \dfrac{\theta}{X^{\hat{\pi}}_t \alpha \sigma}e^{-r(T-t)}.\end{equation} \end{proposition}
\begin{proof}
From (2.9): $\hat{\pi}_t =-\dfrac{\theta}{x\sigma} \dfrac{V_x}{V_{xx}} = \dfrac{\theta}{x\alpha\sigma}e^{-r(T-t)}$. Hence Proposition 1 holds and by Theorem 19.6 in \cite{Bjork} (verification theorem), \textit{V} is the optimal value function and $\hat{\pi}$ is the corresponding optimal strategy. \end{proof}
Note that  $\hat{\pi}$  depends on time and wealth, so it is not constant. However, the absolute amount invested $\hat{\pi}_t X^{\hat{\pi}}_t = \dfrac{\theta}{\alpha \sigma}e^{-r(T-t)}$ does not depend on the absolute wealth of the investor.  
\newline  The optimal strategy gives us then
\begin{proposition} The optimal wealth process is given by \begin{equation}
X^{\hat{\pi}}_t = X^{\hat{\pi}}_0 e^{rt} + t \dfrac{\theta^2}{\alpha} e^{r(t-T)}+\dfrac{\theta}{\alpha}e^{r(t-T)}W_t \   with \  X^{\hat{\pi}}_0 = x.\end{equation}\end{proposition} \begin{proof}  Substituting (2.10) into (2.3) we get  \newline $dX^{\hat{\pi}}_t = [rX^{\hat{\pi}}_t + \dfrac{\theta^2}{\alpha} e^{-r(T-t)}] dt + \dfrac{\theta}{\alpha}e^{-r(T-t)} dW_t$ . \newline This is a linear Stochastic Differential Equation (SDE) and the derivation of its solution can be found in Appendix 6.2. \end{proof} In particular, for t=T it follows: $ X^{\hat{\pi}}_T = X_0 e^{rT} + T\dfrac{\theta^2}{\alpha} + \dfrac{\theta}{\alpha} W_T$ , so the terminal wealth is normally distributed with  $\mathbb{E}_0[X^{\hat{\pi}}_T]=X_0 e^{rT} + T\dfrac{\theta^2}{\alpha}$ and Var$(X^{\hat{\pi}}_T) = \dfrac{\theta^2}{\alpha^2}T$.

\subsection{Brief Analysis of the Optimal Strategy}
We will start by analysing the optimal strategy with respect to factors that can be influenced by the investor, namely the initial investment $X_0$, the risk aversion $\alpha$ and the investment horizon \textit{T}.  The market parameters r, $\mu$ and $\sigma$ will be considered as fixed values here, but their impact will be investigated in detail in the next section.\newline
Since the optimal strategy is inversely proportional to the investor's risk aversion parameter $\alpha$ and since the latter tends to take very low values, we will see that it is particularly sensitive to it. Also the initial wealth plays an important role, as low risk aversion can be compensated by high initial wealth in order to reach the same expected return on initial wealth.
 \subsubsection{Parameters:  $\alpha$ and \textit{T} }
If we look at the formula in (2.10), it can be seen that the investment strategy is inversely proportional and therefore highly sensitive to $X^{\hat{\pi}}_t$ and to $\alpha$. This can be interpreted in the way that the proportion invested in the stock is reduced with increasing wealth, which makes sense in the context of decreasing marginal utility. 
Also, since $\alpha$ is a measure for aversion to risk, it is intuitive that the proportion of wealth put at risk (i.e. invested in the stock) is decreasing with increasing $\alpha$.
\newline 
 In order to assess the impact of the  parameter $\alpha$, we will from now on consider the absolute amount invested at t, $\hat{\pi}_t X^{\hat{\pi}}_t$, which is independent of $X^{\hat{\pi}}_t$. Since it is deterministic, it is sufficient to look at the initial investment (t=0) to characterize the impact of $\alpha$ , as it is done in Figure 1 for a 'standard' setting (Note that T is rather small here, but this will not make much difference, as we will see later).  Clearly, for low risk aversion, even small changes in $\alpha$ can have a huge impact on the investment. It therefore needs to be accorded special attention to the choice of $\alpha$.

\begin{figure}[H]
\centering  \includegraphics[width=90mm]{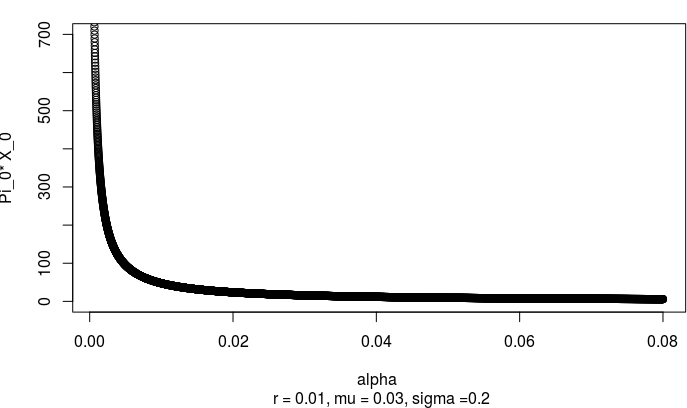}
\caption{Impact of risk aversion on initial investment for $\hat{\pi}$ (T=5)}
\end{figure}
What values for $\alpha$  are reasonable? Allthough the estimation of risk aversion depends on the experimental setting and methods used, studies indicate a similiar range: \cite{Guiso} find a median of 0.000708 and an average of 0.01978 for Italian households, \cite{Buccola} find a best estimate of 0.001 (with a smallest value estimation of 0.000708) for Californian tomato growers, and \cite{Paravisini} find an average absolute risk aversion of 0.037 and a median of 0.0439 and refer to \cite{Holt} for similiar values (0.003 average, 0.109 median) for individual investors on a lending platform.
Besides their impact on investment, there are other reasons susggesting to focus on rather smaller values of $\alpha$: Generally the setting is targeted at people that are somewhat willing to invest, so extremly high risk aversion could be excluded. Also, the distribution of risk aversion in a population seems to be right-skewed (see \cite{Bakshi}), hence the median 0.0007 might be a better reference value than the average. However, values around 0.01 and 0.001 seem also to be realistic options for $\alpha$ and should be considered as well. 
\newline \newline
Compared to $\alpha$, the impact of the (reasonable) time horizon T is rather small. For instance, in the second diagram  of Figure 2 the intersection with $\hat{\pi}_t X^{\hat{\pi}}_t$-axis is at $\alpha$ = 0.0004. It results in a difference of the initial investment in the stock of circa 250 for T between 10 and 30 years. This difference doesn't change as much as the initial investment changes with increasing $\alpha$. For example, if $\alpha$ = 0.0007, the difference of initial investment is 150 for the same range of T (whereas the initial investment changed from around 1'400 to around 500).
\begin{figure}[H]
 \includegraphics[width=82mm]{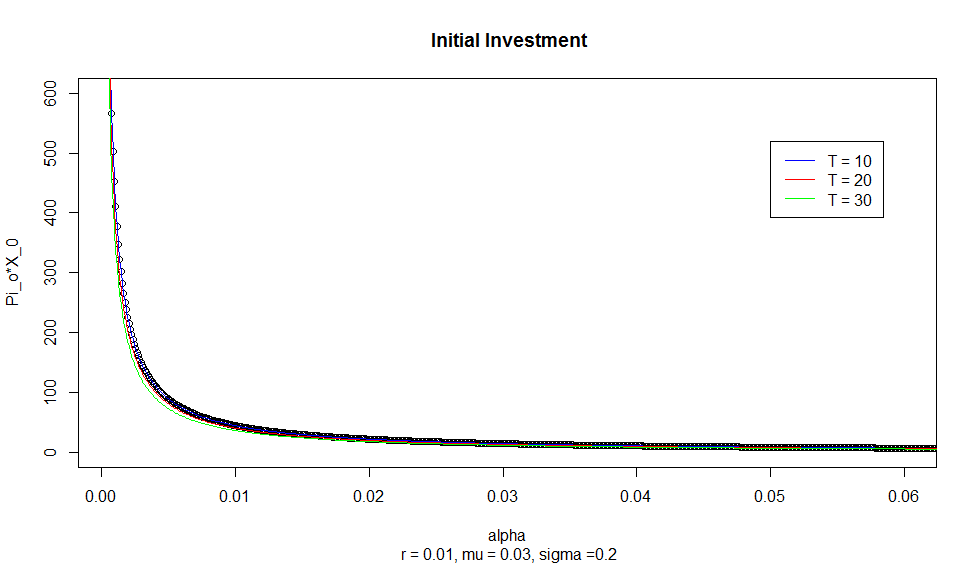}  \includegraphics[width=82mm]{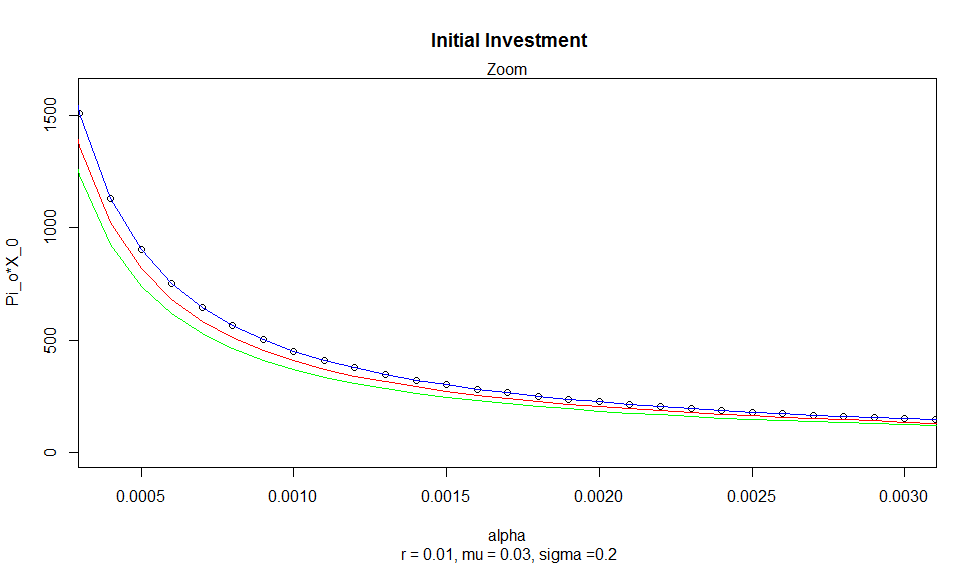} 
\caption{Impact of risk aversion on initial investment for $\hat{\pi}$ (varying T)}
\end{figure}

\subsubsection{Initial Wealth and Terminal Wealth Distribution}
The absolute amount invested increases exponentially with time at the risk-free interest rate r, independently of the performance of the stock (see Fig.3).This means, that if wealth increases at a higher rate than r, the proportion invested in stocks is decreasing, which in turn leads to a smaller variance of terminal wealth, relative to its expected value. On the other hand, for smaller wealth, we would expect a higher return, since a bigger proportion is invested in stocks, which have a higher expected return ($\mu$ > r). 
\begin{figure}[H]
 \centering \includegraphics[width=55mm]{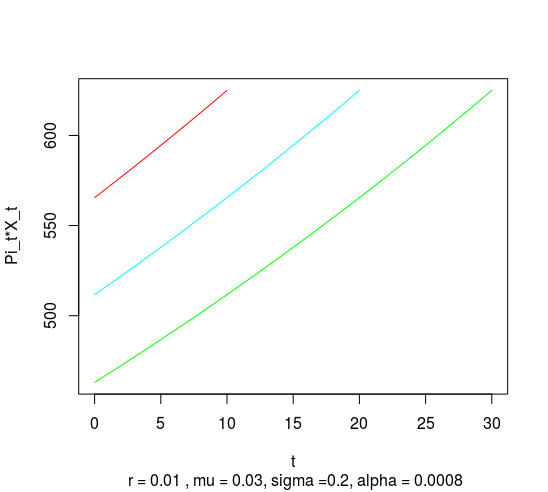} 
\caption{Absolute investment over time for $\hat{\pi}$ (varying T)}
\end{figure} 
 $\alpha$ affects the proportion of wealth invested to the same extent as $X_0$. This means, that for high initial wealth, but small risk aversion, an investor would put the same percentage of wealth at risk, as a highly risk averse person with a smaller amount of initial wealth. Consequently, the distributions of terminal wealth are the 'same' with respect to the return on initial wealth and variance relative to the expected return. In other words: Low wealth can be compensated by risk tolerance. This can be confirmed by the empirical values, as shown in the table below for the case \textit{r} = 0.01, $\mu$ = 0.03, $\sigma$ = 0.1, \textit{T} = 20. Here,\newline $\sigma(X_T)$ = $\sigma_{ML}(X_T)$/$\mu_{ML}(X_T)$, where $\mu_{ML}(X_T) = \mathbb{E}[X_t]$ and $\sigma_{ML}(X_T)$ are the absolute values for expected value and standard deviation of the terminal wealth distribution, obtained from a maximum-likelihood fit for a sample of 1'000. Note that, since $X^{\hat{\pi}}_T$ follows normal distribution, the 50$\%$-quantile is identical to the expected return. For more details on other quantiles refer to the Appendix A. 
\begin{table}[H]
\begin{minipage}{0.30\linewidth} \centering
\subcaption{$\alpha$ = 0.01}

\setlength{\tabcolsep}{1mm}
\renewcommand{\arraystretch}{1.2}
\begin{tabular}{|l|c|c|r|}
\hline 
$X_0$& $\dfrac{E[X_T]}{X_0}$&$\sigma$($X_T$)&$\alpha X_0$\\ \hline 
\rowcolor{LightMintGreen} 
10&915$\%$ &96$\%$&0.1\\ 
\rowcolor{Lavender} 
$10^2$&201$\%$&44$\%$&1\\
\rowcolor{Lavender}  
$10^3$&130$\%$&7$\%$&10\\
\rowcolor{Lavender}  
$10^4$&123$\%$&1$\%$&100\\ 
$10^5$&122$\%$&0$\%$&1000\\  \hline 
 \end{tabular} 
\end{minipage}
 \quad
\begin{minipage}{0.32\linewidth} \centering
\subcaption{$\alpha$ = 0.001}

\setlength{\tabcolsep}{1mm}
\renewcommand{\arraystretch}{1.2}
\begin{tabular}{|l|c|c|r|}
\hline 
$X_0$&$\dfrac{E[X_T]}{X_0}$&$\sigma$($X_T$)&$\alpha X_0$\\ \hline 
\rowcolor{LightMintGreen} 
10&8'684$\%$&104$\%$&0.01\\
\rowcolor{LightMintGreen} 
$10^2$&938$\%$&94$\%$&0.1\\ 
\rowcolor{Lavender} 
$10^3$&197$\%$&45$\%$&1\\
\rowcolor{Lavender} 
$10^4$&130$\%$&7$\%$&10\\
\rowcolor{Lavender} 
$10^5$&123$\%$&1$\%$&100\\ \hline 
 \end{tabular} 
 \end{minipage}
 \quad
\begin{minipage}{0.32\linewidth} \centering
\subcaption{$\alpha$ = 0.0001}

\setlength{\tabcolsep}{1mm}
\renewcommand{\arraystretch}{1.2}
\begin{tabular}{|l|c|c|r|}
\hline
$X_0$&$\dfrac{E[X_T]}{X_0}$&$\sigma$($X_T$)&$\alpha X_0$\\ \hline 
10&77'311$\%$ & 115$\%$&0.001\\
\rowcolor{LightMintGreen} 
$10^2$&8'056$\%$ & 110$\%$&0.01\\
\rowcolor{LightMintGreen} 
$10^3$&938$\%$&89$\%$&0.1\\ 
\rowcolor{Lavender} 
$10^4$&200$\%$&46$\%$&1\\
\rowcolor{Lavender} 
$10^5$&130$\%$&7$\%$&10\\ \hline
 \end{tabular}  
 \end{minipage}

 \caption{Empirical expected return and variance of terminal wealth for $\hat{\pi}$}
\end{table}

Note that the expected return at T is a mix of interest on riskfree bond and return on stock, depending on the proportion invested. Therefore, for bigger values of $X_0$ the return converges to the deterministic return on the riskfree bond ($e^{0.01 \times 20}$ = 1.22) , and its variance converges to zero.  On the other hand, for small initial wealth, the amount invested is much higher than the initial wealth itself (and generally higher for smaller $\alpha$), hence the expected return is a large multiple of the initial value. For example, for $\alpha$= 0.0001 and $X_0$ =10 the strategy requires to invest 16'375 in stocks and -16'365 in the riskfree asset, which in turn then leads to an expected theoretical absolute return of around 8'012. This means, that for small wealth, the strategy is connected to investing more money than provided by the investor (i.e. short-selling or borrowing), which one might want to avoid. 
\subsection{Extended Analysis of the Optimal Strategy under a Restriction on Investment}
In this section we introduce a restriction on the investment in the risky asset and investigate its effects on the terminal wealth distribution. Alltough interactions of the parameters are complex, a pattern is indicated: Generally, the restriction reduces smaller quantiles and has less effect on higher ones, since it changes the distribution towards a log-normal distribution.  The impact generally seems to be higher for high initial wealth, large $\mu$-\textit{r} and large $\sigma$. For small  \textit{r} and small $\alpha$, the advantages of the restriction on investment can be seen the best.  In addition to this empirical approach, a more theoretical one can be found in Appendix 6.4.\newline \newline
We start by defining the restriction to the values for $\hat{\pi}_tX^{\hat{\pi}}_t$ in order to avoid that the amount invested in the risky stock required by the optimal strategy exceeds the level of wealth:  
\begin{modification}\textbf{Restriction on Investment}\newline
Let the modified strategy $\hat{\pi}_m$ be defined for $(t, X^{\hat{\pi}_m}_t$) $\in$ $[0,T)$ $\times$ $\mathbb{R}$ by 
 \begin{center}
$\hat{\pi}_m (t,X^{\hat{\pi}_m}_t)$ = $\begin{cases}{\hat{\pi}_t }&\text{if $X^{\hat{\pi}_m}_t \geq \hat{\pi}_t X^{\hat{\pi}_m}_t$}\\{1}&\text{if $X^{\hat{\pi}_m}_t <  \hat{\pi}_t X^{\hat{\pi}_m}_t$}\end{cases} $ \end{center}
where $X^{\hat{\pi}_m}_t$ is the corresponding wealth process with $X^{\hat{\pi}_m}_0$ = $X_0$, and $\hat{\pi}_t$ is the optimal strategy from Proposition 2.    \end{modification}

\subsubsection{Initial wealth}
Let's first look at the impact of the value of initial wealth on the strategy. Intuitively, the higher this value, the less likely it is that wealth falls under the investment amount required by the optimal strategy during the process and hence the difference between modified and original strategy should be small (or could be zero). From another perspective, it also means that the proportion of wealth invested is small, so that the strategy (modified or not)  plays only a small part in the overall terminal wealth distribution and so will changes in parameters that affect the strategy. Qualitatively, this can be observed in the following plots, where the initial wealth is varied to be 120$\%$, 100$\%$ and 80$\%$ of the amount required by $\hat\pi_0$. 
 \newpage
\begin{landscape}
 \begin{figure}
 \begin{minipage}{0.16\textwidth} \subcaption{$X_0$=120$\% \hat{\pi}_0 X_0$} \end{minipage}
\begin{minipage}{0.31\linewidth}  
\includegraphics[width=70mm]{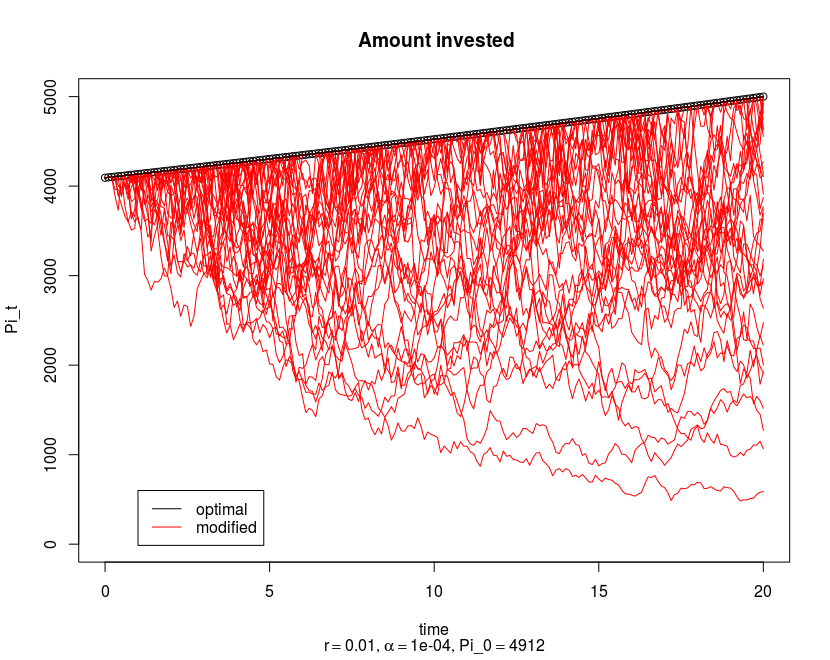} \end{minipage} \begin{minipage}{0.31\linewidth} \includegraphics[width=70mm]{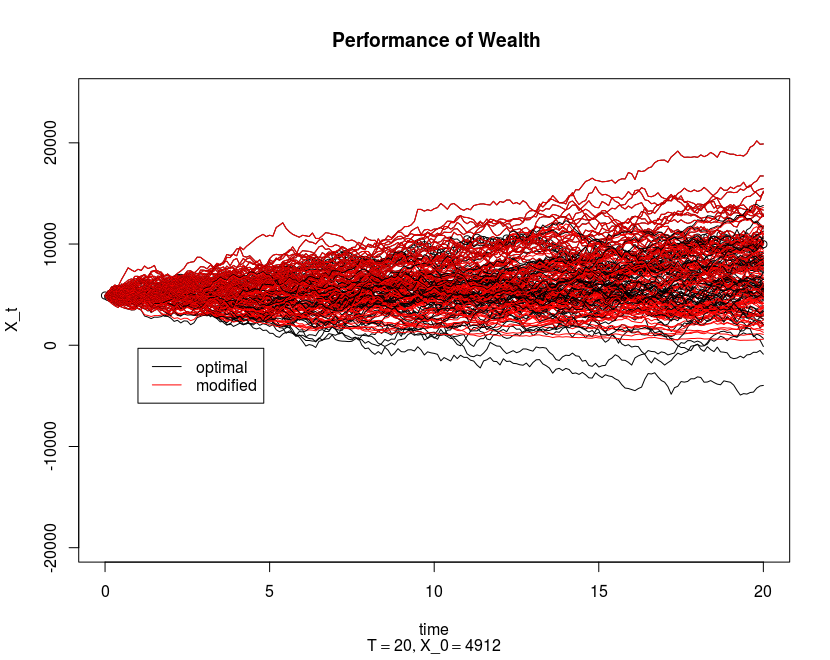}\end{minipage}  \begin{minipage}{0.28\linewidth}\includegraphics[width=70mm]{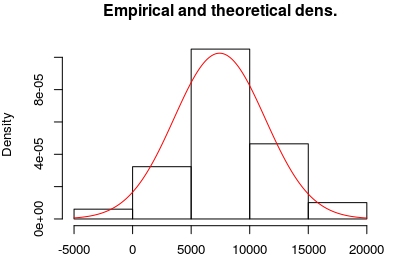} \end{minipage}  

\begin{minipage}{0.16\textwidth} \subcaption{$X_0$=100$\% \hat{\pi}_0 X_0$} \end{minipage}
\begin{minipage}{0.31\linewidth} 
\includegraphics[width=70mm]{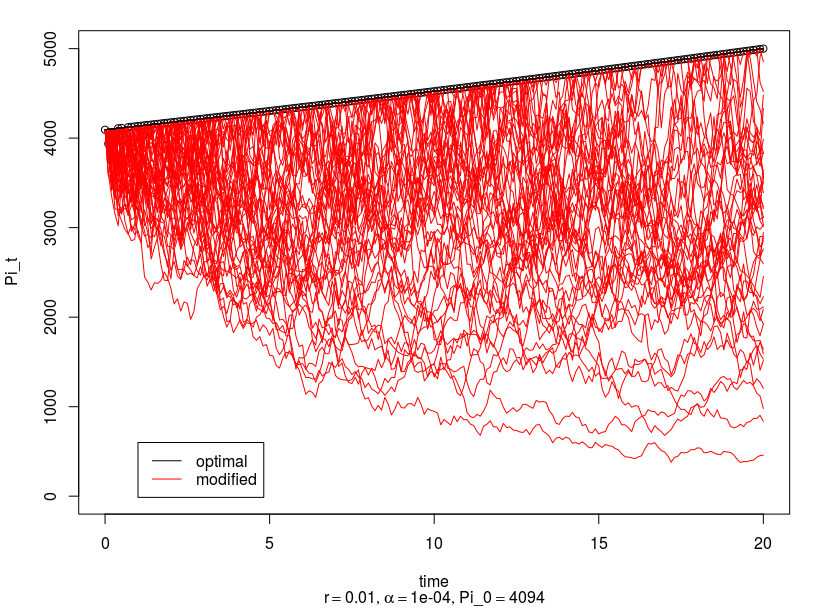} \end{minipage} \begin{minipage}{0.31\linewidth} \includegraphics[width=70mm]{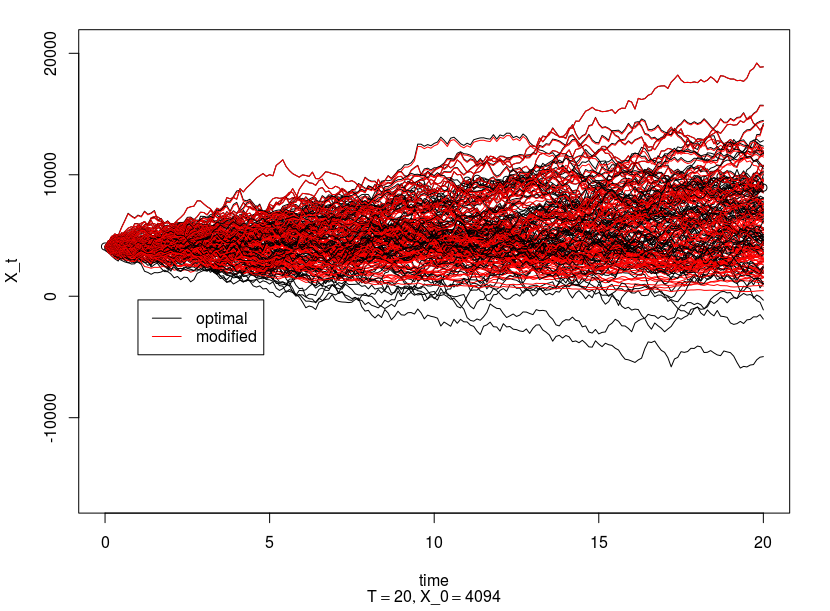}\end{minipage}  \begin{minipage}{0.27\linewidth}\includegraphics[width=60mm]{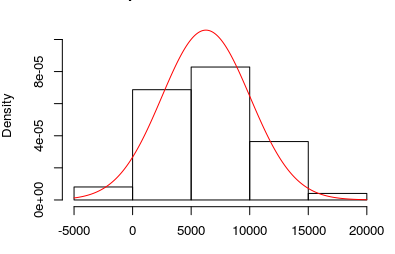} \end{minipage}
 
 \begin{minipage}{0.16\textwidth} \subcaption{$X_0$=80$\% \hat{\pi}_0 X_0$} \end{minipage}
 \begin{minipage}{0.31\linewidth} 
\includegraphics[width=70mm]{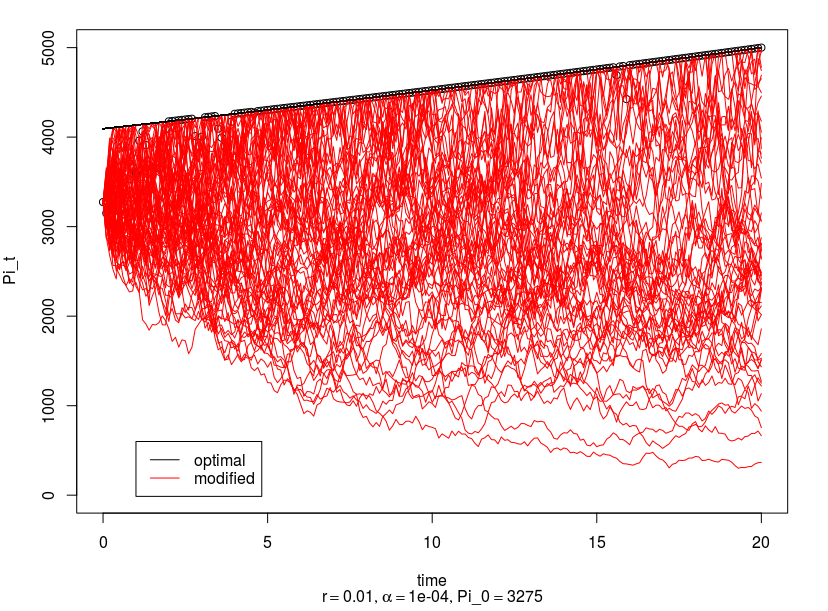} \end{minipage} \begin{minipage}{0.31\linewidth} \includegraphics[width=70mm]{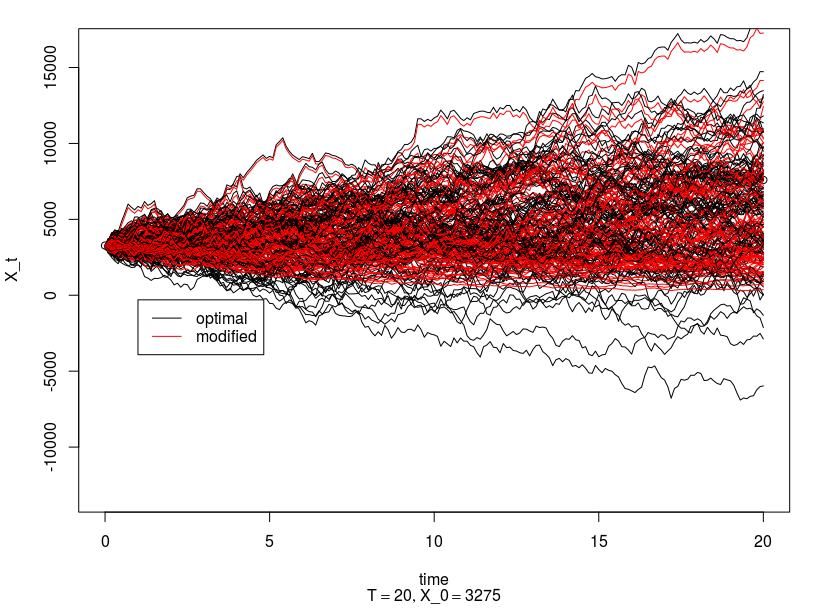}\end{minipage}  \begin{minipage}{0.27\linewidth}\includegraphics[width=60mm]{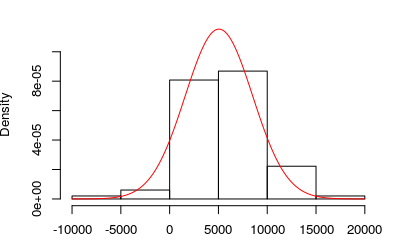} \end{minipage}
 
\caption{Investment and wealth processes for $\hat{\pi}$ and $\hat{\pi}_m$ (varying $X_0$) and terminal wealth distribution for $\hat{\pi}_m$ }
 \end{figure}
\end{landscape}

One can observe a shift from higher to lower values of terminal wealth for decreasing $X_0$ and increasing effect of the investment-restriction (the red line in the histograms is a normal distribution fit for reference). This is also reflected in the quantiles of the empirical terminal wealth distribution, where $\mathcal{Q}_{p}(X_0)$ is the p-quantile resulting from initial wealth $X_0$,$\hat{\pi}_0^{-1}$ = $X_0$/($\hat{\pi}_0X_0 $) is the initial wealth as percentage of the required initial investment and $\Delta (\mathcal{Q}_{p})  = \mathcal{Q}_{p}(X_0) / \mathcal{Q}_{p}(\hat{\pi}_0X_0 )$ is the change of quantiles relative to $\mathcal{Q}_{p}(\hat{\pi}_0X_0 )$= $\mathcal{Q}_{p}(100$\%$\hat{\pi}_0X_0)$ :
\begin{table}[H]
\setlength{\tabcolsep}{0.5mm}
\centering

\begin{tabular}{|l|c|c|c|c||c|c|c|c|}
\hline
$X_0$ &  \cellcolor{OldLace} $\mathcal{Q}_{0.25}$($X_0$) & \cellcolor{OldLace}  $\mathcal{Q}_{0.50}$($X_0$) & \cellcolor{OldLace}  $\mathcal{Q}_{0.75}$($X_0$) &   \cellcolor{OldLace}$\mathcal{Q}_{0.95}$($X_0$) &\cellcolor{LightCyan}$\Delta (\mathcal{Q}_{0.25})$  &\cellcolor{LightCyan} $\Delta (\mathcal{Q}_{0.5})$  & \cellcolor{LightCyan} $\Delta (\mathcal{Q}_{0.75})$ & \cellcolor{LightCyan} $\Delta (\mathcal{Q}_{0.95}$) \\\hline 
4'912  & 4'037  &7'549 & 10'848 & 15'043 & 134 $\%$&117$\%$&112$\%$&107$\%$\\
 
4'094& 3'010& 6'450 & 9'653  &14'064 &  &  & &\\
 
 3'275 & 2'292  & 4'660 & 8'067 & 12'673 & 76$\%$ & 72$\%$ &84$\%$&90$\%$\\ \hline

 \end{tabular}  \caption{Quantiles of $X^{\hat{\pi}_m}_T$ and comparison for $\hat{\pi}_0^{-1}$= 120$\%$, 100$\%$ and 80$\%$} 
 \end{table}
 It is remarkable that the change of $X_0$ affects lower quantiles more than higher quantiles: While an increase of $X_0$ by 20$\%$  leads to an increase greater than 20$\%$ of the 25$\%$-quantile, the 95$\%$-quantile only increases by 7$\%$. The same happens when $X_0$ is reduced by 20$\%$. It can also be observed an asymmetry between increase and decrease of $X_0$ in terms of its effect on quantiles. For instance, at $\mathcal{Q}_{0.5}$, we observe a decrease by 28$\%$ versus an increase of only 17$\%$.
Similar effects can also be observed for a lower expected return $\mu$ and higher volatility $\sigma$. In particular, a higher volatility seems to lead to stronger changes, for example with $\sigma$ = 40$\%$, we reach an increase of $\mathcal{Q}_{0.25}$ to 140$\%$ and a decrease to 72$\%$.
\newline \newline
The reason for the shift of the quantiles lies in the missed upside potential of the paths of wealth that fall under the amount required to follow the optimal strategy. Instead of investing the full amount required by the optimal strategy, only the wealth available can be invested, this results generally in lower values for the terminal wealth. The lower the initial wealth is, the more paths are limited by the modification of the strategy and lead to smaller wealth, which results in lower quantiles. As extremely well performing paths are more rare and less likely to be constrained by the modification, the change of $X_0$ affects them less, and so the upper quantiles (i.e. $\mathcal{Q}_{0.95}$) are less affected.\newline Another reason could be the fact, that the proportion invested is smaller for larger $X_0$ and therefore the modification of the strategy has less impact on the distribution of higher values of terminal wealth. In that case, even if the part invested in the stock performs less well due to the modification, this will have limited effect, since the outcome of the strategy is dominated by the high proportion invested in riskfree bond.\newline 
The asymmetry in the changes of quantiles might be explained by the following interpretation: When wealth falls and stays under the amount required by the optimal strategy, it is 100$\%$ invested in stocks, hence its distribution is identical to the one of the stocks, which is a log-normal distribution. So, for lower $X_0$  the distribution is closer to log-normal, whereas for increasing $X_0$ it converges to the normal distribution of the optimal terminal wealth.
\subsubsection{Parameter: $\mu$ -\textit{r}}
In contrary to the original strategy, where the amount invested was independent of the performance of the stock, the restriction on investment establishes a link to the stock market. Therefore market-parameters as expected return $\mu$, risk free interest rate \textit{r} and volatility $\sigma$ will influence the resulting terminal wealth distribution to another extent. Let us first look at the difference of the market rate and the riskfree rate.
Since the initial investment amount required by the optimal strategy changes with variation of $\mu$ - \textit{r}, we set $X_0 = \hat{\pi}_0 X_0$  and \textit{r} = 0.01 and compare the resulting output for different $\mu$ and fixed $\alpha$ = 0.001.
\begin{figure} [H]
 \includegraphics[width=74mm]{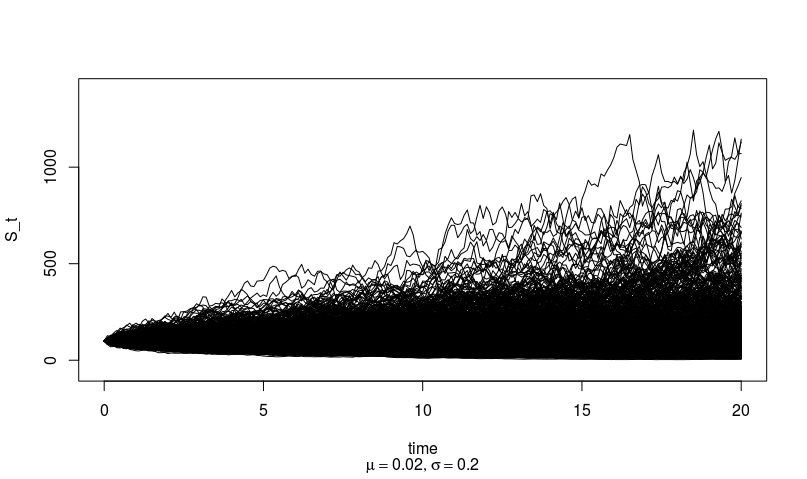} \includegraphics[width=74mm]{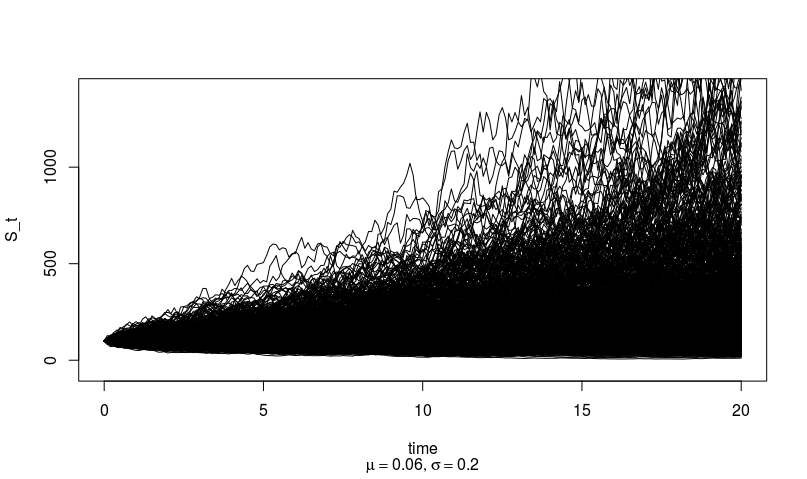} 
 \caption{Performance of stock ($\mu$ -\textit{r} = 0.01 and $\mu$ - \textit{r} = 0.05)}
\end{figure} 
\begin{figure} [H] 
 \begin{minipage}{0.50\linewidth} \includegraphics[width=70mm]{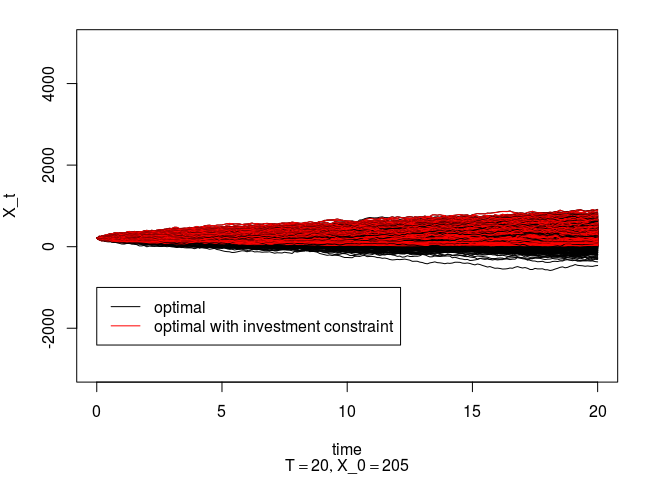}\end{minipage}  
     \begin{minipage}{0.50\linewidth} \includegraphics[width=70mm]{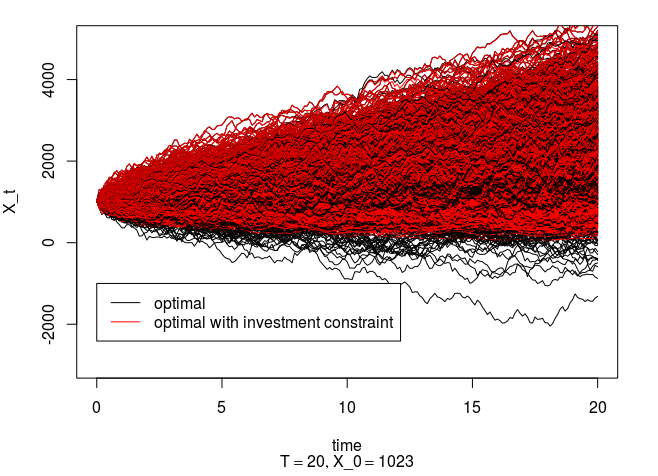}\end{minipage}  
\caption{Wealth process for $\hat{\pi}$ and $\hat{\pi}_m$ ($\mu$-\textit{r} = 0.01 and $\mu$-\textit{r} = 0.05 )}
\end{figure}
Since the absolute initial wealth varies for the different scenarios, we look at the quantiles of the total return instead of the absolute terminal wealth to make a better comparison.\newline  In the diagram below, the x-axis shows the total return in 100$\%$ and the red line represents a normal distribution fit.
\begin{figure}[H]
\includegraphics[width = 70mm]{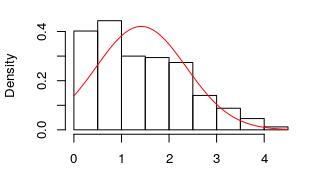}\includegraphics[width = 70mm]{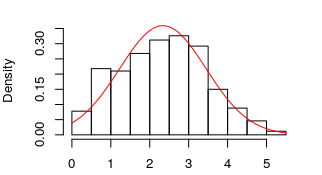}
\caption{Histogram of total return on initial wealth for $\hat{\pi}_m$ ($\mu$-\textit{r} = 0.01 and $\mu$-\textit{r}= 0.05)}
\end{figure}  Again, we see an increasing weight of lower quantiles with decreasing gap, indicating a change in the type of the distribution rather than a simple shift of values. For the empirical quantiles, let $\mathcal{Q}_{p}$($\mu$-\textit{r}) be the p-quantile of the return on initial investment at \textit{T}= 20 for a gap $\mu$-\textit{r} and $\Delta(\mathcal{Q}_{p})$=$\mathcal{Q}_{p}$($\mu$)/$\mathcal{Q}_{p}$(0.03) be the change in quantiles in comparison to a gap $\mu$-\textit{r} of 3$\%$.
\begin{table}[H]
\setlength{\tabcolsep}{0.5mm}
\centering
\begin{tabular}{|l| c |c|c|c||c|c |c|r| }
\hline
$\mu$-r &  \cellcolor{OldLace} $\mathcal{Q}_{0.25}$($X_0$) & \cellcolor{OldLace}  $\mathcal{Q}_{0.50}$($X_0$) & \cellcolor{OldLace}  $\mathcal{Q}_{0.75}$($X_0$) &   \cellcolor{OldLace}$\mathcal{Q}_{0.95}$($X_0$) &\cellcolor{LightCyan}$\Delta (\mathcal{Q}_{0.25})$  &\cellcolor{LightCyan} $\Delta (\mathcal{Q}_{0.5})$  & \cellcolor{LightCyan} $\Delta (\mathcal{Q}_{0.75})$  & \cellcolor{LightCyan} $\Delta (\mathcal{Q}_{0.95}$)  \\ \hline 
5$\%$   & 149$\%$  & 239$\%$ & 314$\%$ & 419$\%$ & 240 $\%$  & 213$\%$  & 195$\%$  & 182$\%$  \\
 
3$\%$ & 103$\%$& 187$\%$ & 268$\%$   & 383$\%$ &  &  &  &  \\
 
 1 $\%$ & 58$\%$  & 122$\%$ & 211 $\%$& 322$\%$& 19$\%$  & 22$\%$ & 26$\%$ & 28$\%$ \\ \hline

 \end{tabular}  
\caption{Quantiles of  $X^{\hat{\pi}_m}_T$ as total return and comparison for $\mu$-r = 5$\%$, 3$\%$ and 1$\%$}
\end{table}
As before, we see the effect of change of parameters more pronounced in lower quantiles, for example with $\mu$= 5$\%$, $\mathcal{Q}_{0.25}$ increases by +140$\%$ while $\mathcal{Q}_{0.95}$ only increases by +82$\%$. This pattern can be observed for different combinations of the other parameters as well, allthough the effect on quantiles is generally stronger for smaller initial values and large variance. At the same time, this is also where the difference to the results for the optimal unconstrained strategy are most obvious (for instance, at $\mu -\textit{r}$ = 5$\%$, $\hat{\pi}$ leads to an increase by +116$\%$ of $\mathcal{Q}_{0.25}$ and by +82$\%$ of $\mathcal{Q}_{0.95}$, the latter being the same as for $\hat{\pi}_m$.)\newline
Then again, an asymmetry between the effects of reduction and increase of $\mu$-\textit{r} can be observed. For example, if the gap is reduced to 1/3, the respective quantiles are reduced to values between 20$\%$ and 30$\%$, whereas an increase by 2/3 leads to in increase of lower quantiles to more than 200$\%$ (the proportional increase would lead to 166$\%$). For other paramters $\sigma$ and $\alpha$ the behaviour is similar. 
\newline \newline
The reason for this asymmetry lies in the application of the investment-constraint.\newline Generally, for the basic secenario $\mu$-\textit{r} = 0.03, the majority of the wealth paths seems to be affected by the constraint (see Fig. 7) and so the terminal return's distribution is close to log-normal. Since a reduction of $\mu$-\textit{r} does not have much effect on the application of the constraint and the terminal distribution, the effect on the quantiles is rather low. In contrary, an increase of $\mu$-\textit{r} might lead to a few more paths performing better and not being affected by the investment-constraint and so the terminal wealth distribution will be slightly closer to the normal distribution of the optimal terminal wealth. This change of the distribution can be observed in the quantiles as an (overproportional) increase of the lower quantiles and in the histograms as a shift of values to the left for low $\mu$. Another way of interpretation, as before, is the following: the more paths of wealth follow the constrained part of the strategy, the less upside and downside potential is used, which leads to more concentration of lower values and ultimately to a right-skewed log-normal-distribution. Also note, that if $\hat{\pi}_t X_t$ grows at the lower rate \textit{r} < $\mu$, the absolute gap between stocks (and thus wealth) and $\hat{\pi}_t X_t$ grows with time  -leading to a falling probability to switch to the modified strategy.  We therefore observe most of the differences between original and modified strategy in the first years of investment.\newline  Another effect, as seen before, could be the linked to the fact, that high quantiles indicate high wealth, which in turn are connected to a low proportion invested and so they are less affected by the restriction. \newline 
Note, that the amount required to be invested in the risky asset by the optimal strategy is proportional to $\mu$-\textit{r}. Changes here have an equivalent effect on the modified strategy as `inverse` changes of the initial wealth, which we considered in the section above (in other words: whether or not the modification of the strategy has an impact depends on how high $X_0$ is in comparison to $\hat{\pi}_0 X_0$). However, this can not be observed in the plots and quantiles as we had set $X_0= \hat{\pi}_0 X_0$. 
\subsubsection{Parameter: $\sigma$}
For now, we have seen that the stronger the impact of the constraint on investment is, the more the terminal wealth distribution changes towards a log-normal distribution. We will assume that the same effect will happen, when the impact of the constraint changes because of different parameters $\sigma$ or \textit{r}. This is why, instead of \textit{how}, we are now interested in assessing \textit{how much} the effect of this restriction can be. We will therefore compare the empirical quantiles of the terminal wealth distribution from this modified strategy with those from the original one. Since the values of the investment-constrained strategy might be of more interest, they are listed here, but those from the optimal unconstrained strategy can be found in Appendix 6.6.  \newline
We look at the quantiles of a small and high volatility for different expected return-interest gaps, having set $\alpha$ = 0.001, \textit{r} = 0.01 and $X_0 = \hat{\pi}_0 X_0$.
\newline

\begin{table}[H]
\setlength{\tabcolsep}{0.5mm}
\begin{center}\begin{tabular}{|l| c |c|c|c| }
\hline

$\mu$-r & $\mathcal{Q}_{0.25}$($X_0$) &   $\mathcal{Q}_{0.50}$($X_0$) &  $\mathcal{Q}_{0.75}$($X_0$) &  $\mathcal{Q}_{0.95}$($X_0$)   \\ \hline 
5$\%$   & 28$\%$  & 172$\%$ & 354$\%$ &585$\%$  \\
 
3$\%$&17$\%$&96$\%$&296$\%$&532$\%$  \\
 
 1 $\%$&11$\%$&52$\%$&231$\%$&481$\%$ \\ \hline \end{tabular}  \subcaption{  $\sigma$ = 0.4 }  \end{center}

\begin{center}   
\begin{tabular}{|l|c |c|c|c| }
\hline

$\mu$-r & $\mathcal{Q}_{0.25}$($X_0$) &   $\mathcal{Q}_{0.50}$($X_0$) &  $\mathcal{Q}_{0.75}$($X_0$) &  $\mathcal{Q}_{0.95}$($X_0$)   \\ \hline 
5$\%$&209$\%$&245$\%$&284$\%$&340$\%$  \\
 
3$\%$&155$\%$&195$\%$&231$\%$&284$\%$  \\
 
 1 $\%$&103$\%$&144$\%$&181$\%$&234$\%$ \\ \hline

 \end{tabular} \subcaption{ $\sigma$ = 0.1}  \end{center}
 \caption{Quantiles of $X^{\hat{\pi}_m}_T$ as total return for $\sigma$ = 0.1 and 0.4 (and varying $\mu$-r) } \end{table} 
 Unsurprisingly, the spread of the terminal distribution is bigger for a high variance and more narrow for a small one. More striking seem to be the rather disappointing results for ${\sigma = 0.4}$, where $\hat{\pi}_m$ yields worse results than the unconstrained strategy for all the returns shown in the table: the peak being at $\mathcal{Q}_{0.25}$ and $\mu$ -\textit{r} = 0.05 with 28$\%$ instead of 94$\%$ and the only exception being for a particularly 'bad' performance (i.e. having 11$\%$  instead of -3$\%$ at $\mathcal{Q}_{0.25}$ for $\mu$-\textit{r} = 0.01). In contrary to scenarios with high volatility that seem to be very much affected by the investment-constraint, a rather low volatility of 10$\%$ yields almost the same results as the optimal strategy. Of course, in that case, the range is generally much smaller and the results are lower. \newline
Overall, this could be good news anyway, because we would expect small $\mu$-\textit{r} to be connected to smaller risk, for which $\sigma$ is an indicator, and a higher gap to be linked to a higher volatility. So, $\mu$-\textit{r} = 0.01 and $\sigma$ = 0.1 resp.  $\mu$-\textit{r} = 0.05 and $\sigma$ = 0.4 might be realistic situations and $\mathcal{Q}_{0.5}$ = 144$\%$ resp. $\mathcal{Q}_{0.5}$= 172$\%$ are acceptable results. 
 \newline Also note, that the variance does have a rather big impact on the initial investment amount as well, for $\mu$-\textit{r} = 5$\%$ we get $\hat{\pi}_0 X_0$ = 256 when $\sigma$ = 0.1, compared to $\hat{\pi}_0 X_0$ = 4'094 for ${\sigma = 0.4.}$  
\subsubsection{Parameter: \textit{r}}
Let us consider here the cases of zero and negative interest rates, a small volatility $\sigma$ = 0.1, \textit{T}=20 and $\alpha$ = 0.001. As before, we set $X_0 = \hat{\pi}_0 X_0$. For comparison, the median of terminal wealth of a 100 $\%$ investment in stocks and of a 100 $\%$ investment in the bond are also listed. Note that the stock's distribution is right-skewed and the median is smaller than the mean.\newline

\begin{table}[H]
\setlength{\tabcolsep}{0.6mm}
\begin{center} \begin{tabular}{|l|c||c|c|c|c||c|c|}
\hline
$\mu$ &$X_0$&  \cellcolor{OldLace} $\mathcal{Q}_{0.25}$ & \cellcolor{OldLace}  $\mathcal{Q}_{0.50}$ & \cellcolor{OldLace} $\mathcal{Q}_{0.75}$ &   \cellcolor{OldLace}$\mathcal{Q}_{0.95}$ & $\mathcal{Q}_{0.50}(\pi_t \equiv 1)$&$\mathcal{Q}_{0.50}(\pi_t \equiv 0)$\\ \hline 
10$\%$ &10'000& 270$\%$ & 302$\%$& 331$\%$ & 374 $\%$  & 739$\%$ & 100$\%$ \\
 
6$\%$ &6'000& 189$\%$& 221$\%$ & 250$\%$ & 293$\%$ & 332$\%$  & \\
 
 1$\%$ &1000& 84$\%$ & 118$\%$ & 148 $\%$& 191$\%$& 122$\%$  & \\ \hline

 \end{tabular}  \subcaption{ r = 0$\%$} \end{center} 
 
\setlength{\tabcolsep}{0.8mm}   
\begin{center} \begin{tabular}{|l|c||c|c|c|c||c|c|}
\hline
$\mu$ &$X_0$&  \cellcolor{OldLace} $\mathcal{Q}_{0.25}$& \cellcolor{OldLace}  $\mathcal{Q}_{0.50}$ & \cellcolor{OldLace}  $\mathcal{Q}_{0.75}$ &   \cellcolor{OldLace}$\mathcal{Q}_{0.95}$ & $\mathcal{Q}_{0.50}(\pi_t\equiv 1)$&$\mathcal{Q}_{0.50}$($\pi_t\equiv0)$\\ \hline 
 4$\%$ &  6'107 & 138$\%$ & 164$\%$& 188$\%$ & 223 $\%$  & 223$\%$ & 82$\%$ \\
 
2$\%$ &  3'664 & 104$\%$& 131 $\%$ & 155$\%$ &190$\%$ & 149$\%$  & \\
 
 0$\%$ &  1'221 & 69$\%$ & 97$\%$ & 122 $\%$& 157$\%$& 100$\%$  & \\ \hline

 \end{tabular}  \subcaption{ r = -1$\%$} \end{center}  \caption{Quantiles of $X^{\hat{\pi}_m}_T$ as total return for r= 0$\%$ and r= -1$\%$ (and varying $\mu$)} \end{table}
Naturally, the return on initial wealth is lower for lower risk-free rates. However, in most of the cases it is still better than a 100$\%$ investment in bonds. At the same time, the more $\mu$ exceeds \textit{r}, the more the missed upside potential is expressed in the return on terminal wealth, compared to a 100$\%$  investment in stocks. In the example above we set $X_0 = \hat{\pi}_0 X_0$, but if we fix the initial wealth and consider a less risk averse person, the initial amount to be invested would be larger and therefore the proportion invested in the stock smaller. The terminal wealth in terms of return will therefore be closer to the 100$\%$ stock investment, than to the return on a 100$\%$ bond investment. For example, for $\alpha = 0.0001$, we get \newline $\mathcal{Q}_{0.50}$ = 675$\%$ if \textit{r}= 0 and $\mu$ = 10$\%$  and $\mathcal{Q}_{0.50}$ = 203$\%$ if $\mu$ = 4 $\%$ and \textit{r} =-1$\%$. \newline This means that the quantiles above also reflect the risk aversion preferences of the investor and a comparison to alternatives (like investing entirely in bonds or stocks) should take this into consideration.
\newline
Comparing these results with the unrestricted optimal strategy shows almost no difference for $\alpha$ =0.001: Here, the maximum difference is +6$\%$ at $\mathcal{Q}_{0.95}$ = 90$\%$ for $\hat{\pi}$ and $\mu$ = 0.01, \textit{r} =0. But, for $\alpha$ =0.0001, which implies much more investment in the risky asset, we find a great difference: For almost all the quantiles, the missed upside potential is huge. For example, for \textit{r}=0$\%$ it lies between 200$\%$ and 1'510$\%$ for all $\mu$ considered. The interesting exception to this are the lower quantiles with $\mu$=0.01. In these cases, 25$\%$ of optimal terminal wealth paths resulting from the unrestricted strategy take values below -2$\%$. This is where the restriction kicks in and shows a great advantage: For $\mu$ = 1$\%$ and \textit{r}=0$\%$, $\mathcal{Q}_{0.25}$ is at 82$\%$, and for $\mu$=0$\%$, \textit{r}= -1$\%$ it still is at 67$\%$ (see Table 14 in Appendix 6.5).

\section{An Optimal Strategy for Exponential Utility and Lower Constraint $K_l$}
\subsection{Derivation of the $K_l$-Strategy }
We now introduce a lower constraint to limit the terminal wealth and derive the optimal strategy for this setting. We will find, that it involves investing one part of the initial wealth following the optimal strategy (yielding a \textit{shadow wealth process})   and using the other part to buy a put option to hedge this process. For simplicity, we will sometimes refer to it as $K_l$-Strategy.\newline
We begin by modifiying Problem 1 by a constraint $K_l$ $\in$ $\mathbb{R}_{(-\infty,X_0 e^{rT})}$.
\begin{problem} 
Find an optimal strategy $\hat{\pi}_l \in \mathcal{A}$ such that \begin{equation} \mathbb{E}[U(X^{\hat{\pi}_l}_T)] = \sup\limits_{\pi\in \mathcal{A}}\mathbb{E}[U(X^{\pi}_T)]\ \text{and} \ X^{\hat{\pi}_l}_T \ge K_l \text{ holds a.s.}\end{equation} \end{problem}
In order to solve this problem, we first determine the optimal terminal wealth. 
\begin{proposition} For Problem 2, the optimal terminal wealth is of the form \begin{equation}X^{\hat{\pi}_l}_T = \tilde{X}^{\hat{\pi}}_T + \max\{K_l- \tilde{X}^{\hat{\pi}}_T ,0\},\end{equation} where $\tilde{X}^{\hat{\pi}}_t$ is the optimal wealth process from (2.15), with \newline  $\tilde{X}_0 = (-$ln$(\dfrac{y}{\alpha}) + rT-\dfrac{\theta^2}{2}T)\dfrac{1}{\alpha} e^{-rT}$ (called \textit{shadow value}) for y > 0 such that  $\mathbb{E}[H_T X^{\hat{\pi_l}}_T] = X_0$ and $\hat{\pi}_t$ the corresponding optimal strategy . \end{proposition}
\begin{proof} The statement follows directly from  Lemma 2 in \cite{Zhou}. \newline To see this, let $I(y) = U'^{-1}(y) =  -\dfrac{1}{\alpha}\ln(\dfrac{y}{\alpha})$ be the inverse of $U'(x) = \alpha e^{-\alpha x}$.\newline  Note that U is strictly increasing and concave. \newline 
By Lemma 2, $X^{\hat{\pi}_l}_T = \max\{K,I(yH_T)\} = I(yH_T) + \max\{K-I(yH_T),0\}$  for y >0 such that  $\mathbb{E}[H_T X^{\hat{\pi_l}}_T] = X_0$ and where $H_T = H(T)$ is the state price density at T.\newline  Determine $\tilde{X}_0$ such that $I(yH_T)= \tilde{X}^{\hat{\pi}}_T$.  This is the case if \newline $I(yH_T)=\dfrac{-1}{\alpha} [\ln(\dfrac{y}{\alpha}) -(r+\dfrac{\theta^2}{2})T-\theta W_T)]= \tilde{X}_0 e^{rT} + T\dfrac{\theta^2}{\alpha} + \dfrac{\theta}{\alpha} W_T$
\newline Hence we find  $\tilde{X}_0 = ( rT-\dfrac{\theta^2}{2}T-\ln(\dfrac{y}{\alpha}))\dfrac{1}{\alpha} e^{-rT}$.  \end{proof} Note that the optimal terminal wealth has the structure of the optimal terminal value for an unconstrained wealth process plus a put option (on this unconstrained wealth process). We can thus find the optimal strategy by determining a replicating portfolio that yields $\max\{K_l-\tilde{X}^{\hat{\pi}}_T ,0\}$ at \textit{T}. This will be done via risk-neutral valuation arguments, in analogy to \cite{Donnelly}.
\begin{proposition}The price at time t $\in$ [0,T] of a put option with payoff $\max\{K_l- \tilde{X}^{\hat{\pi}}_T ,0\}$ is given by \begin{equation}{p(t, \tilde{X}^{\hat{\pi}}_t)= \Phi (d_l(t,\tilde{X}^{\hat{\pi}}_t)) (e^{-r(T-t)} K_l -  \tilde{X}^{\hat{\pi}}_{t}) + \dfrac{\theta \sqrt{T-t}}{\alpha}  e^{-r(T-t)} \phi(d_l(t,\tilde{X}^{\hat{\pi}}_t))}\end{equation} \newline where $\Phi (x)$ is the cumulative normal distribution,  $\phi (x)$ its density and \newline $d_l(t, \tilde{X}^{\hat{\pi}}_t) =( K_l- \tilde{X}^{\hat{\pi}}_{t}e^{r(T-t)}) \dfrac{\alpha}{\sqrt{T-t} \theta}$. \end{proposition}
\begin{proof} Assume that the market is free of arbitrage and complete. Then there exists a risk-neutral measure $\mathbb{Q}$ such that $W^{\mathbb{Q}}_t := W_t + \theta t$ is a standard Brownian motion. Hence, under $\mathbb{Q}$ the discounted wealth process is a martingale: \newline $E_{\mathbb{Q}}[e^{-rt} \tilde{X}^{\pi}_t | \mathcal{F}_{t-1}] =  E_{\mathbb{Q}}[e^{-rt} (e^{rt} \tilde{X}^{\pi}_0 + e^{r(t-T)} \dfrac{\theta}{\alpha} W_t^{\mathbb{Q}})| \mathcal{F}_{t-1}]= E_{\mathbb{Q}}[ \tilde{X}^{\pi}_0 + e^{-rT} \dfrac{\theta}{\alpha} W_t^{\mathbb{Q}}| \mathcal{F}_{t-1}] \newline= \tilde{X}^{\pi}_0 + e^{-rT} \dfrac{\theta}{\alpha} W^{\mathbb{Q}}_{t-1} = e^{-r(t-1)} \tilde{X}^{\pi}_{t-1}
\ \forall \pi \in \mathcal{A},\  \forall t \in [1, T].$ \newline So we can evaluate the put option by risk-neutral pricing. Before this is done, note that: \newline $\tilde{X}^{\hat{\pi}}_{T} < K_l  \newline \iff 
\tilde{X}_{0}e^{rT} + \dfrac{\theta}{\alpha} W_T^{\mathbb{Q}} < K_l \newline  \iff W_T^{\mathbb{Q}}  < (K_l-  \tilde{X}_{0}e^{rT}) \dfrac{\alpha}{\theta}\newline \iff Z < (K_l-  \tilde{X}_{0}e^{rT}) \dfrac{\alpha}{\sqrt{T} \theta} =:d_0 $ \newline
for the substitution $W_T^{\mathbb{Q}} := Z \sqrt{T}$ and Z $\sim \mathcal{N}(0,1)$. 
 \newline Then the price of the put option at time 0 is given by: \newline $p(0, \tilde{X}_0) = e^{-rT} \mathbb{E}_{\mathbb{Q}} [\max\{K_l-\tilde{X}^{\hat{\pi}}_T,0\}] = e^{-rT} \mathbb{E}_{\mathbb{Q}} [(K_l-\tilde{X}^{\hat{\pi}}_T) \mathds{1}_{\{\tilde{X}^{\hat{\pi}}_T < K_l\}}] \newline = e^{-rT} \mathbb{E}_{\mathbb{Q}} [K_l \mathds{1}_{\{Z < d_0\}}] - e^{-rT} \mathbb{E}_{\mathbb{Q}} [(\tilde{X}_{0}e^{rT} + \dfrac{\theta}{\alpha} Z \sqrt{T}) \mathds{1}_{\{ Z < d_0\}} ] \newline
  = e^{-rT} K_l \Phi (d_0)-  \tilde{X}_{0} \Phi (d_0) - \dfrac{\theta \sqrt{T}}{\alpha}  e^{-rT}\mathbb{E}_{\mathbb{Q}} [ Z  \mathds{1}_{\{ Z < d_0\}} ] 
  \newline = \Phi (d_0) (e^{-rT} K_l -  \tilde{X}_{0}) - \dfrac{\theta \sqrt{T}}{\alpha}  e^{-rT} \dfrac{1}{\sqrt{2\pi}}
  \bigintsss_{-\infty}^{d_0} Z e^{-{Z^2}/2} dZ \newline = \Phi (d_0) (e^{-rT} K_l -  \tilde{X}_{0}) + \dfrac{\theta \sqrt{T}}{\alpha}  e^{-rT} \dfrac{1}{\sqrt{2\pi}} e^{\small{-\dfrac{d_0^2}{2}}}$ \newline Similarly, for any t $\in$ [0,T] this gives: \newline
 $p(t, \tilde{X}^{\hat{\pi}}_t) = \Phi (d_l(t,\tilde{X}^{\hat{\pi}}_t)) (e^{-r(T-t)} K_l -  \tilde{X}^{\hat{\pi}}_{t}) + \dfrac{\theta \sqrt{T-t}}{\alpha \sqrt{2\pi}}  e^{-r(T-t)} e^{\small{-\dfrac{d_l(t,\tilde{X}^{\hat{\pi}}_t)^2}{2}}}\newline= \Phi (d_l(t,\tilde{X}^{\hat{\pi}}_t)) (e^{-r(T-t)} K_l -  \tilde{X}^{\hat{\pi}}_{t}) + \dfrac{\theta \sqrt{T-t}}{\alpha}  e^{-r(T-t)} \phi(d_l(t,\tilde{X}^{\hat{\pi}}_t)) $ \newline for the notation defined in the Proposition.\end{proof} As a replicating portfolio $\tilde{\pi}_p$ for the pricing function $p(t, \tilde{X}^{\hat{\pi}}_t)$ we suggest
\begin{proposition} The replicating portfolio of the put option in (3.2) is given by  \begin{equation} \tilde{\pi}_p (t, \tilde{X}^{\hat{\pi}}_t) = \dfrac{-\Phi(d_l)}{\sigma \sqrt{T-t}(\Phi(d_l)d_l + \phi(d_l))}
\end{equation} \newline for $d_l=d_l(t, \tilde{X}^{\hat{\pi}}_t) =( K_l-  \tilde{X}^{\hat{\pi}}_{t}e^{r(T-t)}) \dfrac{\alpha}{\sqrt{T-t} \theta}$ as in Proposition 5. \end{proposition} \begin{proof}It needs to be shown that $p(t,\tilde{X}^{\hat{\pi}}_t)$ satisfies the wealth equation for $\tilde{\pi}_p$ from (3.4), because then it is a wealth process that replicates the price function $p(t,\tilde{X}^{\tilde{\pi}}_t)$ and reaches the  terminal value $\max\{K_l- \tilde{X}^{\hat{\pi}}_T ,0\}$, so $\tilde{\pi}_p$ would be a suitable strategy. \newline To do this, form the partial derivatives via the substitution

 $d:= d_l(t,\tilde{X}^{\hat{\pi}}_t)=(K_l-\tilde{X}^{\hat{\pi}}_{t}e^{r(T-t)}) \dfrac{\alpha}{\sqrt{T-t} \theta}$ and $d_x = \dfrac{\alpha}{\theta \sqrt{T-t}}(-e^{r(t-t)})$:
 \newline
$p_x = \dfrac{d}{d\tilde{X}^{\hat{\pi}}_t} p(t, \tilde{X}^{\hat{\pi}}_t) = - \Phi(d) + \phi(d) d_x\{K_le^{-r(T-t)}-\tilde{X}^{\hat{\pi}}_t\} + \dfrac{\theta \sqrt{T-t}}{\alpha} e^{-r(T-t)}\phi(d) (-d) d_x \newline  = -\Phi(d)+ \phi(d) d_x (K_le^{-r(T-t)}-\tilde{X}^{\hat{\pi}}_t - e^{-r(T-t)} ( K_l-  \tilde{X}^{\hat{\pi}}_{t}e^{r(T-t)}) =  -\Phi(d), \newline p_{xx} = \dfrac{d^2}{d(\tilde{X}^{\hat{\pi}}_t)^2} p(t, \tilde{X}^{\hat{\pi}}_t) = \phi(d)\dfrac{\alpha}{\theta \sqrt{T-t}} e^{r(T-t)}$ and\newline
 $p_t  = \dfrac{d}{dt}p(t, \tilde{X}^{\hat{\pi}}_t)= \phi(d) d_t (e^{-r(T-t)}K_l- \tilde{X}^{\hat{\pi}}_t) + \Phi(d)rK_le^{-r(T-t)}  
 $\newline$ + \dfrac{\theta\phi(d)}{\alpha}(\sqrt{T-t}r e^{-r(T-t)} -  \dfrac{1}{2 \sqrt{T-t}}e^{-r(T-t)})  +\dfrac{\theta}{\alpha} \sqrt{T-t} e^{-r(T-t)} \phi(d) (-d)d_t$ \newline which simplifies (first and last term cancel) to  \newline $p_t = \Phi(d)rK_le^{-r(T-t)} + \dfrac{\theta\phi(d)}{\alpha}e^{-r(T-t)}(\sqrt{T-t}r  - \dfrac{1}{2 \sqrt{T-t}}) $
.\newline Note that  $p(t, \tilde{X}^{\hat{\pi}}_t)$ is twice differentiable and  $\tilde{X}^{\hat{\pi}}_t$ satisfies the differential equation from the proof of (2.15), hence is an Ito drift diffusion process. So, Ito's Lemma can be applied, and \newline $dp(t, \tilde{X}^{\hat{\pi}}_t) = \{\Phi(d)rK_l e^{-r(T-t)} + \dfrac{\theta}{\alpha} \phi(d)e^{-r(T-t)}r \sqrt{T-t} + (r \tilde{X}^{\hat{\pi}}_t + \dfrac{\theta^2}{\alpha}e^{-r(T-t)})(-\Phi(\theta))\}dt + \dfrac{\theta}{\alpha}e^{-r(T-t)}(-\Phi(d))dW_t$. \newline Then $p(t, \tilde{X}^{\hat{\pi}}_t)$ satisfies the wealth equation $dp(t, \tilde{X}^{\hat{\pi}}_t) = (r + \tilde{\pi}_p \theta \sigma) p(t, \tilde{X}^{\hat{\pi}}_t)dt + \sigma \tilde{\pi}_p p(t, \tilde{X}^{\hat{\pi}}_t) dW_t$ iff 
 $\dfrac{\theta}{\alpha}e^{-r(T-t)}(-\Phi(d)) =   \sigma \tilde{\pi}_p p(t, \tilde{X}^{\hat\pi}_t)$, which is the case for $\tilde{\pi}_p(t, \tilde{X}^{\hat{\pi}}_t) = \dfrac{-\Phi(d)}{\sigma\sqrt{T-t}(\Phi(d)d + \phi(d))}$. \newline \newline
 \end{proof} 
 Combining the optimal strategy of the modified unconstrained problem and of the 
 replicating portfolio gives the overall strategy.\newline
 \begin{proposition}
  An optimal strategy for Problem 2 is given by the amount to be invested in the risky asset at t
 \begin{equation}\hat{\pi}_l(t,\tilde{X}^{\hat{\pi}}_t)= \dfrac{\theta}{\alpha \sigma}e^{-r(T-t)} + p(t,\tilde{X}^{\hat{\pi}}_t)\dfrac{-\Phi(d_l)}{\sigma \sqrt{T-t}(\Phi(d_l)d_l + \phi(d_l))} \end{equation} 
where $\tilde{X}^{\hat{\pi}}_t$ is an optimal wealth process with initial wealth $\tilde{X}_0$ (the shadow value from Proposition 4) and $d_l$ \text{as in Proposition 6}. \end{proposition}
We will sometimes refer to $\tilde{X}^{\hat{\pi}}_t$ as the \textit{shadow wealth process}.
\begin{proof} 
Note that by definition of $\hat{\pi}_l(t,\tilde{X}^{\hat{\pi}}_t) = \hat{\pi}_t \tilde{X}^{\hat{\pi}}_t + \tilde{\pi}_p(t,\tilde{X}^{\hat{\pi}}_t) p(t,\tilde{X}^{\hat{\pi}}_t) $  and (2.3), the resulting wealth process is 
$X^{\hat{\pi}_l}_t = \tilde{X}^{\hat{\pi}}_t + p(t, \tilde{X}^{\pi}_t)$.  \newline
In particular, at t=T: \newline  
$X^{\pi_l}_T = \tilde{X}^{\hat{\pi}}_T + \max\{K- \tilde{X}^{\hat{\pi}}_T ,0\}$, so it yields the optimal terminal value under lower constraint K, hence it is an optimal strategy.  \end{proof}
\subsection{Analysis of the $K_l$-Strategy}
In order to understand the consequences of a lower constraint on terminal wealth on the optimal strategy and the performance of wealth, we will first look at the formal structure of $\hat{\pi}_l$. From this, we can deduce that the investement is generally lower than the one required by the optimal unconstained strategy, and goes to zero when nearing terminal time. Also, it is most sensitive for values of shadow wealth around $K_l e^{-r(T-t)}$. \newline
We then look at the qualitative behaviour in different scenarios, where we can observe that the $K_l$-strategy is more successful than the unconstrained one for a continued decrease of the stock value, but underperforms it for increasing stock values, in which case the difference is amplified by the expected return rate.

\subsubsection{First Observations}
To start, consider the formula for the $K_l$-strategy given by Proposition 7. Note that the first term, $\hat{\pi}_t\tilde{X}^{\hat{\pi}}_t$ is identical to the optimal unconstrained strategy, since the amount invested is independent of the initial wealth (so, substituting $X_0$ by the shadow value $\tilde{X_0}$ has no effect on the strategy). Hence, the difference of this new strategy to the previous optimal one is determined by the second term, $p(t,\tilde{X}^{\hat{\pi}}_t$)$\tilde{\pi}_p$.\newline 
As $p(t,\tilde{X}^{\hat{\pi}}_t$) is the price of a put option, it is always positive. Then, $p(t,\tilde{X}^{\hat{\pi}}_t$)$\tilde{\pi}_p$ > 0 is the case if and only if $f(d) := \Phi(d)d+ \phi(d)  <$ 0 for a suitable $d$. But $f$ is a strictly positive function, therefore the optimal strategy under lower constraint $K_l$ requires always less investment than the optimal unconstrained strategy. \newline
The parameter time does have an effect on the strategy, as the formula suggests, and as can be seen in the example of Figure 9. Again, we focus on the difference $-\tilde{\pi}_p p(t,\tilde{X}^{\hat{\pi}})$.
\begin{figure}[H]\begin{center}
\includegraphics[width=100mm]{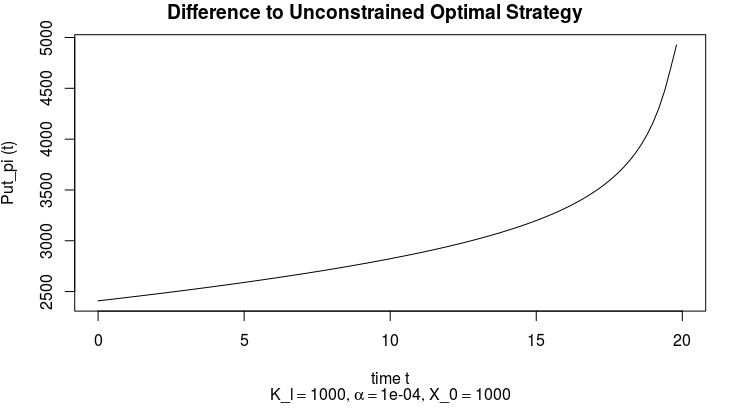}\end{center}
\caption{Difference  $\hat{\pi}$-$\hat{\pi}_l$ over time (amount invested)}
\end{figure} 
The closer we get to the terminal time T (here 20 years), the greater the difference to the optimal strategy becomes, and so the smaller the optimal constrained strategy \newline $\hat{\pi}_l$ = $\hat{\pi} - (-\tilde{\pi}_p p(t,\tilde{X}^{\hat{\pi}}))$ gets and so the less is invested in the risky stock.\newline Since we are rather interested in the maximum difference to the optimal strategy, we now fix t=19 and look at the behaviour of $-\tilde{\pi}_p p(t,\tilde{X}^{\hat{\pi}})$ with respect to the shadow wealth.
 \begin{figure}[H]\begin{center}
\includegraphics[width=100mm]{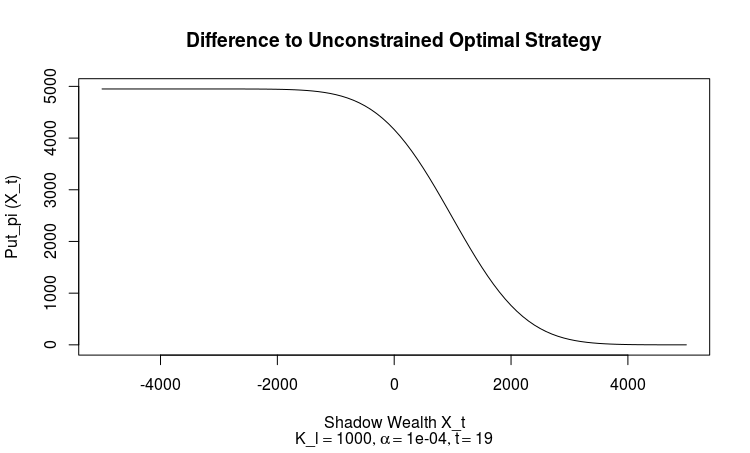}\end{center}
\caption{Difference $\hat{\pi}$-$\hat{\pi}_l$ in function of current shadow wealth (amount invested at t=19)}
\end{figure} 
Here, one can see that for a good performance of the stock (i.e. one that leads to a shadow wealth > 4000), the part of the strategy related to the put-option is very close to zero. So, the constrained and unconstrained strategies are practically identical. However, these high values of wealth are presumably rarely reached, as we will start the process with very small or negative initial shadow wealth. \newline On the other hand, a bad performance implies an investment amount for $\tilde{\pi}_p p(t,\tilde{X}^{\hat{\pi}}_t)$ close to the one required by the optimal unconstrained strategy. Overall, in those cases, the investment for the constrained wealth drops to values close to zero.\newline  Note, that between those two scenarios, the strategy is very sensitive to changes of shadow wealth, as indicated by the steep slope in Figure 10. As a point of orientation, fix \newline
$\tilde{X}^{\hat{\pi}}_t = K_l e^{-r(T-t)}$. By the formula of the strategy, one can easily check that this implies ${d_l = 0}$ and therefore $\hat{\pi}_l = \dfrac{1}{2} \hat{\pi}$, hence the midpoint between the two scenarios is reached. This is why we would generally expect the constrained strategy to adapt fast to changes and be fluctuating, when the shadow wealth is near $K_l e^{-r(T-t)}$.\newline Also note that $ p(t,X^{\hat{\pi}}_t) \tilde{\pi}_p$ is inversely proportional to the variance of the stock. Accordingly, the difference between optimal constrained strategy and optimal unconstrained strategy is reduced for higher $\sigma$. \newline To see how the strategy behaves in practice and how it affects the (terminal) wealth, we look at three different scenarios.
\subsubsection{Different Scenarios} 
\textbf{Scenario 1: $X^{\hat{\pi}}_T > K_l$, $X^{\pi_l}_T =  K_l$ } \begin{figure} [H] 
\begin{minipage}{0.62\linewidth}
\includegraphics[width=72mm]{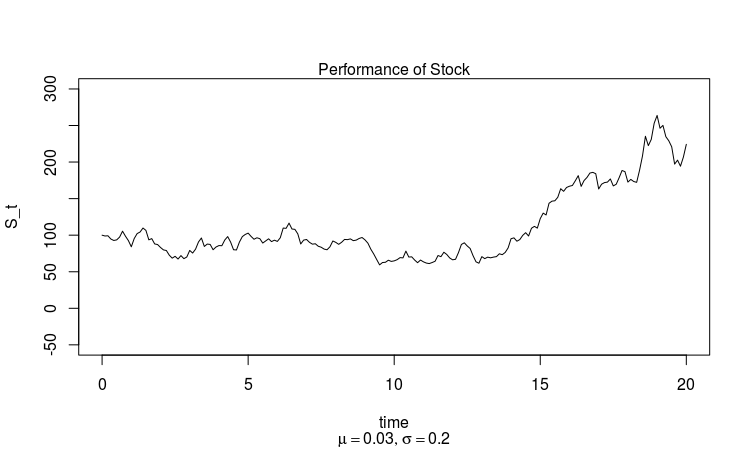} \end{minipage} \begin{minipage}{0.62\linewidth} \includegraphics[width=72mm]{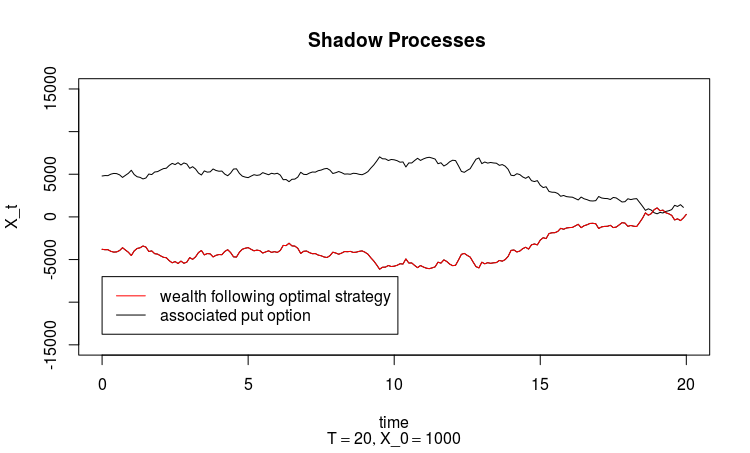}\end{minipage}  \begin{minipage}{0.62\linewidth}\includegraphics[width=72mm]{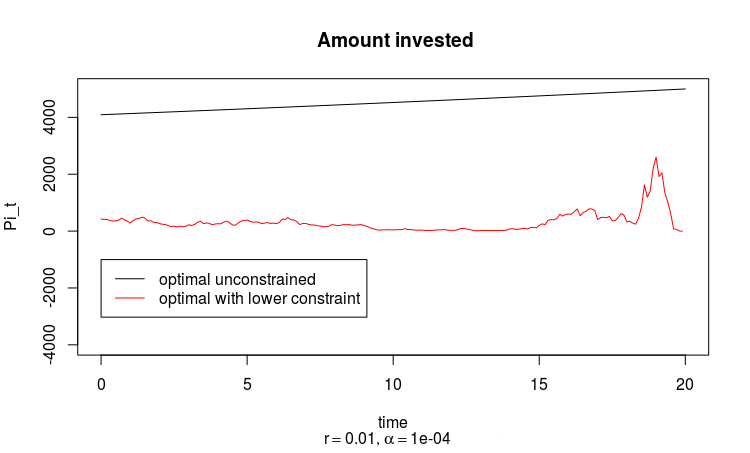} \end{minipage} \begin{minipage}{0.62\linewidth}\includegraphics[width=72mm]{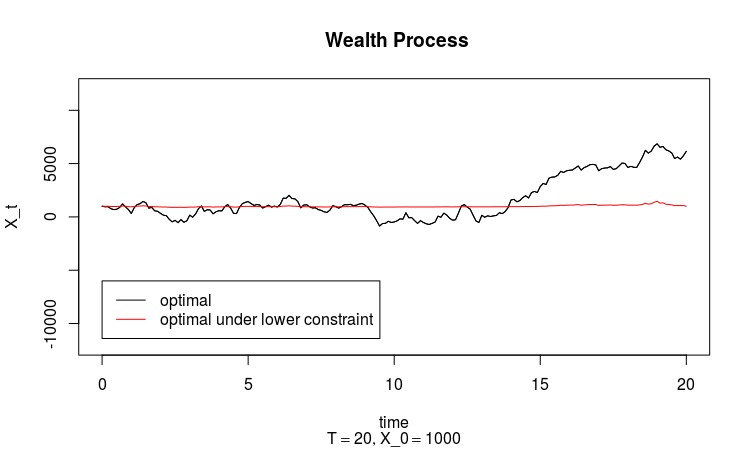} \end{minipage}
\caption{Higher shadow wealth leads to higher sensitivity of $\hat{\pi}_l$.}
\end{figure}
In this case, we can see an increase of the stock value from t=15 on, with a peak around t=18. This increase leads to an increase of the shadow wealth, which implies a decrease of the price of the associated put option. Following the optimal $K_l$-constrained strategy, the investment goes up. Since the values of shadow  wealth are in the sensitive zone (around 0 from t=15 on, see Fig.11), the strategy shows a clear peak and is very volatile. However, the amount invested remains at a rather low level both compared to the optimal strategy, and also proportionally to the overall wealth. Also, the 'jumps' in the investment are a reflection of past movements of the stock. As these big movements are often followed by smaller movements of stock (e.g. at t=18.6 we have an increase of +30 vs. at the next time step we have a decrease of -10 of value of stock), the consequences on the optimal lower constrained wealth are limited.  These two effects could be the reasons why the constrained optimal wealth process is barely affected by the volatility of the stock.
Note that we set ${K_l=  X_0}$=1'000, which is rather high, so the upside and downside potential of the new strategy is limited and most of the wealth is invested in the riskfree bond. 
\newline

\textbf{Scenario 2. $X^{\hat{\pi}}_T < K_l$, $X^{\pi_l}_T =  K_l$ }
\begin{figure} [H] 
\begin{minipage}{0.62\linewidth}
\includegraphics[width=72mm]{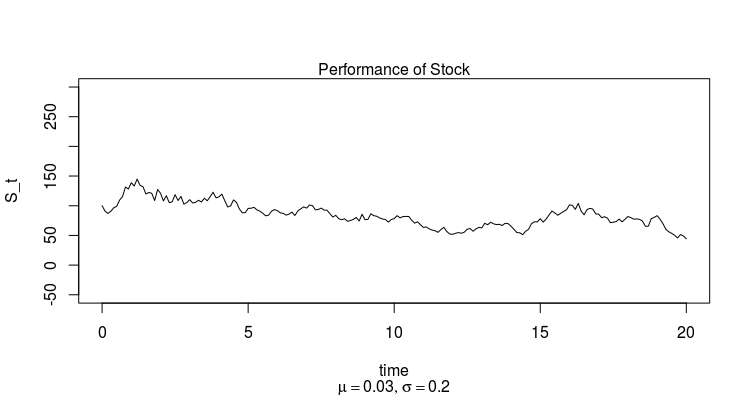} \end{minipage} \begin{minipage}{0.62\linewidth} \includegraphics[width=72mm]{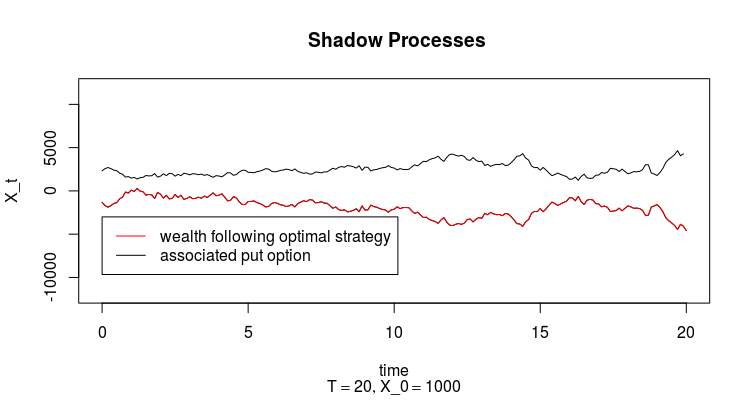}\end{minipage}  \begin{minipage}{0.62\linewidth}\includegraphics[width=72mm]{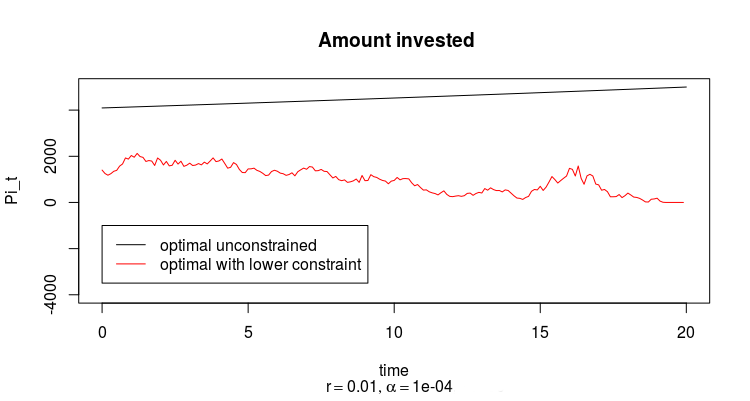} \end{minipage} \begin{minipage}{0.62\linewidth}\includegraphics[width=72mm]{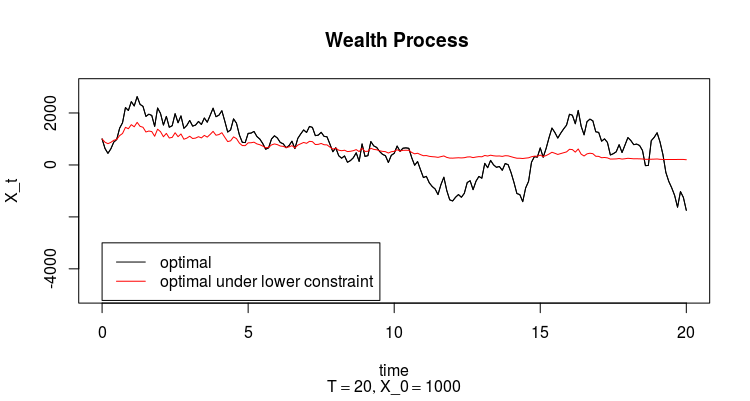} \end{minipage}
\caption{For a decrease in stock value, $\hat{\pi}_l$ goes to zero.}
\end{figure}
In the second scenario, one can see the case in which the constrained strategy outperforms the unconstrained one. Here, we set ${K_l=200}$, to better observe how the new strategy behaves. As the optimal unconstrained strategy consists of investing a deterministic amount, independently from the actual value of wealth (or stock), it fails in the scenario of a decreasing value of stock. This is where the advantage of the constrained strategy kicks in: it reacts to a decrease in wealth by reducing the amount invested. If the downward trend of the stock is continued, this strategy is therefore more successful. 
\newpage
\textbf{Scenario 3:  $X^{\pi_l}_T >  K_l$ }
\begin{figure}[H] 
\begin{minipage}{0.62\linewidth}
\includegraphics[width=72mm]{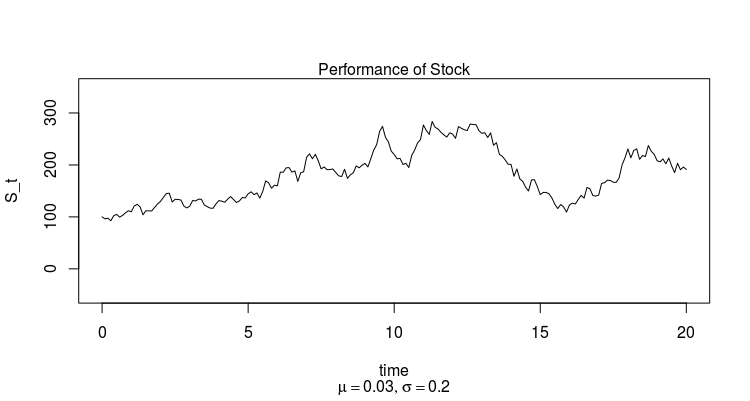} \end{minipage} \begin{minipage}{0.62\linewidth} \includegraphics[width=72mm]{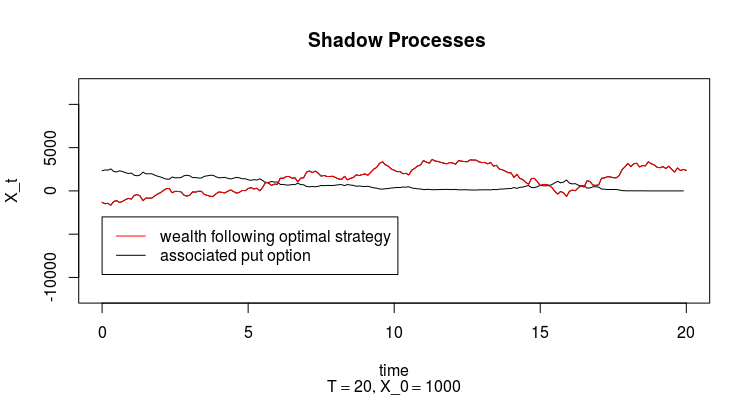}\end{minipage}  \begin{minipage}{0.62\linewidth}\includegraphics[width=72mm]{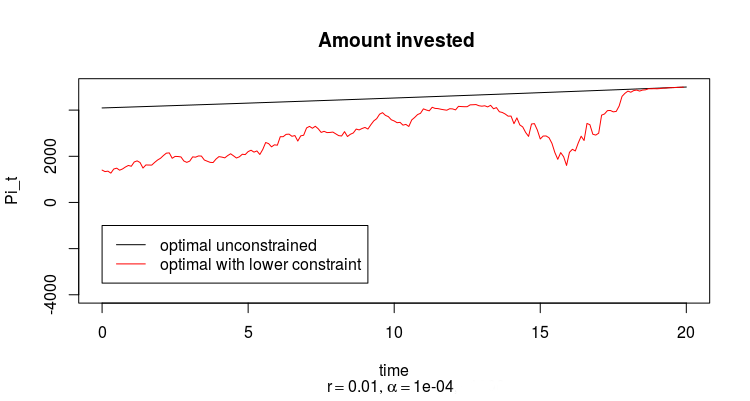} \end{minipage} \begin{minipage}{0.62\linewidth}\includegraphics[width=72mm]{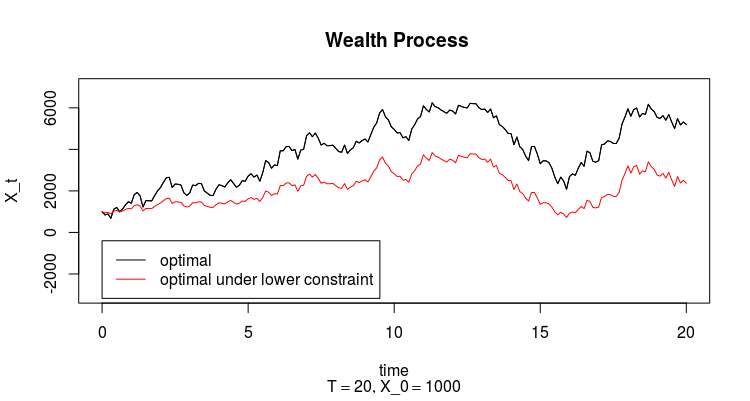} \end{minipage}
\caption{If the stock performs well, the first years are more critical to the missed upside potential by the constrained strategy $\hat{\pi}_l$.}
\end{figure}
The last scenario shows what happens when we have a continued upward trend. As with higher shadow wealth the strategies converge, the investment processes are nearly the same, so the wealth performs in a parallel way. The gap between the resulting terminal wealths, that is a realisation of the missed upside potential, is therefore characterised by what happens for lower shadow wealth values. Since the change of wealth depends on the amount invested multiplied by the change in value of the stock, one should consider both the initial difference between the strategies and the increase of the stock. Both are highly dependant on the expected rate $\mu$, while the other parameters can be neglected. For example, in the case above, the initial difference between the investment is 2'692, which increases for $K_l$>200 (to 3'668 for $K_l$ =1'000), decreases for other reasonable values of $\alpha$ >0.0001, \textit{r} >0.01, $\sigma$>0.2, but for $\mu$=0.06 it is 8'486. This difference in investments is then expected to be multplied by the same high $\mu$, which results in an even higher difference of wealth. This effect can also be observed in Fig.13, where, even though the difference of investments is decreasing, the gap between the corresponding wealth is still growing until circa t=10. 
\newline \newline
Note that only for a continued positive trend of the stock (as we have seen in the last scenario) the optimal constrained strategy results in a terminal wealth larger that $K_l$. For the more frequent scenarios, the terminal wealth is exactly $K_l$. We therefore would expect a probability mass point at $K_l$ in the terminal wealth distribution.
\subsection{Brief Analysis of the $K_l$-Strategy under a Restriction on Investment}
In analogy to the optimal strategy without lower bound, we restrict the investment to be maximum 100$\%$ of the actual wealth.\newline 
\begin{modification}
Let the modified strategy $\hat{\pi}_{l,m}$ be defined for $(t,X^{\hat{\pi}_{l,m}}_t) \in$ [0,T]$\times \mathbb{R}$ by 
 \begin{center}
$\hat{\pi}_{l,m}(t,X^{\hat{\pi}_{l,m}}_t)$ = $\begin{cases}{\hat{\pi}_l (t,X^{\hat{\pi}_{l,m}}_t)}&\text{if $X^{\hat{\pi}_{l,m}}_t \geq \hat{\pi}_l(t,X^{\hat{\pi}_{l,m}}_t) X^{\hat{\pi}_{l,m}}_t$}\\{1}&\text{if $X^{\hat{\pi}_{l,m}}_t <  \hat{\pi}_l(t,X^{\hat{\pi}_{l,m}}_t)X^{\hat{\pi}_{l,m}}_t$}\end{cases}$  \end{center}
 
and $\hat{\pi}_l(t,X_t)$ is the optimal $K_l$-constrained strategy from Proposition 7.   
\end{modification}
Since the investment in the risky stock required by the $K_l$-restricted strategy is always smaller than the one of the unconstrained strategy, the absolute effect of restriction of the investment will generally be lower. But since the new strategy adapts in function of the shadow wealth, which in turn depends on the stock performance, we can observe that a lower value of wealth is connected to a lower investment amount. The difference between this investment-restricted $K_l$-strategy and the unrestricted one in terms of the amount invested will therefore depend less on the movements of stock. However, as seen before, once the investment is limited to 100$\%$ of wealth, it falls below the optimal one and so the wealth is likely to grow less since the process can not fully profit from the upside potential. This results in a higher chance of staying under the required optimal investment amount, so the wealth process is more likely to be 'trapped' at low values.  The earlier in the process it happens, the more effect it will have, so we will look on the starting conditions more closely. It is not surprising that we will see a connection between the resulting terminal wealth and the initial setting of wealth and optimal investment. 
\newline \newline
Two parameters can be identified to have the most impact on the performance of the strategies and the restriction: the variance $\sigma$ and the lower bound $K_l$. \newline
Generally speaking, the standard situation is that the investment required at t=0 is already exceeding the initial wealth $X_0$. This can not easily be seen and  has to do with the structure of the $K_l$-constrained strategy, being that $X_0$ is the basis for the calculation of the shadow value $\tilde{X_0}$, which in turn determines the initial investment. In the diagrams of Figure 14 we see the initial investment as a function of $X_0$ and the line y= $X_0$ for comparison (having set \textit{r}=0.01, $\mu$=0.03, \textit{T} =20) for different combinations of those two parameters. The corresponding wealth processes (for a sample of 30) for a fixed $X_0$= 1000 are also displayed. 
\begin{figure}[H]
\begin{minipage}{0.27\linewidth} 
\includegraphics[width=40mm]{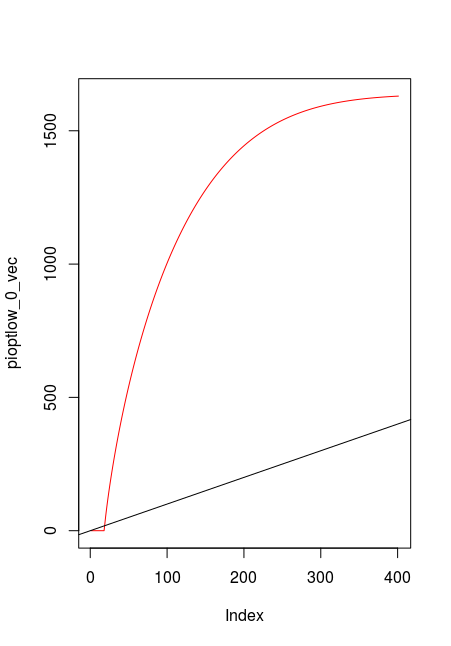} \end{minipage} \begin{minipage}{0.55\linewidth} 
\includegraphics[width=80mm]{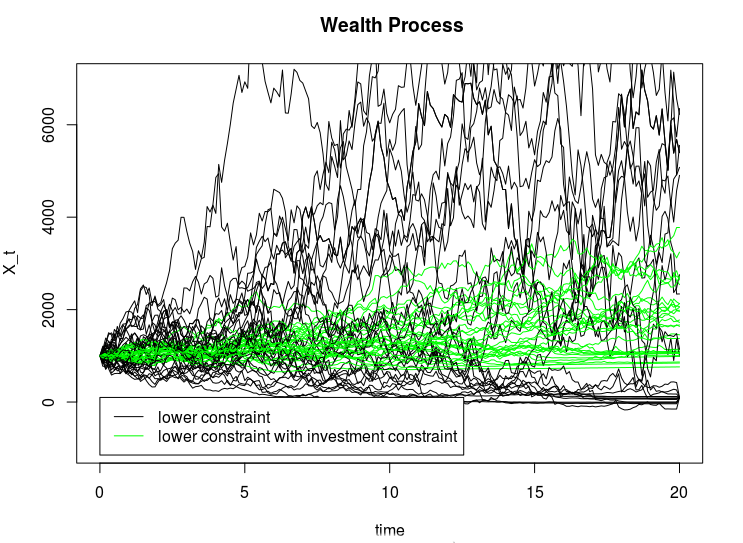}\end{minipage}\begin{minipage}{0.1\textwidth}{\small{$K_l$=100, $\sigma$=0.1}} \end{minipage}\newline
\begin{minipage}{0.27\linewidth}
\includegraphics[width=40mm]{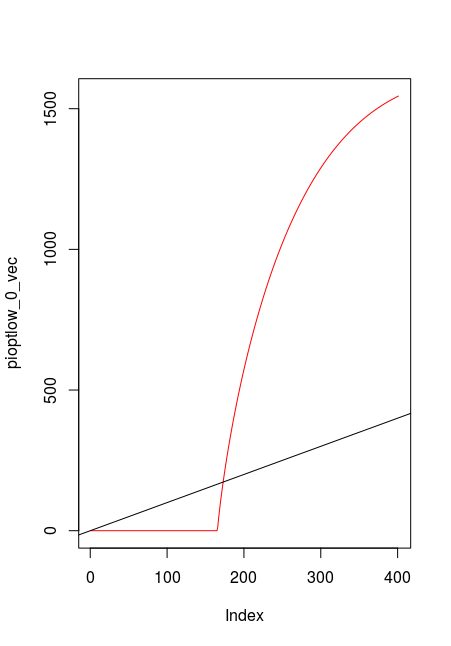} \end{minipage} \begin{minipage}{0.55\linewidth}
\includegraphics[width=80mm]{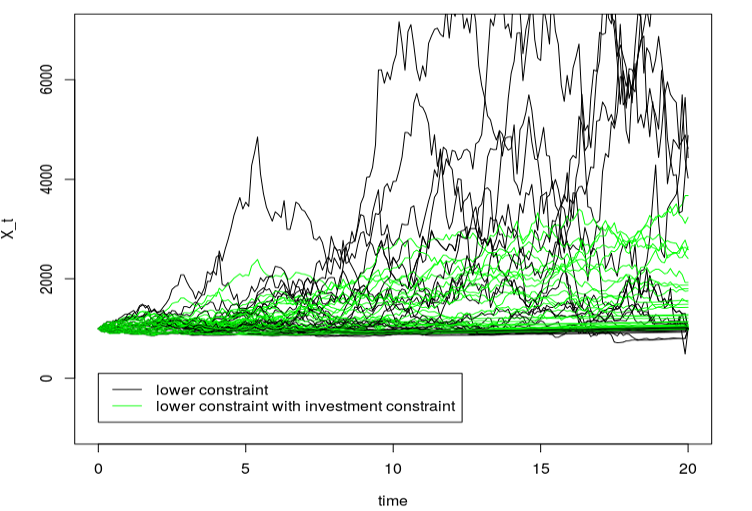} \end{minipage} \begin{minipage}{0.1\textwidth}{\small{$K_l$=1000, $\sigma$=0.1}} \end{minipage}\newline
\begin{minipage}{0.27\linewidth} 
\includegraphics[width=40mm]{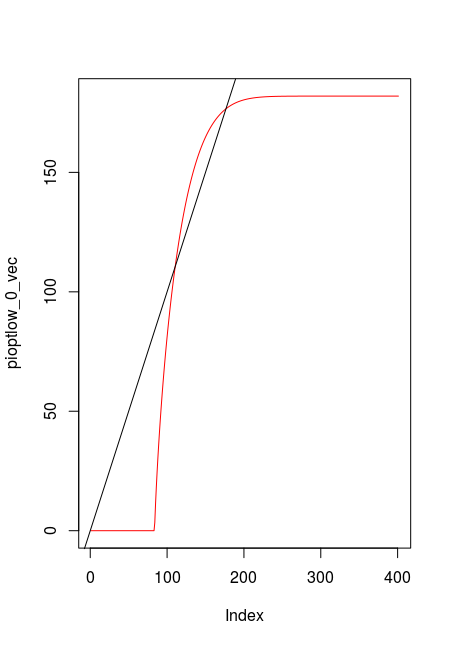} \end{minipage} 
\begin{minipage}{0.55\linewidth}
\includegraphics[width=80mm]{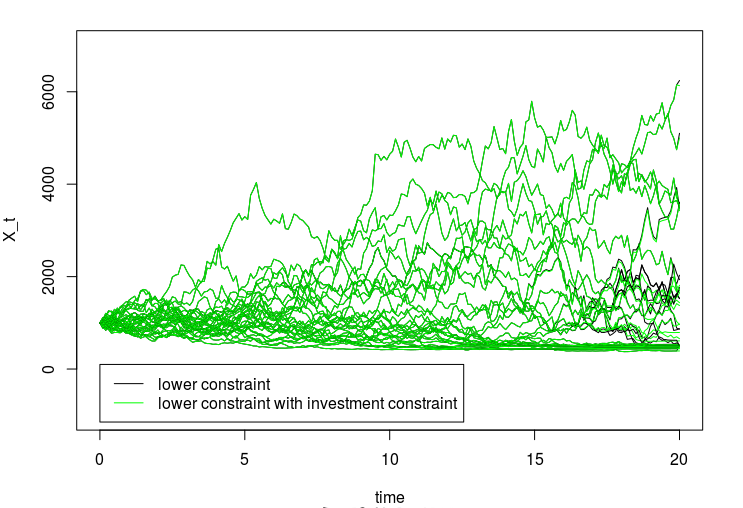} \end{minipage} \begin{minipage}{0.1\textwidth}{\small{$K_l$=500, $\sigma$=0.3} }\end{minipage}\newline
\caption{Initial investment as a function of $X_0$, and wealth process for fixed $X_0$ }
\end{figure}

One can observe that for small $K_l$ the difference between those lines is rather big (except for small values $X_0$, which mostly can be excluded as we had set $X_0$ >$K_le^{-rT}$ $\approx$ 80$\%$ $K_l$), and so the actual investment will be much smaller for the strategy with the investment constraint. This will lead to fewer movements in the wealth, hence the terminal wealth will be more concentrated. For a downward trend of stocks, the investment will be close to 0, which will be an advantage, as the investment-constrained strategy will yield a concentration of terminal wealth around values higher than $K_l$ (where the wealth of the strategy without investment-constraint would end up in those cases). Also note, that even if the wealth processes are identical for some time, a decrease in value of the stock will reduce the investment amount faster for the investment-constrained strategy, so possible loss is minimised. At the same time, of course, for an upwards trend of stocks, the strategy without constraints on investment benefits from higher investments and therefore higher returns.\newline  
Another relevant parameter for the modified strategy seems to be the market volatility $\sigma$. For example, using $\sigma$=0.3 instead of $\sigma$=0.1, generally leads to an initial investment lower than the initial wealth, hence less difference between the investment-restricted and unrestricted strategies. 
\subsection{Comparison of Terminal Wealth Distributions}
After the qualitative observations, we will now quantify the resulting terminal wealth distribution. We are particularly interested in the difference to the distribution of the unconstrained optimal terminal wealth and the impact of the investment constraint. \newline
Let us first look at the theoretical distribution of the optimal terminal wealth under the lower constraint, and then compare it to the empirical results. 
\subsubsection{Theoretical vs. Empirical Distribution }
We begin with the derivation of the theoretical distribution of terminal wealth for an investor that follows the optimal unconstrained strategy from the last chapter. From Proposition 3, we have \newline
$\mathbb{P} [X^{\hat{\pi}}_T \le x] = \mathbb{P}[X^{\hat{\pi}}_0 e^{rT}+T\dfrac{\sigma^2}{\alpha} + \dfrac{\theta}{\alpha} W_T \le x]$
\newline = $\mathbb{P}[W_T \le (x-X_0 e^{rT}-T\dfrac{\theta^2}{\alpha})\dfrac{\alpha}{\theta}] \newline = \Phi(d_T)$, where $d_T := (x-X_0 e^{rT}-T\dfrac{\theta^2}{\alpha})\dfrac{\alpha}{\theta \sqrt{T}}$. \newline Now consider the the optimal strategy, where the terminal wealth is subject to a  lower constraint. \newline
$\mathbb{P}[X^{\hat{\pi}_l}_T \le x]\newline  =  \mathbb{P}[\tilde{X}^{\hat{\pi}}_T + \max\{K_l- \tilde{X}^{\hat{\pi}}_T ,0\}\le x] $ by Proposition 4,$ \newline = \mathbb{P}[K_l \le x |\tilde{X}^{\hat{\pi}}_T < K_l ] \mathbb{P}[\tilde{X}^{\hat{\pi}}_T < K_l] + \mathbb{P}[X^{\hat{\pi}_l}_T \le x |\tilde{X}^{\hat{\pi}}_T \geq K_l ]P[\tilde{X}^{\hat{\pi}}_T \geq K_l] \newline
= \mathbb{P}[K_l \le x ] \mathbb{P}[\tilde{X}^{\hat{\pi}}_T < K_l] + \mathbb{P}[ K_l \leq  X^{\hat{\pi}}_T \le x]$ \newline
=  $\begin{cases}{\mathbb{P}[\tilde{X}^{\hat{\pi}}_T < x]}&\text{if $x \geq K_l $}\\{0 }&\text{if $x < K_l.$}\end{cases} 
\newline \newline
$So, the cumulative distribution of the optimal terminal wealth under the $K_l$-constrained strategy has a probability mass point at $K_l$ and follows the distribution of the terminal wealth of an unconstrained optimal strategy with a shadow value for initial wealth for $x > K_l$. This is expressed as a jump from zero to a positive value at $K_l$ in the CDF, as it is shown below for a standard example ($X_0$= 1'000, $K_l$=800, $\alpha$=0.0001, r=0.01, $\mu$=0.03, $\sigma$=0.1 and T = 20, samplesize 1'000) . \newline
\begin{figure} [H]
 \includegraphics[width=70mm]{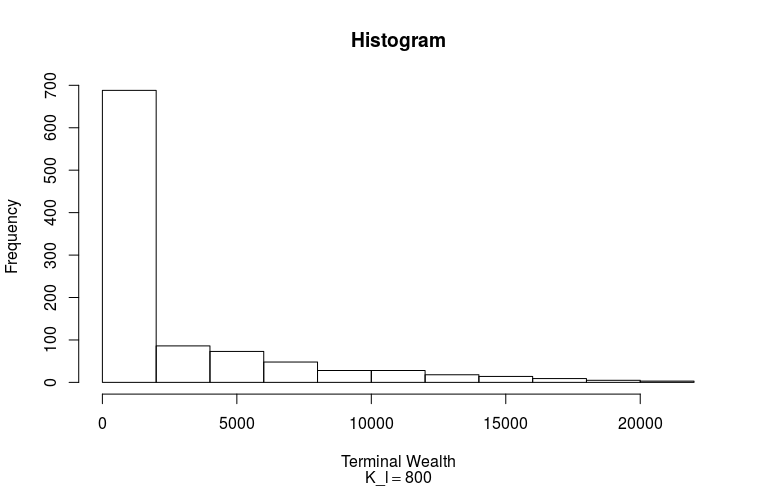} \includegraphics[width=70mm]{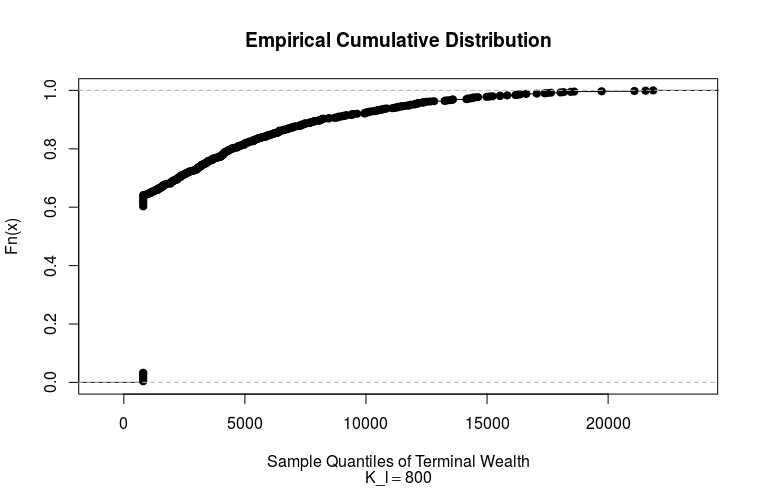} 
 \caption{Empirical distribution of terminal wealth for $\hat{\pi}_l$}
\end{figure}
In this example, the probability, that the optimal terminal wealth lands exactly on the lower constraint, is circa 60$\%$ and with higher $K_l$ this probability increases even more. \newline 
Consequently, as $K_l$ $\rightarrow$ - $\infty$, the jump is shifted to the left, and the distribution converges to the distribution of the shadow terminal wealth, which is normally distributed. At the same time, the shadow value $\tilde{X}_0$ converges to $X_0$, so the shadow wealth process converges to the optimal unconstrained process. This can also be quantitatively observed in the theoretical quantiles below, where the unconstrained strategy is highlighted as a reference. 
\begin{table}[H]

\begin{center} \begin{tabular}{|l| c |c|c|c|c|}
\hline
 $K_l$ & $\tilde{X_0}$ & $\mathcal{Q}_{0.25}$($X_0$) &   $\mathcal{Q}_{0.50}$($X_0$) &  $\mathcal{Q}_{0.75}$($X_0$) &  $\mathcal{Q}_{0.95}$($X_0$)   \\ \hline 
 \cellcolor{OldLace} -$\infty$  &   \cellcolor{OldLace} -&  \cellcolor{OldLace} 3'188.6 &  \cellcolor{OldLace} 9'221.4 &  \cellcolor{OldLace} 15'254.2 & \cellcolor{OldLace} 23'933.4  \\\hline 
 \cellcolor{LightGrey}  1'000 & -10'701.41 &\cellcolor{LightGrey}  1'000 &\cellcolor{LightGrey}  1'000 & \cellcolor{LightGrey} 1'000 &9'641.3  \\
\cellcolor{LightGrey}  -1'000 &-3'377.14  &\cellcolor{LightGrey}  -1'000 & 3'875.2 & 9'908.0  &18'587.2  \\
 -30'000 &999.5621& 3'188.0 & 9'220.9 & 15'253.7 &23'932.9 \\ \hline \end{tabular} \end{center}
 \caption{Theoretical quantiles of $X^{\hat{\pi}_l}$ for different $K_l$} 
\end{table}

Before we look at the impact the investment-constraint has on the terminal $K_l$-bounded wealth distribution, we want to get a sense of the errors between the theoretical distribution and the empirical one (i.e. the one we obtain from the simulation implemented in the code). Here, there are mainly two sources of error: One being the sample size that allows only an approximation of the quantiles, the other one being the step width h, which describes the frequency of re-balancing the portfolio. The latter does also play a role in real-life situations, as it is unrealistic to expect continuous trading (as assumed in the derivation of the theoretical strategy) because of technical, time and cost constraints.\newline
To gain an idea of the dimension of the error, we consider a standard case ($X_0$ = 1'000, $K_l$=100, T = 20, $\alpha$ =0.0001, $\mu$ = 0.03, \textit{r} = 0.01, $\sigma$ = 0.1)  with a fixed h = 0.1 and different sample sizes. As a reference, the theoretical quantiles of $X^{\hat{\pi}_l}_T$ are also listed in terms of total return.

\begin{table}[H]
 
\begin{center}\begin{tabular}{|l| c |c|c|c|}
\hline
 Samplesize s & $\mathcal{Q}_{0.25} (X^{\hat{\pi}_l}_T)$ &   $\mathcal{Q}_{0.50} (X^{\hat{\pi}_l}_T)$ &   $\mathcal{Q}_{0.75} (X^{\hat{\pi}_l}_T)$ &   $\mathcal{Q}_{0.95} (X^{\hat{\pi}_l}_T)$  \\ \hline 
 Theoretical return  &  10$\%$ &   116$\%$ &  719$\%$& 1'587$\%$ \\\hline 
s = 1'000 & 0$\%$ &  17 $\%$& 0 $\%$ & -1 $\%$  \\
s = 3'000  & 0$\%$  &  5$\%$& -2$\%$  & 1 $\%$   \\
s = 5'000 & 0$\%$& 12$\%$ & -1 $\%$ & 0$\%$\\ \hline \end{tabular} \end{center}
 \caption{Deviation of empirical from theoretical quantiles of $X^{\hat{\pi}_l}_T$ (different samplesizes)}
\end{table}

The deviation at $\mathcal{Q}_p (X^{\hat{\pi}_l}_T)$ is given by $[\mathcal{Q}^{emp,s}_p (X^{\hat{\pi}_{l}}_T) - \mathcal{Q}^{theor}_p (X^{\hat{\pi}_l}_T) ]/ {\mathcal{Q}^{theor}_p (X^{\hat{\pi}_l}_T})$, where $\mathcal{Q}^{emp,s}_p (X^{\hat{\pi}_{l}}_T)$ is the p-quantile from the empirical distribution obtained from a sample size s.

Now we fix s = 3'000 and look at different step widths h. For the reader's orientation, the concrete interpretation of h would be: h=1/10 means 'once a month', h=1/49 means 'once a week', h=1/100 means 'twice a week'. \newline Again, the deviation is  $[\mathcal{Q}^{emp,h}_p (X^{\hat{\pi}_{l}}_T) - \mathcal{Q}^{theor}_p (X^{\hat{\pi}_l}_T) ]/ {\mathcal{Q}^{theor}_p (X^{\hat{\pi}_l}_T})$, where $\mathcal{Q}^{emp,h}_p (X^{\hat{\pi}_{l}}_T)$ is the p-quantile from the empirical distribution obtained from a step width h.

\begin{table}[H]
 
\begin{center}\begin{tabular}{|l| c|c|c|c|}
\hline
 step width h & $\mathcal{Q}_{0.25} (X^{\hat{\pi}_l}_T)$ &   $\mathcal{Q}_{0.50} (X^{\hat{\pi}_l}_T)$ &   $\mathcal{Q}_{0.75} (X^{\hat{\pi}_l}_T)$ &   $\mathcal{Q}_{0.95} (X^{\hat{\pi}_l}_T)$ \\\hline 
h = 1/10 & 0$\%$ &  12 $\%$& -5 $\%$ & -2 $\%$  \\
h = 1/49 & 0$\%$  & -4$\%$& 4$\%$  & 0 $\%$   \\
h = 1/100 & 0$\%$& -6$\%$ & -1 $\%$ & -1$\%$\\ \hline \end{tabular} \end{center}
 
 \caption{Deviation of empirical from theoretical quantiles of $X^{\hat{\pi}_l}_T$  (different stepwidths)}
\end{table}

First note, that the standard deviation of the theoretical normal distribution is very high (770$\%$) and so the standard deviation of the sample mean also is: ${770 \% / \sqrt{3'000} \approx 14\%}$. Therefore the values of the deviation at $\mathcal{Q}_{0.50}$ are high, but lie within reasonable range. For this thesis, the computational and time ressources were limited, but it would be interesting to further investigate the tradeoff between error reduction by using a greater sample size and the additional computing resources needed. Compared to the sample size, the size of step width seems to have less impact on the error, indicating that the focus should be on the sample size if one wished to achieve more confidence.\newline
The error of 0$\%$ at $\mathcal{Q}_{0.25}$ can also be explained, because this quantile lies below the lower constraint $K_l$= 100 that is equivalent to a total return of 10$\%$.  \newline
\newline
In addition, these results do not change significantly for lower $K_l$. For example if ${K_l = -30'000}$, the range of errors is similar. This means in particular, that these observations are also true for implementation of the optimal unconstrained strategy. \newline
It also needs to be added that the calculation of the shadow value in the code produces a slight error, but since it is usually <$ 10^{-9}$, we consider it to be negligible. Furthermore, the normal random generator implemented in R might produce inaccuracy of the normal distribution in the tails. However, we would not see this here in the quantiles.
\subsubsection{Impact of Restriction on Investment}
In order to evaluate the difference to the theoretical distribution produced by the limitation of investment to 100$\%$, we look at the three scenarios from Section 3.3 (having fixed ${h = 1/49}$, s = 3'000).
Here, the quantiles are given as a total return on initial wealth, and $\Delta\mathcal{Q}_p := [\mathcal{Q}^{emp}_p (X^{\hat{\pi}_{l,m}}_T) - \mathcal{Q}^{theor}_p (X^{\hat{\pi}_l}_T) ]/ {\mathcal{Q}^{theor}_p (X^{\hat{\pi}_l}_T})$ ( for $\mathcal{Q}^{emp}_p$ the empirical and $\mathcal{Q}^{theor}_p$ the theoretical p-quantile)  measures the deviation. 
\begin{table}[H]
\begin{tabu}{|l|c|c|c|c|c|}
\hline
 Scenario &Distr. & $\mathcal{Q}_{0.25}$($X_0$) &   $\mathcal{Q}_{0.50}$($X_0$) &  $\mathcal{Q}_{0.75}$($X_0$) &  $\mathcal{Q}_{0.95}$($X_0$)   \\ \hline 
$K_l$= 100,&  Theor. & 10$\%$ &   116$\%$ &  719$\%$& 1'587$\%$  \\
 $\sigma$ = 0.1 & $\hat{\pi}_{l,m}$ &  110 $\%$& 162 $\%$ & 225 $\%$  & 343 $\%$ \\\cline{2-6}
\rowfont{\color{DarkBlue}} 
  & $\Delta\mathcal{Q}_p$ &  1004 $\%$& 40 $\%$ & -69 $\%$ & -78 $\%$  \\ \cline{1-6}
$K_l$= 1000, & Theor. & 100$\%$& 100$\%$  & 100 $\%$  & 964 $\%$ \\
 $\sigma$ = 0.1  & $\hat{\pi}_{l,m}$  &  106$\%$& 128$\%$  & 194 $\%$  & 340 $\%$   \\\cline{2-6}
 \rowfont{\color{DarkBlue}}
  & $\Delta\mathcal{Q}_p$ &  6$\%$& 28$\%$  & 94 $\%$  & -65 $\%$  \\ \cline{1-6}
 $K_l$ = 500,  & Theor. &  50$\%$& 50$\%$  & 231 $\%$  & 520$\%$  \\
  $\sigma$ = 0.3  & $\hat{\pi}_{l,m}$  & 51$\%$& 69$\%$  & 223 $\%$  & 519 $\%$  \\\cline{2-6}
  \rowfont{\color{DarkBlue}}& $\Delta\mathcal{Q}_p$& 2$\%$ & 39 $\%$ & -3$\%$ & 0 $\%$\\ \hline \end{tabu} 
 
 \caption{Empirical, investment-constrained vs. theoretical quantiles of $X^{\hat{\pi}}_T$ (varying scenarios)}
\end{table}
 These results reflect the behaviour we have seen in the qualitative analysis: For low $K_l$, the effect of the constraint on investment is generally stronger. Also, the positive effect on lower quantiles can be seen clearly, which is linked to the concentration around values larger than $K_l$. For $K_l$ = 1'000, the uplift of small quantiles is not that strong, the impact of the investment-constraint can rather be detected around values in the middle, which is also consistent with Figure 14.  Note, that in these two cases the 'price' for the uplift of lower quantiles is a strong reduction of $\mathcal{Q}_{0.95}$. This result is even worse if we consider the relative loss in utility instead of wealth, which might be a more consistent way to evaluate the return. For example in the first scenario, the theoretical utility  at  $\mathcal{Q}_{0.5}$ of -0.89 increases by 6$\%$, whereas at $\mathcal{Q}_{0.95}$ it declines from -0.2 by 255$\%$. In more volatile market conditions, however, the optimal $K_l$-constrained strategy will change much less if a restriction on the investment in the stock is introduced. 
\newline \newline
To see if the lower-bounded terminal wealth distribution is more affected by a constraint on investment to 100$\%$ than the optimal unbounded terminal wealth distribution, we look at the situation from Section 2.3 for comparison. Here, the initial wealth was set to be 120$\%$ , 100$\%$ and 80$\%$  of the amount required by the optimal strategy for an initial investment in stocks. Because of the convergence of the strategies, for $K_l$ = -30'000, the results  from $\hat{\pi}_{l,m}$ are identical to the ones from the optimal (unbounded) strategy $\hat{\pi}_{m}$ . 
\begin{table}[H]{
\subcaption{$K_l$ = -30'000}
\begin{tabu}{|l|c|c|c|c|c|}
\hline
 $X_0$ &Distr. & $\mathcal{Q}_{0.25}$($X_0$) &   $\mathcal{Q}_{0.50}$($X_0$) &  $\mathcal{Q}_{0.75}$($X_0$) &  $\mathcal{Q}_{0.95}$($X_0$)   \\ \hline 
4'912&  Theor. & 101$\%$ &   163$\%$ &  224$\%$& 313$\%$  \\
  & $\hat{\pi}_{l,m}$ &  80 $\%$& 160 $\%$ & 215$\%$  & 302 $\%$ \\\cline{2-6}
\rowfont{\color{DarkBlue}} 
  & $\Delta\mathcal{Q}_p$ &  -21 $\%$& -2 $\%$ & -4 $\%$ & -3 $\%$  \\ \cline{1-6}
4'094 & Theor. & 97$\%$& 171$\%$  & 245 $\%$  & 351 $\%$ \\
 & $\hat{\pi}_{l,m}$  &  74$\%$& 158$\%$  & 227 $\%$  & 336 $\%$   \\\cline{2-6}
 \rowfont{\color{DarkBlue}}
  & $\Delta\mathcal{Q}_p$ &  -24$\%$& -7$\%$  & -7 $\%$  & -4 $\%$  \\ \cline{1-6}
 3'275 & Theor &  91$\%$& 183$\%$  & 275 $\%$  & 408$\%$  \\
   & $\hat{\pi}_{l,m}$  & 70$\%$& 141$\%$  & 238 $\%$  & 378 $\%$  \\\cline{2-6}
  \rowfont{\color{DarkBlue}}& $\Delta\mathcal{Q}_p$& -24$\%$ & -23$\%$ & -13$\%$ & -7 $\%$\\ \hline \end{tabu} 

\subcaption{$K_l$ = 3'000}
\begin{tabu}{|l|c|c|c|c|c|}
\hline
 $X_0$ &Distr. & $\mathcal{Q}_{0.25}$($X_0$) &   $\mathcal{Q}_{0.50}$($X_0$) &  $\mathcal{Q}_{0.75}$($X_0$) &  $\mathcal{Q}_{0.95}$($X_0$)   \\ \hline 
4'912&  Theor. & 82$\%$ &   144$\%$ &  205$\%$& 293$\%$  \\
  & $\hat{\pi}_{l,m}$ &  82 $\%$& 145 $\%$ & 205$\%$  & 292 $\%$ \\\cline{2-6}
\rowfont{\color{DarkBlue}} 
  & $\Delta\mathcal{Q}_p$ &  -0 $\%$& 1 $\%$ & 0 $\%$ & -1 $\%$  \\ \cline{1-6}
4'094 & Theor. & 73$\%$& 132$\%$  & 206 $\%$  & 312 $\%$ \\
 & $\hat{\pi}_{l,m}$  & 76$\%$& 133$\%$  & 205 $\%$  & 310 $\%$ \\\cline{2-6}
 \rowfont{\color{DarkBlue}}
  & $\Delta\mathcal{Q}_p$ &  4$\%$& 1$\%$  & 0 $\%$  & -1 $\%$ \\ \cline{1-6}
 3'275 & Theor. &  92$\%$& 95$\%$  & 187 $\%$  & 320 $\%$  \\
   & $\hat{\pi}_{l,m}$  &  91$\%$& 183$\%$  & 185 $\%$  & 309$\%$  \\\cline{2-6}
  \rowfont{\color{DarkBlue}}& $\Delta\mathcal{Q}_p$& 0$\%$ & 2$\%$ & -1$\%$ & -3 $\%$\\ \hline \end{tabu} }
 
\caption{Empirical, investment-constrained vs. theoretical quantiles of $X^{\hat{\pi}}_T$ (varying $K_l$)}
\end{table}
As expected, the impact of the investment constraint is lower for the $K_l$-bounded strategy $\hat{{\pi}_l}$, simply because it generally requires less investment. In particular, $X_0$ is a percentage of the initial investment $\hat{{\pi}}_0 X_0$, which is generally larger than  $\hat{{\pi}_l} (0) X_0$. Another difference is the observation, that the constraint on the investment has a negative effect on the optimal unconstrained strategy $\hat{{\pi}}$, contrary to the positive effect on $\hat{{\pi}}_l$.
This is, because the higher investment in stocks required by $\hat{{\pi}}$ comes with a stronger upside potential in the first place. Limiting this investment (which is more likely in case of bad performance, i.e. at lower quantiles) also reduces the possibility for high returns, which is expressed in lower values of terminal wealth. Finally, the effect on $\mathcal{Q}_{0.25}$ for initial value $X_0$ = 3'275 is particularly interesting, because 92$\%$ is the lower constraint $K_l$. So, even though the difference is small here, it indicates that the terminal wealth can fall below the lower constraint if a restriction on investment is introduced, and in other cases the effect might be stronger. 
\newpage

\section{An Optimal Strategy for Exponential Utility and \mbox{Upper} Constraint $K_u$}
We will now expand the previous problem by adding an upper bound for the terminal wealth. The resulting strategy will then be qualitatively analysed. To understand the resulting terminal wealth distribution we will take a more theoretical approach. 
\subsection{Derivation of the $K_l$-$K_u$-Strategy }
The optimal strategy under upper and lower constraints for terminal wealth will be derived by first deriving the optimal strategy for the isolated case of only an upper constraint. We will find, that it involves selling a call option and using the extra income to follow the optimal strategy. Combining this with the lower constraint will then require to use a part of the new initial shadow wealth to buy a put option, which then will reduce its value. Again, for simpliciy, we will sometimes call this strategy the $K_u$-Strategy and the one combined with the lower constraint the $K_l$ -$K_u$-Strategy.
\subsubsection{Derivation of the $K_u$-Strategy}
We now look at a setting where the terminal wealth faces an upper constraint $K_u\in \mathbb{R}_{[X_0 e^{rT}, \infty)}$ and use similar arguments as before to determine the optimal strategy. We begin by stating the modified problem. \newline
\begin{problem}
Find an optimal strategy $\hat{\pi}_u \in \mathcal{A}$ such that \begin{equation} \mathbb{E}[U(X^{\hat{\pi}_u}_T)] = \sup\limits_{\pi\in \mathcal{A}}\mathbb{E}[U(X^{\pi}_T)]\ \text{and} \ X^{\hat{\pi}_u}_T \leq K_u \text{ a.s. holds.}\end{equation}\end{problem}
This is solved in a similar way to the lower constraint, by determining the structure  of the optimal terminal wealth and constructing a replicating portfolio. \newline However, we first need to establish a statement analogous to Lemma 2 from \cite{Zhou}.\newpage
\begin{lemma}
The optimal terminal wealth corresponding to the solution of Problem 3 is  given by \begin{equation} X^{\hat{\pi}_u}_T = \min\{K_u,I(yH_T )\},\end{equation}  where U is concave, I is the inverse of $U^\prime$ and  y is a positive number such that the budget constraint $\mathbb{E}[H_T X^{\hat{\pi}}_T] = X_0$ holds. \end{lemma}
\begin{proof} The proof is similar to the one  from \cite{Zhou}.\newline Since U is concave, we have  $U(a)-U(b) \leq  U^{\prime}(b)(a-b) \ \forall \ a,b \in \ \mathbb{R}$.\newline From this follows in particular for any admissible strategy $\pi$ such that $X^{\pi}_T \leq K_u$: \newline $\mathbb{E}[U(X^{\pi}_T)] - \mathbb{E}[U(X^{\hat{\pi}_u}_T)] \newline \leq \mathbb{E}[U'(X^{\hat{\pi}_u}_T) (X^{\pi}_T - X^{\hat{\pi}_u}_T)] \newline =  \mathbb{E}[U'(X^{\hat{\pi}_u}_T) (X^{\pi}_T - X^{\hat{\pi}_u}_T)| X^{\hat{\pi}}_T \geq K_u ]   \mathbb{P}[ X^{\hat{\pi}_u}_T \geq K_u]     \newline + \mathbb{E}[U'(X^{\hat{\pi}_u}_T) (X^{\pi}_T - X^{\hat{\pi}_u}_T)|     X^{\hat{\pi}_u}_T < K_u      ] \mathbb{P} [ X^{\pi_u}_T < K_u ]$. \newline Evaluating the second term:  \newline  $X^{\hat{\pi}_u}_T < K_u  \implies  X^{\hat{\pi}_u}_T  = I(yH_T)\text{ and since } U'(I(yH_T)) = yH_T, \text{we have}\newline \mathbb{E}[U'(X^{\hat{\pi}_u}_T) (X^{\pi}_T - X^{\hat{\pi}_u}_T)|     X^{\hat{\pi}_u}_T < K_u ] = \mathbb{E}[ yH_T (X^{\pi}_T - X^{\hat{\pi}_u}_T)|     X^{\hat{\pi}_u}_T < K_u ] $\newline Evaluating the first term: \newline From $X^{\hat{\pi}_u}_T = \min\{K_u,I(yH_T )\} \leq K_u \implies X^{\hat{\pi}_u}_T = K_u$ follows $(X^{\pi}_T - X^{\hat{\pi}_u}_T) = (X^{\pi}_T - K_u) \leq 0$. \newline Also, since $U'$ is decreasing and $X^{\hat{\pi}_u}_T \leq I(yH_T) $ it follows $U'(X^{\hat{\pi}_u}_T) \geq U'(I(yH_T)) = yH_T$. This leads to \newline 
$\mathbb{E}[U'(X^{\hat{\pi}_u}_T) (X^{\pi}_T - X^{\hat{\pi}_u}_T)| X^{\hat{\pi}_u}_T \geq K_u ] = - \mathbb{E}[U'(X^{\hat{\pi}_u}_T) (X^{\hat{\pi}_u}_T - X^{\pi}_T )| X^{\hat{\pi}_u}_T \geq K_u ] \newline  \leq  -  \mathbb{E}[yH_T (X^{\hat{\pi}_u}_T - X^{\pi}_T )| X^{\hat{\pi}_u}_T \geq K_u] =   \mathbb{E}[yH_T (X^{\pi}_T - X^{\hat{\pi}_u}_T )| X^{\hat{\pi}_u}_T \geq K_u]$.
\newline So in summary we have \newline
$\mathbb{E}[U(X^{\pi}_T)] - \mathbb{E}[U(X^{\hat{\pi}_u}_T)]  \newline \leq  y \mathbb{E}[H_T(X^{\pi}_T - X^{\hat{\pi}_u}_T )| X^{\hat{\pi}_u}_T \geq K_u] \mathbb{P}[ X^{\hat{\pi}_u}_T \geq K_u]  + y \mathbb{E}[H_T (X^{\pi}_T - X^{\hat{\pi}_u}_T)|     X^{\hat{\pi}_u}_T < K_u] \mathbb{P} [ X^{\hat{\pi}_u}_T < K_u ] \newline = y \mathbb{E}[H_T (X^{\pi}_T - X^{\hat{\pi}_u}_T)] = y(X_0 -X_0) = 0$, \newline
because the budget constraint holds for both strategies. So, $\mathbb{E}[U(X^{\pi}_T)] \leq \mathbb{E}[U(X^{\hat{\pi}_u}_T)]$ for all admissible strategies $\pi$, from which follows statement (4.1)\end{proof}\newpage  In order to solve Problem 3 for the exponential utility function, we state:
\begin{proposition}The optimal terminal wealth for Problem 3 under the utility function is given by \begin{equation} X^{\hat{\pi}_u}_T = \tilde{X}^{\hat{\pi}}_T - \max\{\tilde{X}^{\hat{\pi}}_T - K_u, 0\} \end{equation}
where $\tilde{X}^{\pi}_t$ is the optimal unconstrained wealth process from (2.15), with shadow value  \newline $\tilde{X}^{\hat{\pi}}_0 = (-ln(\dfrac{y}{\alpha}) + rT-\dfrac{\theta^2}{2}T)\dfrac{1}{\alpha} e^{-rT}$ for y >0 such that  $\mathbb{E}[H_T X^{\hat{\pi}}_T] = {X}_0$.\end{proposition}
\begin{proof}
With Lemma 1 and the notation and results from proof of Proposition 4, we get for the exponential function:\newline 
$X^{\hat{\pi}_u}_T = \min\{K_u,I(y H_T)\} = I(y H_T) - \max\{I(y H_T)-K_u, 0\}$ with the same formula for the shadow value. \end{proof}
So, in this case the optimal terminal wealth is identical to the wealth resulting from the optimal strategy and starting value $\tilde{X}_0$ minus the payoff of a call option with strike price $K_u$. This means, that the optimal wealth process can be replicated by the replicating strategy of the call option plus the optimal strategy from the first chapter. Again, we first determine the pricing function corresponding to the payoff of the call function $\max\{\tilde{X}^{\hat{\pi}}_T - K_u, 0\}$:
 \begin{proposition}The pricing function corresponding to the call option with payoff\newline $\max\{\tilde{X}^{\hat{\pi}}_T - K_u, 0\}$ is given by \begin{equation}{c(t, \tilde{X}^{\hat{\pi}}_t)= \Phi (-d_u) (\tilde{X}^{\hat{\pi}}_{t} -K_ue^{-r(T-t)})  + \dfrac{\theta \sqrt{T-t}}{\alpha}  e^{-r(T-t)} \phi(d_u)}\end{equation} \newline with $\Phi (x)$ the cumulative normal distribution function,  $\phi (x)$ its density, and\newline $d_u = d_u(t, \tilde{X}^{\hat{\pi}}_t) =( K_u-  \tilde{X}^{\hat{\pi}}_{t}e^{r(T-t)}) \dfrac{\alpha}{\sqrt{T-t} \theta}$.\end{proposition}
\begin{proof}
By using the same risk neutral valuation arguments and notation as for the proof of Proposition 5, we get for t= 0 \newline  $c(0, \tilde{X}^{\hat{\pi}}_0) = e^{-rT}\mathbb{E}[(\tilde{X}^{\hat{\pi}}_T - K_u) \mathds{1}_{\{Z > d_0\}}]  \newline = \tilde{X}^{\hat{\pi}}_{0} \Phi (-d_0)-  K_u e^{-rT}\Phi (-d_0) + \dfrac{\theta \sqrt{T}}{\alpha}  e^{-rT} \dfrac{1}{\sqrt{2\pi}}   \bigintsss_{d_0}^{\infty}  Z e^{-Z^2/2} dZ$ \newline  = $\Phi (-d_0)[ \tilde{X}^{\hat{\pi}}_{0}-  K_l e^{-rT}]  + \dfrac{\theta \sqrt{T}}{\alpha}  e^{-rT}  \dfrac{1}{\sqrt{2\pi}} e^{-d_0^2/2}$, using that $1- \Phi(d_0)= \Phi (-d_0)$. The statement then follows from expanding for any $t \in$ [0,$T$].\end{proof}
Next, we determine the replicating portfolio. We state 
\begin{proposition} The replicating portfolio  of the pricing function (4.4) is given by the strategy 
\begin{equation} \tilde{\pi}_c(t, \tilde{X}^{\hat{\pi}}_t) = \dfrac{\Phi(-d_u)}{\sigma \sqrt{T-t}(\phi(d_u) -\Phi(-d_u)d_u)}, 
 \end{equation} with $d_u$  as  in Proposition 9.\end{proposition}
\begin{proof}
First take the partial derivatives of $c(t, \tilde{X}^{\hat{\pi}}_t)$:\newline
$c_t=   e^{-r(T-t)}[-\Phi(-d_u)rK_u + \dfrac{\theta}{\alpha} \sqrt{T-t} \phi(d) (r + \dfrac{1}{2(T-t)})] , \newline 
c_x = \Phi (-d_u), \newline
c_{xx} = \phi(-d_u) \dfrac{\alpha}{\sqrt{T-t} \theta} e^{r(T-t)}$
using $\tilde{X}^{\hat{\pi}}_t-K_ue^{-r(T-t)} = -d_u e^{-r(T-t)} \sqrt{T-t} \dfrac{\theta}{\alpha}$ and cancellations. \newline
Since $\tilde{X}^{\hat{\pi}}_t$ is an Ito drift diffusion process, we can apply Ito's Lemma and get \newline $dc = \{ c_t + c_x (r\tilde{X}^{\hat{\pi}}_t + \dfrac{\theta^2}{\alpha}e^{-r(T-t)}) + c_{xx} \dfrac{\theta^2}{2 \alpha^2}e^{-2r(T-t)}\}dt + c_x \dfrac{\theta}{\alpha} e^{-r(T-t)}dW_t$. 
\newline
In order to satisfy the dynamics of the wealth equation,  $dc=(rc+ \tilde{\pi}_c \theta \sigma)dt + \sigma \tilde{\pi}_c c dW(t)$ , it needs to hold \newline$ c_x \dfrac{\theta}{\alpha} e^{-r(T-t)}$ = $\sigma \tilde{\pi}_c c$, from which follows equation (4.5).
\end{proof}
Again, combining (4.5) with the optimal strategy gives the wanted result. \newline
 \begin{proposition} An optimal strategy for Problem 3 is given by the amount to be invested  at t 
 \begin{equation}\hat{\pi}_u(t,\tilde{X}^{\hat{\pi}}_t)= \dfrac{\theta}{\alpha \sigma}e^{-r(T-t)} - c(t,\tilde{X}^{\hat{\pi}}_t) \dfrac{\Phi(-d_u)}{\sigma \sqrt{T-t}(\phi(d_u) -\Phi(-d_u)d_u)} \end{equation} 
for the shadow wealth process $\tilde{X}^{\hat{\pi}}_t$ and $d_u = d_u(t, \tilde{X}^{\hat{\pi}}_t) \text{ as in Proposition 9}$. \end{proposition} \begin{proof}
By definition of $\hat{\pi}_u(t,\tilde{X}^{\hat{\pi}}_t) = \hat{\pi}_t \tilde{X}^{\hat{\pi}}_t - \tilde{\pi}_c(t,\tilde{X}^{\hat{\pi}}_t) c(t,\tilde{X}^{\hat{\pi}}_t) $ the resulting wealth process is 
$X^{\hat{\pi}_u}_t = \tilde{X}^{\hat{\pi}}_t - c(t, \tilde{X}^{\hat{\pi}}_t)$  with 
$X^{\hat{\pi}_u}_T = \tilde{X}^{\hat{\pi}}_T - \max\{\tilde{X}^{\hat{\pi}}_T -K_u ,0\}$, so (4.3) holds. \end{proof}

\subsubsection{The $K_l$-$K_u$-Strategy}
If we combine the upper and lower constraints on terminal wealth, we get an intuitive result, which is summarized in the following Proposition (using the previous notation).
\newline
\begin{proposition}
The strategy determined by the absolute investment in the stock at t, 
\begin{equation}\hat{\pi}_{l,u}(t,\tilde{X}^{\hat{\pi}}_t)=  \hat{\pi}_t \tilde{X}^{\hat{\pi}}_t +\tilde{\pi}_p(t,\tilde{X}^{\hat{\pi}}_t)  p(t,\tilde{X}^{\hat{\pi}}_t) - \tilde{\pi}_c(t,\tilde{X}^{\hat{\pi}}_t) c(t,\tilde{X}^{\hat{\pi}}_t), \end{equation} provided that ${X}_0$ =  $\mathbb{E}[H_T X^{\hat{\pi}_{l,u}}_T]$, 
is a solution to \begin{equation} \mathbb{E}[U(X^{\hat{\pi}_{l,u}}_T)] = \sup\limits_{\pi\in \mathcal{A}}\mathbb{E}[U(X^{\pi}_T)]\ \text{and} \ K_l\leq X^{\hat{\pi}_u}_T \leq K_u \text{ a.s. for }K_l < K_u .\end{equation} \end{proposition}
\begin{proof} By definition of $\hat{\pi}$, $\tilde{\pi}_p$ and $\tilde{\pi}_c$, the strategy (4.7) results in the terminal wealth \newline $\tilde{X}^{\hat{\pi}}_T +   \max\{K_l- \tilde{X}^{\hat{\pi}}_T ,0\}   - \max\{\tilde{X}^{\hat{\pi}}_T - K_u, 0\}$ , which both satisfies (3.2) and(4.2) (note that $K_u > K_l$) and hence solves both Problem 2 and Problem 3. 
\end{proof}

This implies that the shadow initial wealth $\tilde{X}_0$ will increase if we add an upper constraint to the (pre-existing) lower constraint on terminal wealth, because an additional amount of 'money' is available due to selling a call-option. The effect of the increase will depend on the price of the call-option, which in turn is determined by the strike price $K_u$. Before we look at this effect more closely, we analyse the behaviour of the $K_u$-strategy, the optimal strategy that only faces an upper constraint, separately.

\subsection{Analysis of the $K_l$-$K_u$-Strategy }
In contrary to the case of a lower constraint, now a call option is involved in the optimal strategy, whose value is increasing with wealth. At the same time it also becomes very sensitive and as a consequence this will require to short-sell large amounts of stocks in very short time. The good news is that with increasing $K_l$ and  $K_u$ an almost sure convergence to $\hat{\pi}_{l}$ can be observed, since the frequency of these extremes diminishes. We will also find a criterion for a 'balance' between $K_u$ and $K_l$.
\subsubsection{Analysis of the $K_l$-Strategy}
Since the payoffs of put and call options are antagonistic, so will be the behaviour of the strategies linked to them. This means for example, if the value of the shadow wealth rises, the put-option loses value and the replicating strategy would increase investment, equivalently the call-option would gain value, so its replicating strategy would reduce investment. It is therefore not surprising that the difference to the optimal strategy  due to an upper constraint looks like the 'inverted' difference to the optimal strategy due to a lower constraint (below, an illustration of the case $K_u$ = 2'000 and $\mu$ = 0.03, $r$ = 0.01, $\sigma$ = 0.1 and $T$=20).

\begin{figure}[H]
 \centering \includegraphics[width=80mm]{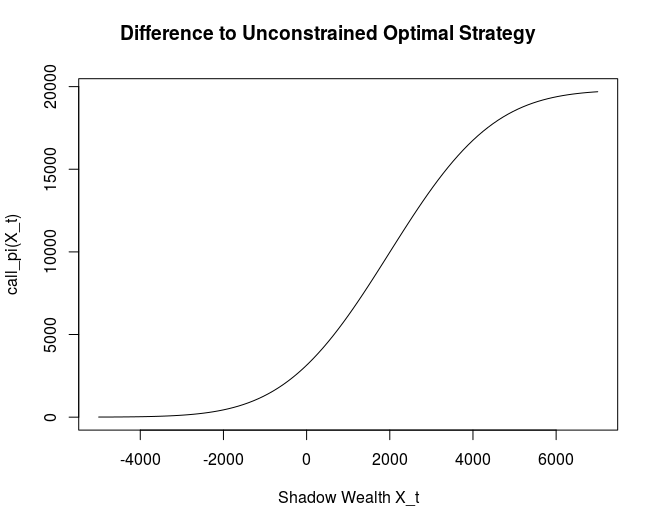}
 \caption{Difference $\hat{\pi}$-$\hat{\pi}_u$ in function of current shadow wealth (amount invested at t=19)}
\end{figure}
This can also be seen in the formula, for example it holds that:
\begin{center}-$\tilde{\pi}_c (d)$=$\dfrac{-\Phi(-d)}{\sigma \sqrt{T-t} (\phi(d)- \Phi(-d)d)} =  \dfrac{-\Phi(-d)}{\sigma \sqrt{T-t} (\phi (-d) + (-d) \Phi(d)}$ =$-\tilde{\pi}_p (-d)\ \forall\ d \in\mathbb{R}.$\end{center} 

A more detailed illustration and comparison of these processes can be found in the Appendix. Here, we will only briefly outline the consequences of this strategy on the wealth process. It can be observed that the behaviour over time is similar for both constrained strategies and it shows a peak on a high values of t (in the case of the Appendix it was at t=17, in the one from the previous chapter it was at t= 19). Combining this with the increase of the strategy for high values of $\tilde{X}^{\hat{\pi}}_t$ as seen in Figure 16, these effects add up for high $\tilde{X}^{\hat{\pi}}_t$ and t. This leads to very small (i.e. negative) values for $\hat{\pi}_l$, since the positive investment from the optimal strategy can't compensate this extreme effect. 
 \newline 
In the example below it can be seen particularely well how the time and the shadow wealth influence the strategy and the wealth resulting from it. 
 \begin{center}
\begin{figure}[H] 
 \begin{minipage}{0.65\linewidth}
\includegraphics[width=77mm]{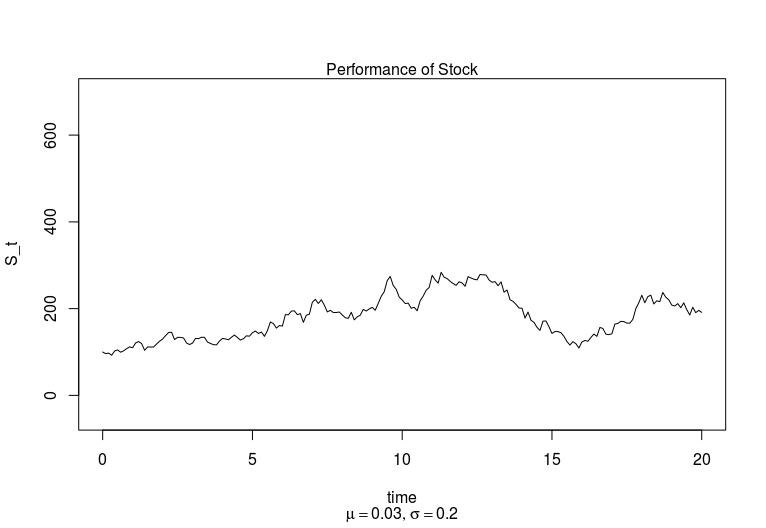} \end{minipage} \begin{minipage}{0.65\linewidth} \includegraphics[width=77mm]{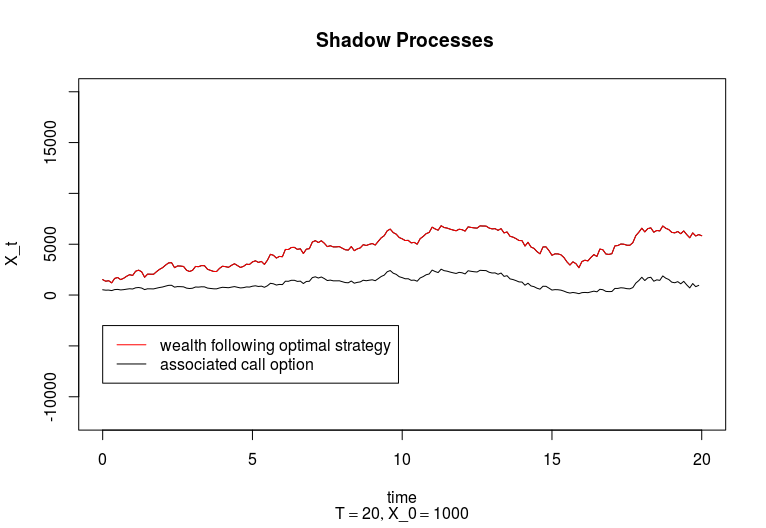}\end{minipage}  \begin{minipage}{0.65\linewidth}\includegraphics[width=77mm]{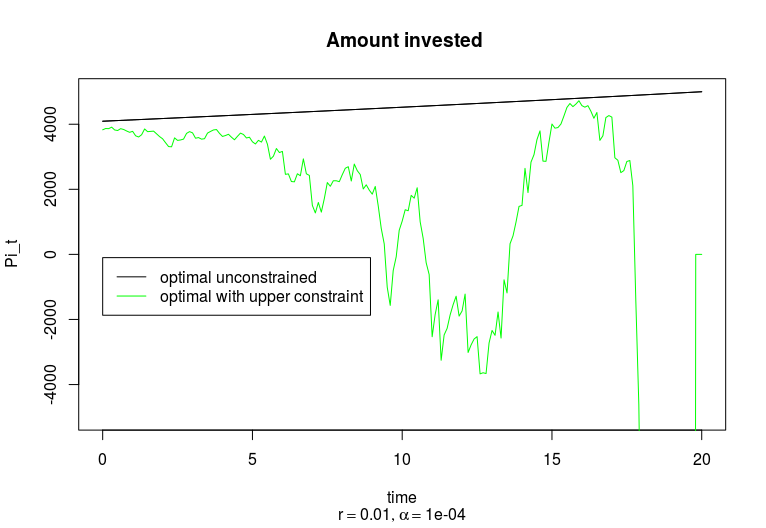} \end{minipage} \begin{minipage}{0.65\linewidth}\includegraphics[width=77mm]{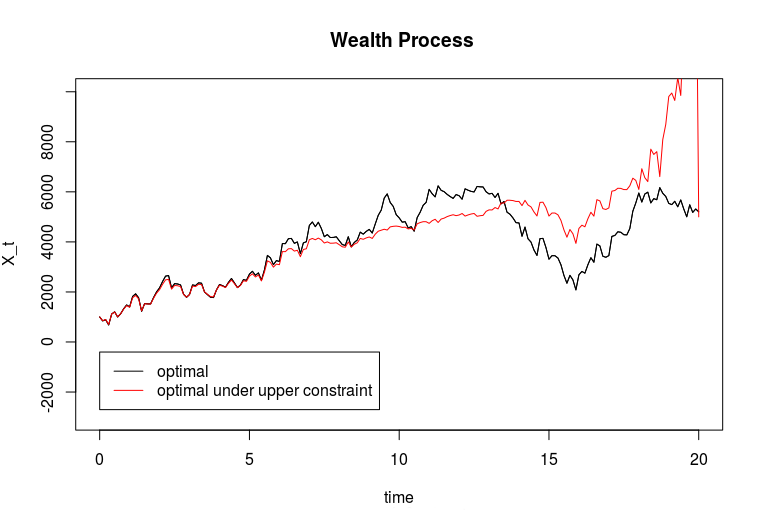} \end{minipage}
\caption{$\hat{\pi}_u$ becomes negative for higher wealth and nearing maturity ($K_u$ = 1'250).}
\end{figure} \end{center} In the first time sequence from t=0 to circa t= 10 the stock price, the shadow wealth and the associated call option are moderately increasing. However, the value of shadow wealth is increasing from circa 1'000 to circa 5'000, which has an overproportional effect on the strategy, as seen in Figure 16. This is why the moderate fluctations of the shadow wealth produce movements in the strategy that are increasingly more accentuated. To this adds the effect of time, which can be seen later on. Even though the movements in the shadow wealth are not increasingly extreme, the strategy $\hat{\pi}_l$ is increasingly sensitive (even around low values of $\tilde{X}^{\hat{\pi}}_t$ at circa $t$=17). This being said, the strategy seems to work well, the strong negative investments in stocks towards maturity pay off when the stock is even just slightly decreasing. In fact, since the goal was to reach a target wealth  $K_u$ or below, it might perform too well. This is could be the reason why investments are drastically reduced to zero in the last periods before maturity.\newline
In practice, the strong negative investments can be interpreted as short-selling and borrowing and could lead to a few problems. First, transaction costs might become a considerable factor, as the strategy requires to sell and buy a large number of stocks within short time. Then, the rebalancing will happen in discrete time, opposite to the assumptions of the theory (in the simulation, the step width was $h$=1/10). Since the strategy is very sensitive, the output might differ more from the theoretical results (or possibly less for shorter $h$) and further investigations of the impact of $h$ would be useful before implementing this strategy in practice.
Finally, in the case of poor stock performance and loss of value, the terminal wealth itself will be negative.  Without a lower constraint, this seems not only to lead to an increased probability of debts, but in particular an increased probability of very hight depts. However, this problem could be easily avoided by setting a lower constraint for terminal wealth. 
Looking at the effect of the upper constraint, the impact of the call-option-related part of $\hat{\pi}_{l,u}$ can be reduced by setting a higher $K_u$. However, it might be difficult to explain the choice of a very high constraint, as the missed potential might be minimal. A reasonable choice of the constraints will also be discussed in the next section. For a first assertion, we have set $K_u$ = 3'000 in the next scenario.
 \begin{center}
\begin{figure}[H] 
 \begin{minipage}{0.65\linewidth}
\includegraphics[width=77mm]{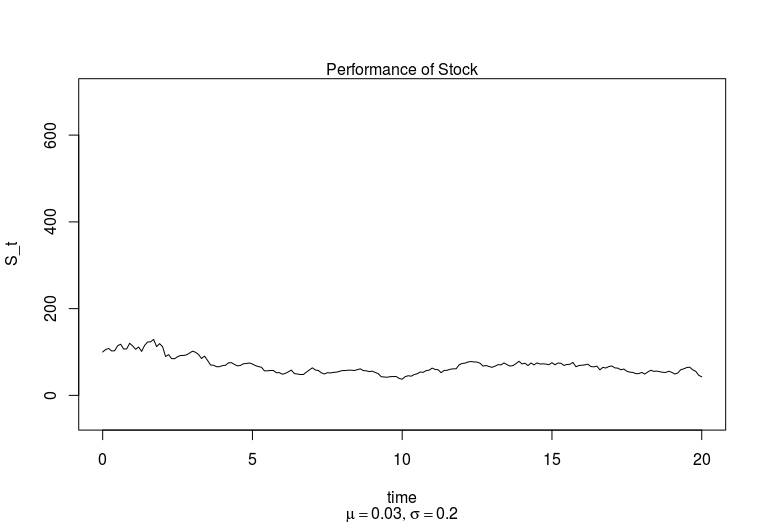} \end{minipage} \begin{minipage}{0.65\linewidth} \includegraphics[width=77mm]{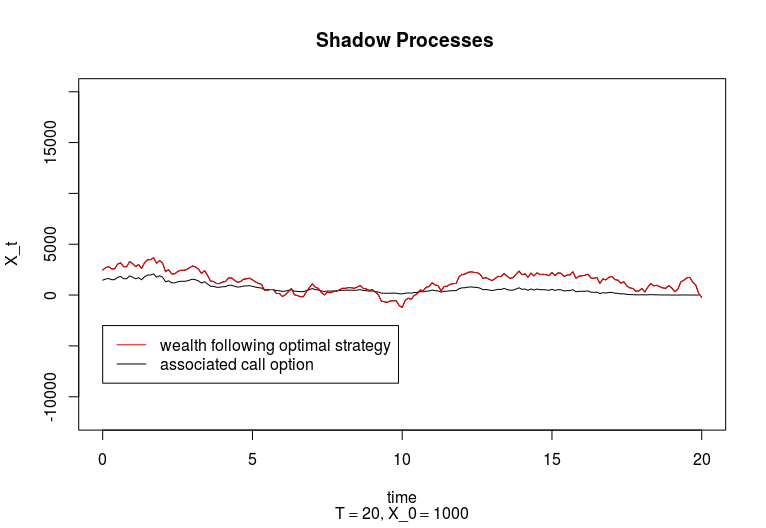}\end{minipage}  \begin{minipage}{0.65\linewidth}\includegraphics[width=77mm]{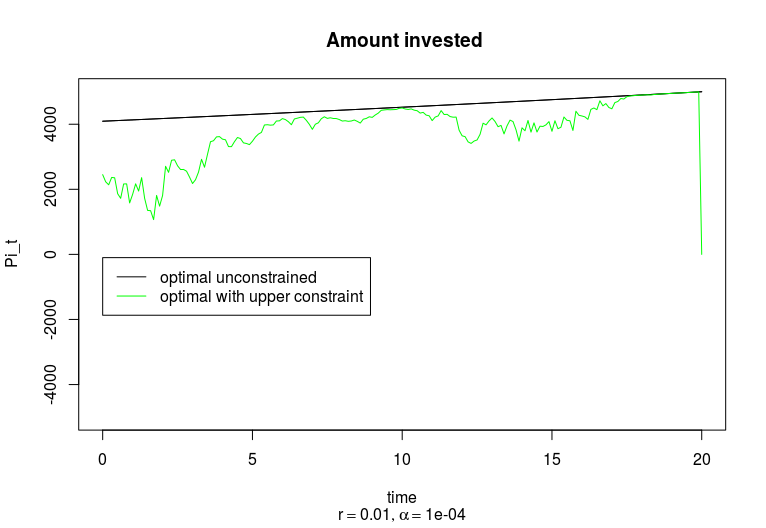} \end{minipage} \begin{minipage}{0.65\linewidth}\includegraphics[width=77mm]{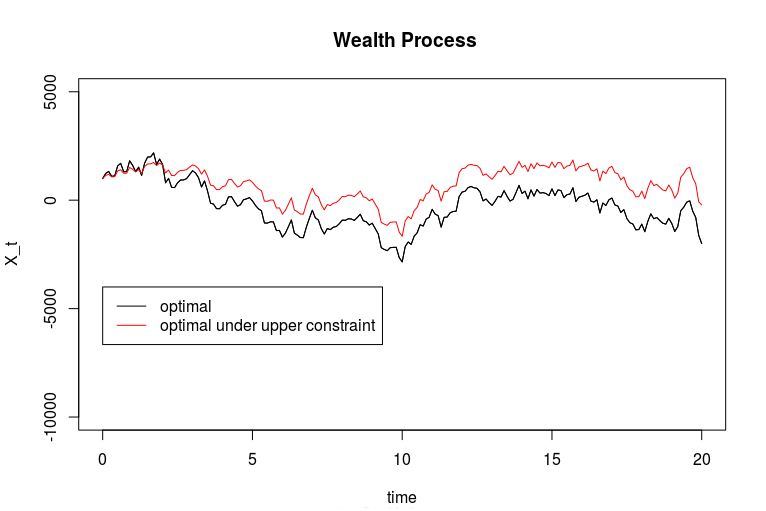} \end{minipage}
\caption{Poor stock performance leads to higher investment ($K_u$ = 3'000).}
\end{figure}\end{center}

Here, the effects we have seen before are also playing a role, but overall their impact is reduced. Towards the end of the investment horizon, we even observe a convergence of the optimal $K_u$-constrained strategy to the optimal unconstrained one. This is, because of two reasons. First, a higher $K_u$ generally leads to a higher probability that the shadow wealth falls under it, which would result in a payoff of zero for the call-option. Second, the performance of the stock in this case is rather poor, and so is the performance of shadow wealth. Hence, the price of the option is near zero, especially for times near to maturity. This again shows the importance of the value of shadow wealth to the strategy: looking back at Figure 16, which illustrated a similar case, we see that the difference of investment to the optimal unconstrained strategy for a shadow wealth around 0 is a factor 10 less than the one from wealth around 5'000. Equivalently, a 'bad' stock performance leads to an increased investment by $\hat{\pi}_u$, which makes sense, because in this scenario the target wealth then would be less likely to be surpassed. 
\newline 
It would be interesting to further investigate the behaviour of $\hat{\pi}_u$. But since the investment-constrained $K_u$-strategy, as well as the combination with the lower constraint, the $K_l$- $K_u$-strategy, might reduce its sensitivity and since ultimately these strategies are of more practical interest to us, we would rather have a closer look at these two. 
\subsubsection{Analysis of the $K_l$-$K_u$-Strategy}
We now combine the upper and the lower bound for terminal wealth and measure the impact of these constraints on the strategy. \newline First, note that $K_l$ and $K_u$ determine the weight of $\tilde{\pi}_c c(\tilde{X}^{\hat{\pi}}_t, t)$ and $\tilde{\pi}_p p(\tilde{X}^{\hat{\pi}}_t, t)$ on the total optimal strategy $\hat{\pi}_{l,u}$. Also, they are critical for the computation of $\tilde{X}_0$. We therefore first establish a point of orientation for comparison of $K_l$ and $K_u$. \newline
\begin{proposition}
For an upper bound $K_u$ and a lower bound $K_l$ for terminal wealth, the initial wealth $X_0$ and a solution for shadow wealth for the  optimal constrained strategy from Proposition 12 $\tilde{X_0}$, it holds: 
\begin{equation}
K_u + K_l = 2X_0e^{rT} \implies \tilde{X}_0 = X_0.
\end{equation} .\end{proposition}

\newpage

\begin{proof}
First note that $K_u = 2X_0e^{rT} -K_l$ implies $d_u$ =  $K_u -X_0e^{rT}$ = -$(K_l - X_0 e^{rT})$ = -$d_l$. \newline Plugging this in the formula of the call option price gives \newline c(0, $X_0$) \newline =  $\Phi(-d_u) ( X_0 - K_u e^{rT}) + \dfrac{\theta}{\alpha} \sqrt{T} e^{-rT}\phi(d_u)$\newline
 =  $\Phi(d_l) ( X_0 - (2X_0 e^{rT}-K_l) e^{rT}) + \dfrac{\theta}{\alpha} \sqrt{T} e^{-rT}\phi(-d_l)$\newline
 =  $\Phi(d_l) ( K_l e^{-rT}-X_0) + \dfrac{\theta}{\alpha} \sqrt{T} e^{-rT}\phi(d_l)$\newline
 = p(0,$X_0$)
 \newline
 So, the prices at t=0 of the options are the same, hence they set each other off and it holds: \newline
 $X_0$ + p (0,$X_0$) - c(0,$X_0$) = $X_0$, hence $X_0$ is a solution for shadow wealth. 
\end{proof} We will see later that if the described relation between $K_u$ and $K_l$ holds, the distribution of terminal wealth on [$K_l$,$K_u$] will be identical to the distribution of the optimal unconstrained terminal wealth. Further, it can easily be seen that a decrease of $K_l$ and an increase of $K_u$ at the same time by the same amount will have no impact on the distribution (except for the limits of its definition area).
\newline
However, if $K_l + K_u > 2X_0e^{rT}$, then $\tilde{X_0}$< $X_0$ and the same holds for the inequality in the other sense. To see how $\tilde{X}_0$ changes for different ($K_l,K_u$)-combinations, refer to the illustration below, where $X_0$=1'000 (in a setting with $r$=0.01, $\mu$=0.03, $\sigma$=0.2).
\begin{figure}[H] 
\centering \includegraphics[width=150mm]{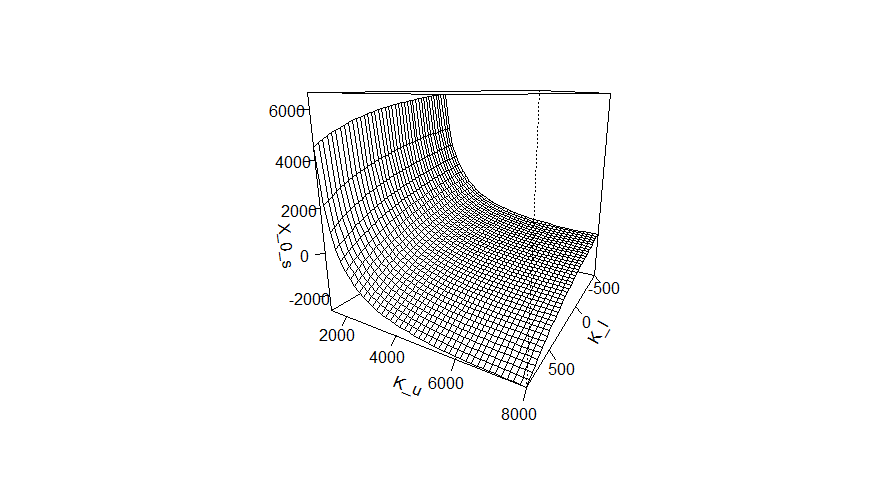} \caption{Shadow initial wealth in function of $K_l$ and $K_u$ ($X_0$ = 1'000)}
\end{figure}
We will now have a closer look on the qualitative implications on the strategies if $K_l$ and  $K_u$ vary and are particularly interested in reducing the sensitivity of $\hat{\pi}_{l,u}$ towards $\tilde{X}^{\hat{\pi}}_t$. The following scenarios are simulated for a sample of 20 paths and $T$=20, $r$=0.01, $\sigma$ = 0.2 and $\mu$ =0.03. $K_l$ is fixed at 0 and for a better comparison, the scales in the diagrams weren't adapted. \newline
\begin{figure}[H] {\subcaption{$K_u$=4'000}
 \begin{minipage}{0.44\linewidth}
\includegraphics[width=75mm]{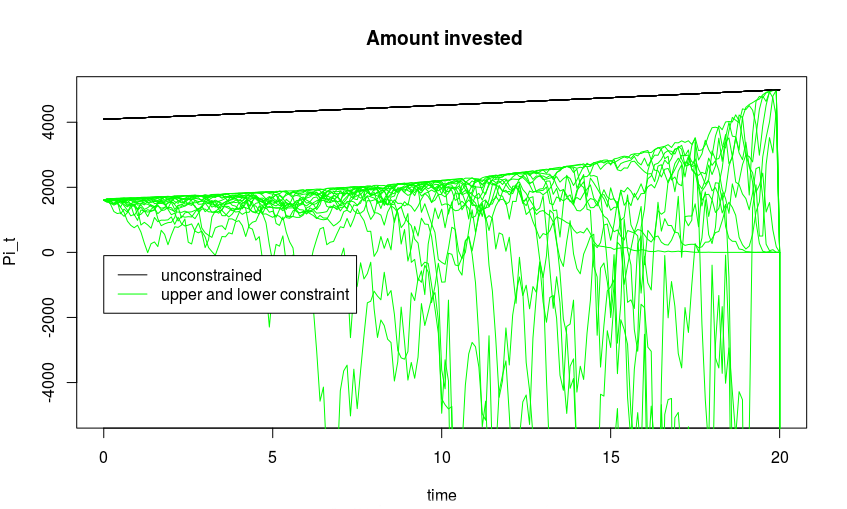} \end{minipage} \begin{minipage}{0.44\linewidth} \includegraphics[width=75mm]{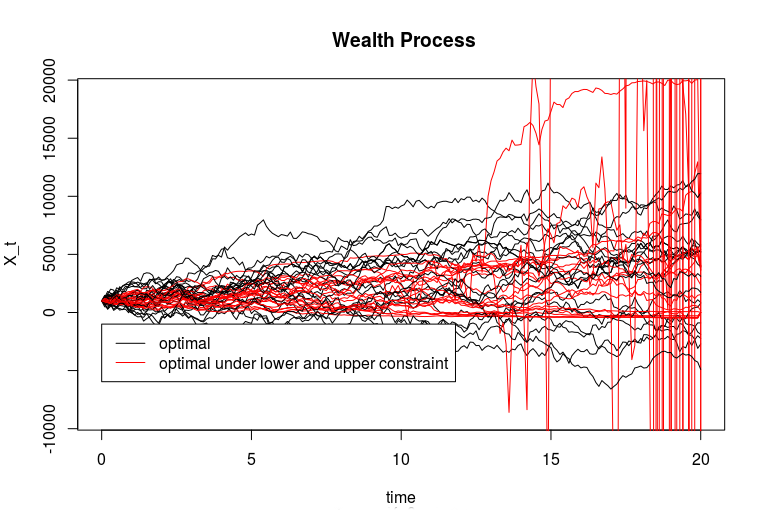}\end{minipage} \subcaption{$K_u$=2'443} 
 \begin{minipage}{0.44\linewidth}
 \includegraphics[width=75mm]{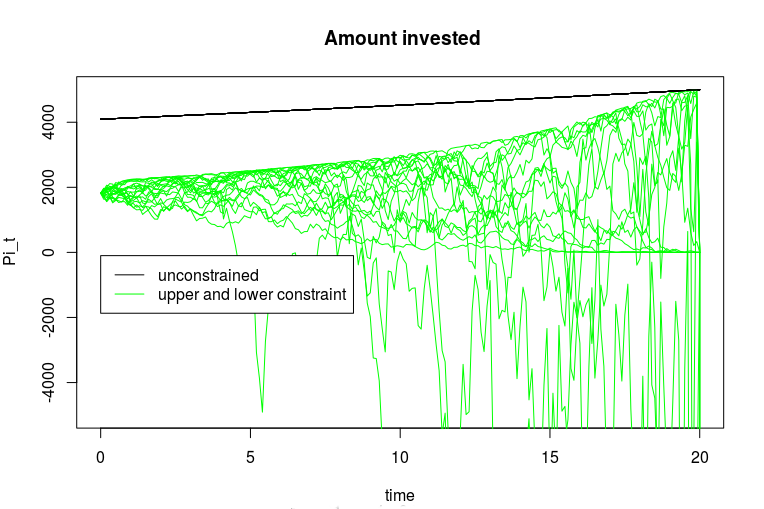} \end{minipage} \begin{minipage}{0.45\linewidth}\includegraphics[width=78mm]{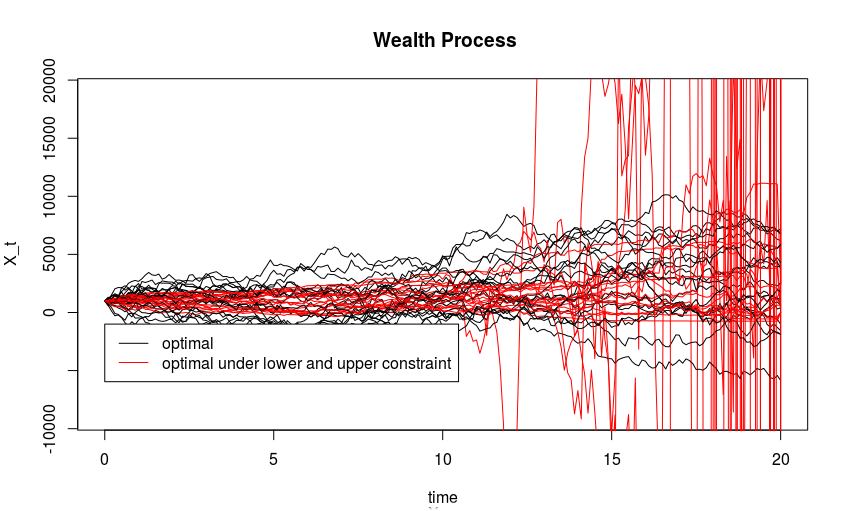} \end{minipage} \subcaption{$K_u$=1'500} 
 \begin{minipage}{0.5\linewidth} 
 \includegraphics[width=76mm]{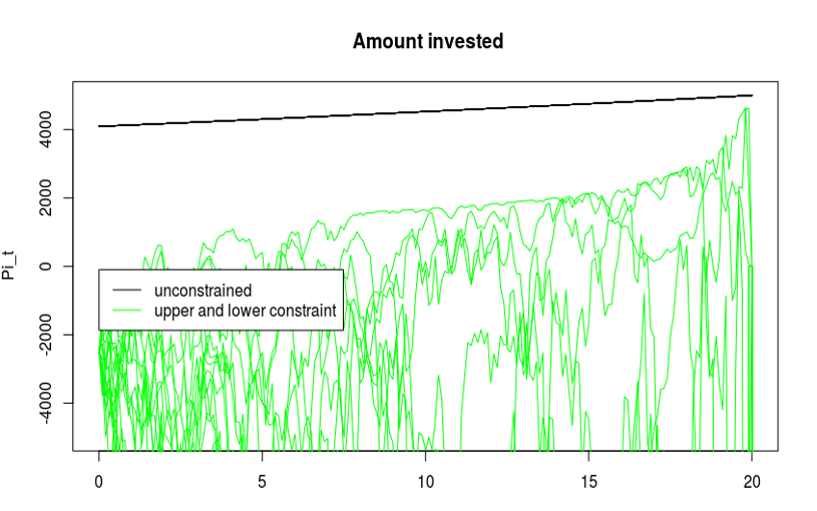} \end{minipage}  \begin{minipage}{0.45\linewidth}\includegraphics[width=78mm]{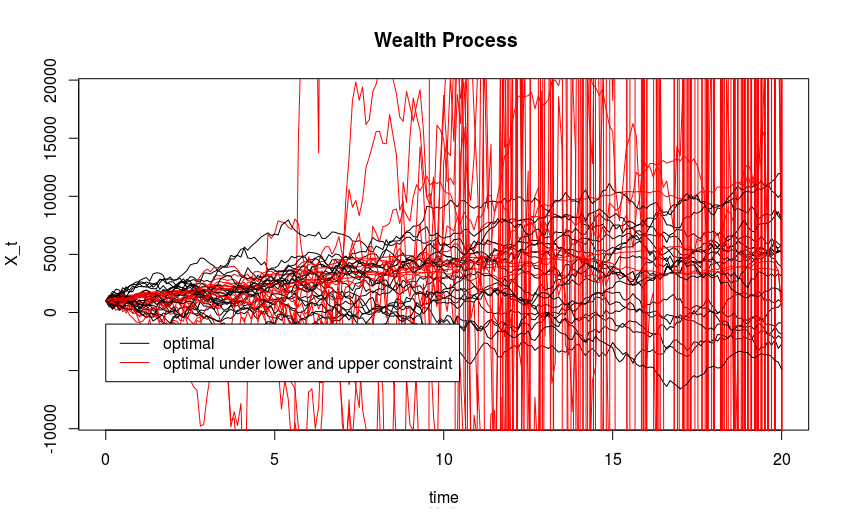}\end{minipage} }
\caption{Amount invested over time and wealth processes for $\hat{\pi}_{l,u}$ (different $K_u$)}
\end{figure}
It can be observed, that the impact of the part of the strategy related to the upper constraint, $\hat{\pi}_u c(0,\tilde{X}^{\hat{\pi}}_t)$ is still strong, and its extremely low investments dominate the total strategy and the wealth process. However, a difference between different upper constraints can be found. In the first scenario, $K_u$ = 4'000 is rather large (but it lies between the 50$\%$ and 75$\%$ quantile of the optimal strategy under the lower bound $K_l$=0, so it still represents a significant reduction of upside potential) and so the weight of the upper constrained-part of the strategy is smaller. This can be seen, since fewer strategy-paths follow extremely low investments and there is a concentration on an upper line converging to the optimal strategy, which seems to be an upper boundary for the investments. Also, looking at the terminal wealth, two points of concentration can be spotted: $K_l$ and $K_u$, indicating that the terminal wealth distribution now has two probability mass points. \newline
The scenario in the middle represents the setting from Proposition 13, where $X_0$ = $\tilde{X}_0$. It is, in a sense, a balanced mix of upper and lower constraints, and so the negative excesses of investment are more pronounced. From the paths of the optimal wealth processes it can be guessed that the stock performance is rather worse than in the first scenario. Since it is known that $\hat{\pi}_u$ is sensitive especially to high values of $X_t$, the investments of this scenario would be expected to be even smaller for other simulations. 
\newline In the third case, the upper constraint is very low, which leads, as expected, to an even higher impact of $\tilde{\pi}_c c(0,\tilde{X}^{\hat{\pi}}_t$) on $\hat{\pi}_{l,u}$. Here, from $t$=0 onwards, most of the investments are negative. 
 
The case for different $K_l$ is similar: For lower $K_l$ the $\hat{\pi}_{l,u}$ is closer to $\hat{\pi}_u$ and for higher $K_l$ it is closer to $\hat{\pi}_u$. An illustration of this can be found in the Appendix. 
If we combine these two results, and set a high lower constraint and a high upper constraint, the total strategy indeed looks more acceptable. For example, in the case below, we have set $K_l$ = 800 and $K_u$ = 5'000 for an initial investment of 1'000. \newline

\begin{figure} [H]\begin{minipage}{0.53\linewidth}
 \includegraphics[width=76mm]{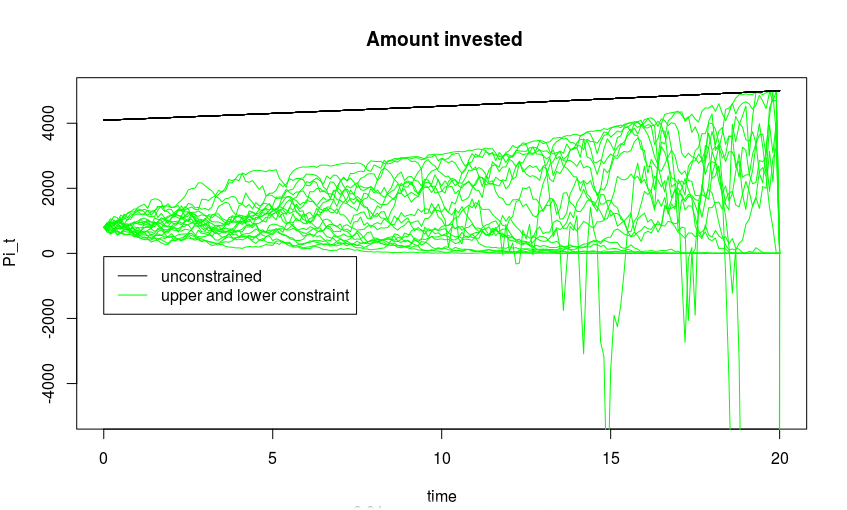} \end{minipage}  \begin{minipage}{0.53\linewidth}\includegraphics[width=76mm]{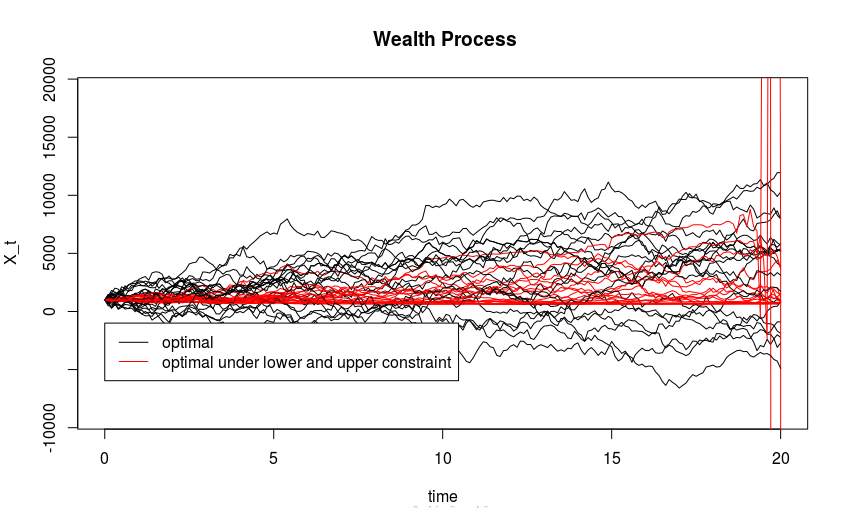}\end{minipage} 
\caption{Amount invested over time and wealth processes for $\hat{\pi}_{l,u}$ (high $K_l$ and $K_u$)}
\end{figure}

\subsection{Distribution of Terminal Wealth}
In this section we will see, that the effect of adding lower and upper bounds to the optimal strategy on terminal wealth distribution can be described as a shift of the quantiles by the future value of the  difference of initial wealth and initial shadow wealth. 
\newline \newline
Let us first calculate the theoretical distribution of terminal wealth, following the optimal strategy subject to both a lower constraint $K_l$ and an upper constraint $K_u$. \newline
$\mathbb{P}[X^{\hat{\pi}_{l,u}}_T \le x]\newline  =  \mathbb{P}[\tilde{X}^{\hat{\pi}}_T - \max\{ \tilde{X}^{\hat{\pi}}_T -K_u ,0\}  + \max\{K_l-\tilde{X}^{\hat{\pi}}_T ,0\} \le x] $ by Proposition (12) \newline
= $\mathbb{P}[ \{ \tilde{X}^{\hat{\pi}}_T | \tilde{X}^{\hat{\pi}}_T \geq K_u \} \cap \{\tilde{X}^{\hat{\pi}}_T | K_u \leq x \}] + \mathbb{P}[\{\tilde{X}^{\hat{\pi}}_T| K_l \leq \tilde{X}^{\hat{\pi}}_T < K_u \} \cap \{\tilde{X}^{\hat{\pi}}_T| \tilde{X}^{\hat{\pi}}_T \leq x\}]$ \newline + $\mathbb{P}[\{ \tilde{X}^{\hat{\pi}}_T | \tilde{X}^{\hat{\pi}}_T < K_l \} \cap \{\tilde{X}^{\hat{\pi}}_T| K_l \leq x \} ]$, \newline
since $K_l <K_u$ and $\{\tilde{X}^{\hat{\pi}}_T | \tilde{X}^{\hat{\pi}}_T \geq K_u \} \cup \{ \tilde{X}^{\hat{\pi}}_T | K_l \leq \tilde{X}^{\hat{\pi}}_T  < K_u \} \cup \{\tilde{X}^{\hat{\pi}}_T | \tilde{X}^{\hat{\pi}}_T < K_l\}$ is the union of disjoint sets and has probability 1. \newline
For $K_l \leq x < K_u$ we then  get: \newline
$\mathbb{P}[X^{\hat{\pi}_{l,u}}_T \le x]$ \newline 
 =  $\mathbb{P} [ \{ \tilde{X}^{\hat{\pi}}_T | K_l \leq \tilde{X}^{\hat{\pi}}_T \leq x \}] + \mathbb{P}[\{\tilde{X}^{\hat{\pi}}_T| \tilde{X}^{\hat{\pi}}_T < K_l \}]$ \newline = $\mathbb{P}[\tilde{X}^{\hat{\pi}}_T \leq x \}] - \mathbb{P}[\tilde{X}^{\hat{\pi}}_T < K_l] + \mathbb{P}[\tilde{X}^{\hat{\pi}}_T < K_l]$ \newline = $\mathbb{P}[\tilde{X}^{\hat{\pi}}_T \leq x]$ ,\newline
and it is easy to see that $\mathbb{P}[X^{\hat{\pi}_{l,u}}_T \le x] = 1$ for $x \geq  K_u$  and $\mathbb{P}[X^{\hat{\pi}_{l,u}}_T \le x] = 0$ for $x <  K_u$ . 
\newline In summary this leads to: \newline \newline
$\mathbb{P}[X^{\hat{\pi}_{l,u}}_T \le x]$ =  $\begin{cases}{\mathbb{P}[\tilde{X}^{\hat{\pi}}_T \leq x]}&\text{if $K_l \leq x < K_u $}\\{1}&\text{if x $\geq K_u$} \\{0}&\text{if x < $K_l $}   \end{cases} $
\newline \newline
So, again, we will find probability mass points at the boundaries $K_l$ and $K_u$ and the probability of reaching values outside these limits is zero. In between, the terminal wealth distribution follows the cumulative distribution of the shadow terminal wealth, i.e normal distribution with $\mathbb{E}_0[\tilde{X}^{\hat{\pi}}_T] = \tilde{X}^{\hat{\pi}}_0 e^{rT} + T\dfrac{\theta^2}{\alpha}$ and Var$(\tilde{X}^{\hat{\pi}}_T) = \dfrac{\theta^2}{\alpha^2}T$.
\newline From this follows: 
\begin{proposition}
The shift of quantiles of terminal wealth that results from introducing upper and lower constaints on terminal wealth is given by: \newline
\begin{equation}
\tilde{\mathcal{Q}}_p = \mathcal{Q}_p + (\tilde{X}_0 -X_0) e^{rT} ,
\end{equation} 
where $\mathcal{Q}_p$ is the p-quantile of the distribution of $X^{\hat{\pi}}_T$, $\tilde{\mathcal{Q}}_p$ is the p-quantile of the distribution of $X^{\hat{\pi}_{l,u}}_T$, $X_0$ is the initial wealth, and $\tilde{X_0}$ the initial shadow wealth from Proposition 12. \end{proposition}

\begin{proof}
To see the statement, first recall that the original unconstrained terminal wealth distribution has normal distribution. For better readability we introduce the notation \newline $E := \mathbb{E}[X^{\hat{\pi}}_t] = X_0 e^{rT} + T \dfrac{T \theta^2}{\alpha}$ and $V := \dfrac{\theta}{\alpha}\sqrt{T}$. \newline We then have by the definition of the quantiles: \newline
$\mathcal{Q}_p$\newline  = $\inf\{ z| \mathbb{P}[X^{\hat{\pi}}_T \leq z] \geq p\}$ \newline
 = $\inf\{ z | \mathbb{P}[\dfrac{X^{\hat{\pi}}_T -E}{V} \leq\dfrac{z -E}{V}] \geq p\}$ \newline
= $\inf\{ z | \mathbb{P}[Z_T     \leq \dfrac{z -E}{V}] \geq p\}$  with $Z_T :=\dfrac{X^{\hat{\pi}}_T -E}{V} \sim \mathcal{N} $(0,1) \newline
= $\inf\{ z | \Phi (\dfrac{z -E}{V}) \geq p\}$ 
and so  $\Phi(\dfrac{\mathcal{Q}_p -E}{V}) = p$, because $\Phi$ is (right) continuous.
 \newline We then have ${\mathcal{Q}_p}$ =  $\Phi^{-1}(p)V + E$ , where $\Phi^{-1}(p)$ is the p-quantile of the normal distribution. \newline
 Applying the same procedure to  $X^{\hat{\pi}_{l,u}}_T \sim \mathcal{N} (\tilde{E}, \tilde{V})$ with $\tilde{E} = \tilde{X}_0 e^{rT} + T\dfrac{\theta^2}{\alpha}$ and $\tilde{V} = V$ leads to: \newline
$ \tilde{\mathcal{Q}}_p = \Phi^{-1}(p) \tilde{V} + \tilde{E}$ \newline = $\Phi^{-1}(p)V + \tilde{X}_0 e^{rT} + T\dfrac{\theta^2}{\alpha}$ = $\Phi^{-1}(p)V +\tilde{X}_0 e^{rT}  + X_0 e^{rT}+ T\dfrac{\theta^2}{\alpha} - X_0 e^{rT}$ \newline = $\mathcal{Q}_p + (\tilde{X}_0- X_0)e^{rT}$ .
\end{proof}

This result has an intuitive interpretation. For example, if we only consider an upper constraint on terminal wealth, then $\tilde{X}_0 = X_0 + c(0,\tilde{X}_0)$ and so the quantiles are shifted exactly by the future value of the price of the call option at t=0. In other words: the results of this strategy are equivalent to selling a call option (on an optimal shadow process) and putting the received money in a bank account. For the lower constraint, the effect is analogous. In this sense, the constained strategy is somewhat trivial.\newline
Combining this result with Figure 19, where the effect of the upper and lower constraints on the shadow wealth is illustrated, we can assess their impact on the terminal wealth distribution. \newline In particular, giving up significant upside potential by lowering $K_u$ leads to a great positive shift of quantiles. This is also reflected in Table 13, which shows the resulting empirical quantiles for a fixed lower constraint $K_l$ = 0,  a market given by $\sigma$= 0.2, r = 0.01, $\mu$ = 0.03 and varying upper constraints(sample size 400).
For example, setting $K_u$ as low as 1'500 leads to a shadow wealth almost 4 times higher than the original initial wealth. Hence, the quantiles are lifted significantly, by $(3'901-1'000)e^{0.01\cdot 20} = 3'543$. However, since the upper constraint is so low, this can only by seen in $\mathcal{Q}_{0.10} (X_T)$.  Note that $K_u$ = 2'443 is the case where the constraints set each other off in the sense that the distribution of terminal wealth between the boundaries is identical to the unconstrained one (here, this can be seen at $\mathcal{Q}_{0.25}(X_T))$.

\begin{table}[H]
\begin{tabular}{|l| c|c|c|c|c|c|}
\hline
 $K_u$ & $\tilde{X_0}$ & $\mathcal{Q}_{0.10}$($X_T$)& $\mathcal{Q}_{0.25}$($X_T$) &   $\mathcal{Q}_{0.50}$($X_T$) &  $\mathcal{Q}_{0.75}$($X_T$) &  $\mathcal{Q}_{0.95}$($X_T$)   \\ \hline 
 \cellcolor{OldLace}unconstrained &  \cellcolor{OldLace}1'000 &  \cellcolor{OldLace} -2'444 &  \cellcolor{OldLace}164 &  \cellcolor{OldLace}3'067 &  \cellcolor{OldLace}6'158& \cellcolor{OldLace}10'649\\\hline 
$\infty$ &  -1'038.1  &\cellcolor{LightGrey}  0 &\cellcolor{LightGrey}  0&  578 & 3'669 &  8'160 \\
 \cellcolor{LightGrey}  4'000 &-288.7 &\cellcolor{LightGrey} 0 &\cellcolor{LightGrey}  0 &  1'493 & \cellcolor{LightGrey} 4'000 & \cellcolor{LightGrey} 4'000 \\
\cellcolor{LightGrey} 2'443 &  999.7 &\cellcolor{LightGrey} 0  &\cellcolor{OldLace} 164 & \cellcolor{LightGrey}2'443 & \cellcolor{LightGrey}2'443 & \cellcolor{LightGrey}2'443  \\
\cellcolor{LightGrey}  1500 & 3'901.0  & 1'100 &\cellcolor{LightGrey} 1'500&\cellcolor{LightGrey} 1'500  &\cellcolor{LightGrey}1'500 &\cellcolor{LightGrey}1'500  \\ \hline \end{tabular} 
\caption{Quantiles of $X^{\hat{\pi}^{l,u}}_T$ (different $K_u$ )} \end{table}

\newpage

\section{Conclusion}
When implementing an optimal strategy using exponential utility, special attention should be accorded to risk aversion. First, in practice, assuming constant risk aversion seems to be problematic. Second, the strategy is very sensitive to the risk aversion parameter, so the latter should be fitted carefully. 
A characteristic of this optimal unconstrained strategy seems to be the deterministic amount that is invested in the stock. As a consequence, it is particularly attractive for a small investors with lower initial wealth around 4'000, but has no noticeable effect on higher wealth (around 10'000). Also, the terminal wealth is normally distributed, which could be an advantage, because it is a well-known concept, but it can also lead to negative retirement wealth. \newline
Introducing constraints on the retirement wealth can be an effective tool to control downside risk and improve the return between the limits, but their impact depends on the choice of their values. A high lower constraint will reduce the investment close to zero and a small lower constraint might be of little meaning. Besides this, the lower constrained optimal strategy seems to be suitable for implementation. Concerning the upper constraint, it might be easier to set an (individual) value, since also larger values show a relevant increase of quantiles. However, this strategy involves short-selling large amounts of stocks and, at least in practice, seems to be very sensitive. This effect can be reduced by setting a lower constraint, which is one of the reasons why it might be best to use the constraints in combination. Another could be the possibility of financing the upper by the lower constraint, yielding the same distribution as the optimal unconstrained strategy between the boundaries.\newline
Restricting the investment to 100$\%$ of wealth does avoid negative terminal wealth, but changes its distribution towards a log-normal, i.e. lower quantiles decrease. In addition, retirement wealth can surpass the upper and lower constraints. The extent of the effect of this restriction depends on many factors and can vary between no effect and full effect even for realistic scenarios.
\newline
As this is a first approach to the optimal strategy under constraints, the model was chosen to be rather simple. On one hand, the simulation of the market with constant volatility and expected return and only one risky asset might not reflect the complexity of  real markets in an accurate way. On the other hand, the investment process does not take into consideration additional requirements by the investor such as a saving process, value loss by inflation or transaction costs. Especially the latter could have a considerable impact and should therefore be taken into account.


\section{Outlook}
A basis for future research could be the consideration of other utility functions, that are more realistic, and compare the results. Since the effect of constraints on the terminal wealth distribution could be quantified in this thesis and an equivalent result is known for power utility, a direct comparison between these strategies could be a start. \newline
Further, it might be a difficult task for investors to estimate their future needs for retirement and fix a constraint at the beginning of a (rather long) investment horizon. Because the developed strategy involves trading with options, it might offer enough flexibility to modify the constraints at some point within the investment period and it could be interesting to further investigate this possibility. \newline
To further develop the optimal constrained strategy for exponential utility, it also seems necessary to consider the investor's saving process, transaction costs and effects of inflation.\newline

\newpage

\section*{Appendix A}
\addcontentsline{toc}{section}{Appendix A}
\subsection{Derivation of HJB}
The idea is to expand the problem by not only considering a wealth process starting from $t_0= 0$ as defined in (2.3), but for any fixed time $t$ $\in$ [0,T]. The optimal strategy then depends on time $t$ and the wealth at time $t$ and yields the expected utility of terminal wealth  given by  the \textit{optimal value function}
\begin{equation}V:[0,T] \times \mathbb{R}^{+} \rightarrow  \mathbb{R}^{+}, \newline (t,y) \longmapsto V(t,y) := \sup\limits_{\pi\in \mathcal{A}}\{ \mathbb{E}[U(X^{\pi}_T) | X^{\pi}_{t} = y]\}.\end{equation} For simplicity we write $\mathbb{E}[U(X^{\pi}_T) | X^{\pi}_t = y ] =: \mathbb{E}_{t,y}[U(X^{\pi}_T)]$.
\newline
Let $\hat{\pi}$ and $\tilde{\pi}$ be two strategies (called \textit{control laws}) defined on [t,T] $\times \mathbb{R}_{[y, \infty)}$, for a fixed starting point ($t$,y)$\in$ [0,T] $\times\mathbb{R}^{+}$, such that $\hat{\pi}$ is the optimal control law and $\tilde{\pi}$ is a control law that switches to the optimal control law after a short period of time h: 
\newline

$\tilde{\pi}(s,y)$= $\begin{cases}{\hat{\pi}(s)}&\text{if t $\in$ [t + h,T]}\\{\pi(s)}&\text{if s  $\in$ [t,t+h)}\end{cases}$ \newline
for a fixed arbitrary control law $\pi \in \mathcal{A}$ and $h>0$ such that $t+h$ < T. \newline  
The \textit{value function} is defined as follows:  \newline \begin{equation}
 J: [0,T]\times\mathbb{R}^{+} \times \mathcal{A}\rightarrow \mathbb{R}^{+},\ (t, y,\pi) \longmapsto J(t,y,\pi) = \mathbb{E}_{t,y}[U(X^{\pi}_T)]
\end{equation}
Of course, for the optimal strategy the value function is identical with the optimal value function. Comparing the two strategies, it is clear that the value function of the optimal strategy should by definition be larger or equal to the value function of any other strategy, in particular $V(t,y)\geq J(t, y,\tilde{\pi})$ $\forall$ (t,y) $\in [0,T]\times\mathbb{R}^{+}.$\newline If  $\tilde{\pi}$ takes $X^{\pi}_{t}$ at time $t$ to $X^{\pi} _{t+h}$ at time t+h, the expected utility at terminal time is $\mathbb{E}[U(X^{\hat{\pi}}_T)|X^{\hat{\pi}}_{t+h} = X^{\pi}_{t+h}]= V(t+h, X^{\pi}_{t+h})$. Since $X^{\pi}_{t+h}$ is stochastic and $X^{\pi}_{t} = y$ is fixed, it follows  $J(t, y, \tilde{\pi})$ = $\mathbb{E}[V(t+h, X^{\pi}_{t+h}) | X^{\pi}_{t} = y ] = \mathbb{E}_{t,y}[V(t+h, X^{\pi}_{t+h})]$. From this results the inequality
 \begin{equation}
 V(t,y) \geq \mathbb{E}_{t,y}[V(t+h, X^{\pi}_{t+h}]
\end{equation}
Using Ito's formula, $V(t+h, X^{\pi}_{t+h})$ can be expanded:\newline $\mathbb{E}[V(t+h, X^{\pi}_{t+h})] = V(t, y) + \newline
\mathbb{E}_{t,y}[ $$\int_{t}^{t+h} \{\dfrac{\partial V}{\partial t}(s,X^{\pi}_s) + 
[(rX^{\pi}_s+ \pi_s(\mu-r)X^{\pi}_s]{\dfrac{\partial V}{\partial x}}(s,X^{\pi}_s) + ( \sigma  \pi_s X^{\pi}_s ) ^2\dfrac{\partial ^2V}{\partial x^2}\} ds$$] + E_{t,y}[\sigma \pi_s X^{\pi}_s dW(s)]$. If we assume enough integrability, the stochastic part vanishes. With equation (2.7) follows\begin{equation} 0 \geq E[\int_{t}^{t+h} \{\dfrac{\partial V}{\partial t}(s,X^{\pi}_s) + [(rX^{\pi}_s+\pi_s(\mu-r)X^{\pi}_s]{\dfrac{\partial V}{\partial x}}(s,X^{\pi}_s)| X^\pi_{t}=y]\end{equation}
Dividing both sides by h, finding the limit for h$\downarrow$ 0 and setting $x = y$ then gives  the \textbf{Hamilton-Jacobi-Bellmann Equation (HJB)} as in Chapter 1. 
\subsection{Solving the linear SDE}
Having set a boundary condition $X(0) = X_0$, we want to solve the differential equation \begin{verse}\begin{center}
$dX^{\hat{\pi}}_t = [rX^{\hat{\pi}}_t + \dfrac{\theta^2}{\alpha} e^{-r(T-t)}] dt + \dfrac{\theta}{\alpha}e^{-r(T-t)} dW_t$ \end{center}\end{verse}
Via the substitutions $dZ_t := rdt$ and $dH_t := \dfrac{\theta^2}{\alpha} e^{-r(T-t)}dt + \dfrac{\theta}{\alpha}e^{-r(T-t)} dW_t $ this transforms to \begin{verse}\begin{center}$dX^{\hat{\pi}}_t = X^{\hat{\pi}}_t dZ_t + dH_t$  and  $X(0) = X_0$ \end{center}\end{verse}
and can be written as  $X^{\hat{\pi}}_t = \bigintsss_{0}^{t} X_s dZs + H_t$. \newline
By Theorem 52 in Chapter 5 of \cite{Protter}, it has a (unique) solution \begin{verse}\begin{center}
$\epsilon _H(Z)_t = \epsilon(Z)_t \{ H_0 +  \bigintsss_{0}^{t} \epsilon(Z)_s^{-1} d( H_s - [H,Z]_s)\}, \ \ \ \ \ \ (\star$)\end{center}\end{verse} where [H,Z] is the quadratic covariation defined in Section 6, Chapter 2. \newline
First, note that$ \bigintsss_{Z_0}^{Z_t} dZ(t) = \bigintsss_{0}^{t}r dt$ , and hence \begin{equation}Z_t = Z_0 + rt. \end{equation}
Then, by Theorem 37 in Chapter 2 of \cite{Protter},  $\epsilon(Z)_t$ has the form $\epsilon(Z)_t = e^{Z_t - [Z,Z]_t /2}$ \newline and since $[Z,Z]_t = 0$, this gives with (6.5) \begin{equation} \epsilon(Z)_t = e^{Z_0 + rt} 
\end{equation}
Further, from the second substitution, we get: \begin{equation} H_t = H_0 + \bigintsss_{0}^{t} \dfrac{\theta^2}{\alpha} e^{-r(T-s)}ds +  \bigintsss_{0}^{t} \dfrac{\theta}{\alpha}e^{-r(T-s)} dW_s. 
\end{equation}Plugging (6.6), (6.5), (6.4) into ($\star$) and having $d[H,Z]_t = 0$ gives a solution  \begin{verse}\begin{center}
$X^{\hat{\pi}}_t$ = $e^{Z_0 + rt} \left\lbrace H_0 + \bigintsss_{0}^{t} e^{-Z_0 -rs} ( \dfrac{\theta^2}{\alpha } e^{-r(T-s)}ds + \dfrac{\theta}{\alpha}e^{-r(T-s)} dW_s) \right\rbrace$ .
\end{center}\end{verse} At t= 0 we have $X_0 = e^{Z_0}(H_0 + 0)$, so we can set $Z_0 = 0$ and $H_0 = X_0$ and get
\begin{verse}\begin{center}
$X^{\hat{\pi}}_t = e^{rt} X_0 +   e^{rt}  e^{-rT}\left\lbrace  \bigintsss_{0}^{t} e^{-rs} (\dfrac{\theta^2}{\alpha }e^{rs} ds + \dfrac{\theta}{\alpha} e^{rs} dW_s) \right\rbrace$  \newline
= $e^{rt} X_0 + e^{t-T} \left\lbrace  \bigintsss_{0}^{t} (\dfrac{\theta^2}{\alpha }ds + \dfrac{\theta}{\alpha}dW_s)\right\rbrace$ \newline
= $e^{rt} X_0 + e^{t-T} (\dfrac{\theta^2}{\alpha }t + \dfrac{\theta}{\alpha}W_t)$,
\end{center}\end{verse}
which is the optimal wealth process from Proposition 3.\newline Note that for a completion of the proof it also needs to be checked that $H_t$ is a semimartingale and $Z_t$ a continuous semimartingale.

\subsection{Expected Utility Theory: From Lotteries to Utility Functions}
Utility theory is an approach to describe the decisions of people in situations when the outcome is not certain and is widely used in financial economics. Since it also is the basis for the results used in this thesis, this section is thought to give a brief outline of its main ideas. However, they can not be reflected in all completeness and rigour. This summary is based on \cite{Noeldeke}, \cite{Machina},\cite{Norstad} and \cite{Kirkwood}, except for the examples.

\begin{wrapfigure}{l}{0.42\textwidth}
\includegraphics[width=49mm]{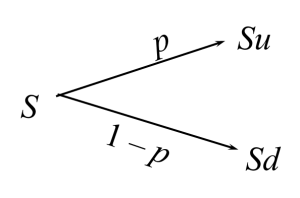}\caption{Example of a Lottery \cite{Lottery}} \end{wrapfigure}
The basic construction of a situation facing uncertainty is called a \textit{lottery}: It has two possible outcomes, each with a probability, and the probabilites add up to one. To stay close to the model of this thesis, this could be the value of a stock moving up to a value $S_u$ with a probability \textit{p} and moving down to a value $S_d$ with a probability 1-\textit{p}.  \newline  Of course, the decision-maker doesn't need to participiate in the lottery -so the question arises, when he would be willing to do so. A key observation here is that the simple expected value of the outcome is not a suitable criterion to describe the decisions of most people. In fact, they prefer a deterministic outcome over an uncertain one to a certain degree, even when the expected value of the lottery is higher than the deterministic alternative. This observation is commonly called \textit{risk aversion}. \newline To assess the decision maker's risk aversion in more detail, it is possible to compare different lotteries. For example, the decision-maker could be in a situation where he can invest in two different stocks with different upward and downward possibilites and different resulting values. He then would prefer one option over the other or be indifferent. These preferences are described by \textit{preference relations} $\succeq$ and are assumed to have some properties to ensure consistency (completeness, transitivity, monotony, continuity). Furthermore it is assumed, that they satisfy the so-called \textit{independence axiom}, which states that if one prefers a lottery A over the lottery B he would also prefer another lottery, that leads to A with a certain probability p (and with  probability 1-p to another lottery C ) over one that leads to the lottery B with p (and also with probability 1-p to another lottery C).  In the previous example, this could mean that if the investor prefers stock A over stock B, he will also prefer a portfolio with one stock A and a stock C over a portfolio with stock B and stock C. \newline 
This axiom implies one main result of utility theory, the \textit{Representation Theorem}, which quantifies preferences between situations of uncertainty. Since the axiom of independence allows to 'continously mix' lotteries, every lottery can be traced down to the comparison with a value, which is the lottery that has a certain output with probability 1. In our example, this would be some value accorded to the investment in a riskfree asset (it is $\textit{not}$ the expected value). The theorem states that the number accorded to any lottery has the form of a linear combination of some function on the outputs of the lottery, weighted by their probabilites. In other words: The preferences of a decision-maker with respect to scenarios with uncertain outcome can be described by the expected value of a function on these outcomes (called $\textit{utility function}$).
\newline  In our example, this would imply, that the decision-maker would invest in the stock, if its expected utility is at least the expected utility of the riskfree asset (invested for the same period): $pU(S_u) + (1-p) U(S_d) \geq \mathbb{E}[U(e^{r})] = U(e^{r})$. 
\subsection{Exponential Utility and Risk Aversion}
One property of utility functions is invariance over affine transformations (called $\textit{cardinality}$), because the same preference relation is invariant over transformations of the lotteries in a similar sense. This is the reason, why instead of using -$e^{-rx}$  to describe exponential utility, often also the form $1-e^{-rx}$ is used. As a result, the values of utility functions cannot be interpretated in an absolute manner. However, they can be used to quantify the 'degree' of risk aversion that an individual exhibits. \newline Arrow and Pratt \cite{Pratt} found, that the more a utility function is concave, the higher the corresponding  $\textit{certainty equivalent}$ is. This is the riskfree alternative outcome that would be considered to be equivalent to the expected utility of a fixed lottery by an investor. For the same risky situation, the more risk-averse an investor is, the smaller is his certainty equivalent. As a very extreme example think of a highly risk-averse person, that, even if negative riskfree interest rates would actually reduce his wealth, he still would prefer this option over the uncertainty of the outcome of a stock with attractive expected return. But for decreasing interest rates, at some point, even he might be willing to invest in the risky asset. This interest would then be considered to be the certainty equivalent to the stock. Since it can be shown that if one utility function is tranformed by a concave function then its certainty equivalent is decreasing, the link from risk aversion to the convexity of the utility function is established. This motivates the definition of the measure for risk aversion ((6) in \cite{Pratt}), called the \textit{coefficient of absolute risk aversion}: 
\begin{verse}

 \begin{center}$\rho(x) := \dfrac{-U''(x)}{U'(x)}$ \end{center}\end{verse}

Clearly, for the exponential utility this gives $\rho (x)$=$\alpha$ , so risk aversion is constant and does not change with  increasing wealth. This property is called \textit{CARA} (constant absolute risk aversion) and is the main reason that exponential utility is often considered as unrealistic (for example, \cite{Guiso} also come to this conclusion looking at empirical data).  

\begin{figure}[H]
 \includegraphics[width=82mm]{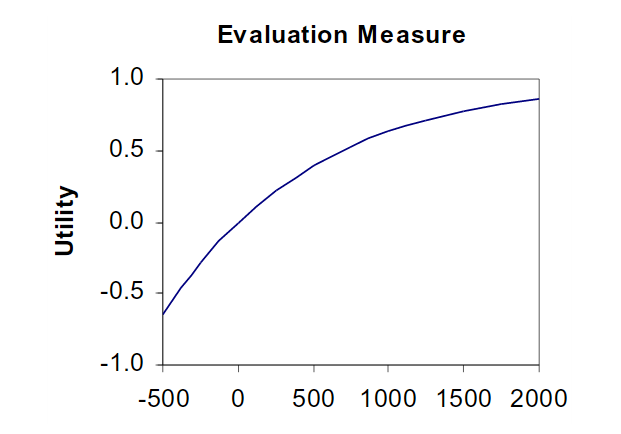} 
\caption{An exponential utility function with parameter $\alpha$=0.0001}
\end{figure}

%
%
%

\subsection{Analytical Approach: Effect of Investment-Restriction}
In order to assess the impact of the restriction on investment in stocks to 100$\%$ of current wealth on a more abstract level, we try to find the probabilities of two 'extreme' scenarios: The first being that the optimal wealth process $X^{\hat{\pi}}_t$ is not affected by the limitation of the investment, hence its terminal wealth distribution is identical with the one from the original strategy. The other case is that the initial wealth is lower than the investment amount required by the optimal strategy $\hat{\pi}$(hence the strategy is constrained by the modification) and the wealth never gets past the limit defined by $\hat{\pi}_t$. It then follows entirely the movements of the stocks and results in a log-normal-distribution.\newline
For the first scenario, we state: 
\begin{result} The probability that Modification 1 (Restriction on Investment) has no effect on the optimal wealth process given by Proposition 2, assuming is given by: 
\begin{equation}\mathbb{P}[X_t > \hat{\pi}_t X^{\hat{\pi}}_t   \forall t \in [0,T-1]] = \prod_{t = 0}^{T-2} \Phi(\dfrac{C_t + \theta}{\sqrt{t+1}})\end{equation} for  $C_t$ := $\dfrac{\alpha}{\theta}X_0 e^{rT}+\theta t -\dfrac{1}{\sigma} $ and assuming $X_0$ > $\hat{{\pi}}_0 X_0$\end{result}
\begin{proof}We want to calculate the probability that $X^{\hat{\pi}}_t$ is always larger than the investment $\hat{\pi}_t X^{\hat{\pi}}_t$ : 
$\mathbb{P}[ X^{\hat{\pi}}_t > \hat{\pi}_t X^{\hat{\pi}}_t \ \forall t \in [0,T-1] ]$ \newline
We iterate over yearly time-steps for simplicity, but for an arbitrary step width h the approach should also work by setting t+h instead of t+1 etc. \newline
$\mathbb{P}[X^{\hat{\pi}}_t > \hat{\pi}_t \ \ X^{\hat{\pi}}_t  \ \forall t \in [0,T-1] ]$
\newline
= $\mathbb{P}[X^{\hat{\pi}}_{T-1} > \hat{\pi}_{T-1} X^{\hat{\pi}}_{T-1} | X^{\hat{\pi}}_{T-2} >\hat{\pi}_{T-2} X^{\hat{\pi}}_{T-2}] \ \mathbb{P}[X_t > \hat{\pi}_t X^{\hat{\pi}}_t\ \forall t \in [0,T-2] ]$ \newline
= $\mathbb{P}[X^{\hat{\pi}}_{T-1} > \hat{\pi}_{T-1} X^{\hat{\pi}}_{T-1} | X^{\hat{\pi}}_{T-2} > \hat{\pi}_{T-2} X^{\hat{\pi}}_{T-2}]\cdot..\cdot\mathbb{P}[X^{\hat{\pi}}_{1} > \hat{\pi}_1 X_1 | X_0 > \hat{\pi}_0X_0 ]$ \newline
 = $\prod_{t=0}^{T-2} \mathbb{P}[X^{\hat{\pi}}_{t+1} > \hat{\pi}_{t+1} X^{\hat{\pi}}_{t+1} | X^{\hat{\pi}}_{t} > \pi_{t} X^{\hat{\pi}}_{t}]$,   ($\star$) \newline assuming $\mathbb{P}[X_0>\pi_0 X_0]$ = 1.\newline
Since $X^{\hat{\pi}}_t$ follows the wealth process given by (2.11), we have: \newline
$\mathbb{P}[X^{\hat{\pi}}_{t+1} > \hat{\pi}_{t+1} X^{\hat{\pi}}_{t+1} | X^{\hat{\pi}}_{t} > \hat{\pi}_{t}X^{\hat{\pi}}_{t}]$  \newline
= $\mathbb{P}[X^{\hat{\pi}}_0 e^{r(t+1)}+(t+1)\dfrac{\theta^2}{\alpha}e^{r(t+1-T)}+\dfrac{\theta}{\alpha}e^{r(t+1-T)}W_{t+1}> \dfrac{\theta}{\alpha \sigma}e^{-r(T-(t+1))} |X_0 e^{rt}+t \dfrac{\theta^2}{\alpha}e^{r(t-T)}+\dfrac{\theta}{\alpha}e^{r(t-T)}W_t > \dfrac{\theta}{\alpha \sigma}e^{-r(T-t)}]$, where  $W_t$ is the standard Brownian movement. 
\newline
For $C_t$ := $\dfrac{\alpha}{\theta}X_0 e^{rT}+\theta t -\dfrac{1}{\sigma} $ and $W_{t+1}-W_t$ := Z$ \sim \mathcal{N}$(0,1), this simplifies to \newline 
$\mathbb{P}[C_t +\theta +W_t+Z > 0|C_t + W_t > 0]$ \newline
= $\mathbb{P}[ \theta + Z > -(W_t + C_t)]$ \newline 
= $\mathbb{P}[\theta + W_{t+1}-W_t > -W_t -C_t]$ \newline
= $\mathbb{P}[W_{t+1} > -C_t - \theta]$ \newline 
= $1-\Phi(-\dfrac{C_t+\theta}{\sqrt{t+1}})$, since $\dfrac{W_{t+1}}{\sqrt{t+1}} \sim \mathcal{N}$(0,1),\newline
= $\Phi(\dfrac{C_t+\theta}{\sqrt{t+1}})$ \newline
Plugging this into ($\star$) gives the result. \end{proof}

By finding the limits, we then get: \begin{result} It holds: 
\begin{equation}\lim_{\theta \to 0}  \mathbb{P}[X^{\pi}_t > \hat{\pi}_t X^{\hat{\pi}}_t \forall t \in [0,T-1] ]  = 1  \end{equation}
\begin{equation}\lim_{\theta \to \infty}  \mathbb{P}[X^{\hat{\pi}}_t > \hat{\pi}_t X_t \forall t \in [0,T-1] ]  = \begin{cases}{1}&\text{if $\mu - r > 1 $}\\{0 }&\text{if $\mu - r < \frac{1}{T-1} $}\end{cases} \end{equation}for the scenario of Result 1 and $\theta = \dfrac{\mu-r}{\sigma}$ the market price of risk.
\end{result}
\begin{proof}
First introduce the notation $\Phi_t := \Phi(\dfrac{\dfrac{\alpha}{\theta}X_0 e^{rT}+\theta (t+1) -\dfrac{1}{\sigma}}{\sqrt{t+1}})$\newline
Now show (6.6): $\theta \downarrow 0$ is the case if  $\mu-r \downarrow  0 $ and/or $\sigma \uparrow \infty$\newline
In both cases , \newline
$ \lim_{\theta \to 0} \dfrac{\dfrac{\alpha}{\theta}X_0 e^{rT}+\theta (t+1) -\dfrac{1}{\sigma}}{\sqrt{t+1}}$ = $\infty$, hence 
$\lim_{\theta \to 0}\Phi_t= 1$ and with (6.5) this gives the result. \newline
For (6.7): $\theta \to \infty$ is the case if $\mu - r \to \infty$ or if $\sigma \downarrow 0$. \newline 
The first case is similar to the proof of 6.6. The second case can be written as \newline
$ \dfrac{ \lim_{\sigma \to 0} \dfrac{\alpha \sigma}{(\mu - r) }X_0 e^{rT}+ \lim_{\sigma \to 0} [\theta (t+1) -\dfrac{1}{\sigma}]}{\sqrt{t+1}}$  \newline = $\dfrac{1}{\sqrt{t+1}} \lim_{\sigma \to 0}\dfrac{\alpha \sigma}{(\mu - r)}X_0 e^{rT}+ \dfrac{1}{\sqrt{t+1}}\lim_{\sigma \to 0} (\dfrac{(\mu-r) (t+1) -1}{\sigma}) $ \newline = $\dfrac{1}{\sqrt{t+1}}\lim_{\sigma \to 0} (\dfrac{(\mu-r) (t+1) -1}{\sigma}) $ 
=$ \begin{cases}{\infty}&\text{if$(\mu - r)(t+1)-1 > 0 $}\\{-\infty}&\text{if $(\mu - r)(t+1)-1 < 0 $}\end{cases}$ \newline
So,  $ \lim_{\sigma \to 0} \Phi_t$ = $\begin{cases}{1 }&\text{if $\mu - r> \dfrac{1}{(t+1)} $}\\{0 }&\text{if $\mu - r < \dfrac{1}{(t+1)} $}\end{cases}$ and therefore \newline
$\prod_{t=0}^{T-2} \Phi_t$ = $\begin{cases}{1}&\text{if $\mu - r > 1 $}\\{0 }&\text{if $\mu - r < \frac{1}{T-1} $}\end{cases} $
\end{proof}

For the second scenario, we get: \begin{result}
The probability, that the optimal wealth process given by Propsition 2 is entirely limited by the Modification 1 (Restriction on Investment) is given by\begin{equation}\mathbb{P}[X_t < \pi_t X_t   \forall t \in [0,T-1]] = \prod_{t = 0}^{T-2} \Phi(\dfrac{\tilde{A}_t}{ \sigma \sqrt{t+1}})\end{equation} for $ A_t$:= ln( $\dfrac{\theta}{\sigma \alpha X_0})-r(T-t-1)+ (\mu-\dfrac{\sigma^2}{2})(t+1)$ and assuming $X_0$ < $\hat{\pi}_0X_0$\end{result}
\begin{proof}
The second extreme case is the scenario, where the wealth $X^{\hat{\pi}}_t$ always stays under the investment amount required by the optimal strategy. We then have:
\newline $\mathbb{P}[X^{\hat{\pi}}_t < \hat{\pi}_t X^{\hat{\pi}}_t \forall t \in [0,T-1]] = \prod_{t=0}^{T-2} \mathbb{P}[X^{\hat{\pi}}_{t+1} < \hat{\pi}_{t+1} X^{\hat{\pi}}_{t+1} | X^{\hat{\pi}}_{t} < \hat{\pi}_{t} X^{\hat{\pi}}_{t}]$ \newline for the same reasons as before.\newline  Since the movement of $X^{\hat{\pi}}_t$ is identical with the movement of the stocks, it follows \newline  $\mathbb{P}[X^{\hat{\pi}}_{t+1} <\hat{\pi}_{t+1} X^{\hat{\pi}}_{t+1}|X_{t} < \hat{\pi}_{t} X^{\hat{\pi}}_t]$ \newline = $ \mathbb{P}[X_0 e^{-(\mu-\dfrac{\sigma^2}{2})(t+1)+\sigma W_{t+1}} < \dfrac{\theta}{\sigma \alpha} e^{-r(T-t-1)}|X_0e^{-(\mu-\dfrac{\sigma^2}{2})t+\sigma W_{t}}<\dfrac{\theta}{\sigma \alpha}e^{-r(T-t)}]$ \newline = $\mathbb{P}[$ln$(X_0)-(\mu-\dfrac{\sigma^2}{2})(t+1)+\sigma W_{t+1}<$ln$(\dfrac{\theta}{\sigma \alpha})-r(T-t-1)|$ln$(X_0)-(\mu-\dfrac{\sigma^2}{2})t+\sigma W_{t}<$ln$(\dfrac{\theta}{\sigma\alpha})-r(T-t)]$\newline = $\mathbb{P}[-(\mu-\dfrac{\sigma^2}{2})+\sigma W_{t+1}<r+A_t|\sigma W_{t} < A_t]$ for $ A_t:= $ln$( \dfrac{\theta}{\sigma \alpha})-r(T-t)-$ln$(X_0)+ (\mu-\dfrac{\sigma^2}{2})t$ \newline
=$\mathbb{P}[0 <A_t -\sigma W_{t}+\sigma W_{t} + (\mu-\dfrac{\sigma^2}{2}) + r- \sigma W_{t+1}| 0< A_t-\sigma W_{t} ]$ \newline = $\mathbb{P}[A_t -\sigma W_{t} > -(\sigma W_{t} + (\mu-\dfrac{\sigma^2}{2}) + r- \sigma W_{t+1})]$ \newline =  $\mathbb{P}[ - \sigma W_{t+1} > -(\mu-\dfrac{\sigma^2}{2}) - r - A_t]$ \newline =  $\mathbb{P}[W_{t+1} < \dfrac{\tilde{A}_t}{\sigma} ]$ \ for $\tilde{A}_t = A_t + r +(\mu-\dfrac{\sigma^2}{2})$ \newline =$ \Phi(\dfrac{\tilde{A}_t}{ \sigma \sqrt{t+1}})$ \end{proof} 
Taking the limits leads to: \begin{result} It holds
\begin{equation}\lim_{\theta \to 0}  \mathbb{P}[X_t < \pi_t X_t \forall t \in [0,T-1] ]  = 0  \end{equation}
\begin{equation}\lim_{\theta \to \infty}  \mathbb{P}[X_t < \pi_t X_t \forall t \in [0,T-1] ]  = 1  \end{equation}for the scenario of Result 1 and $\theta = \dfrac{\mu-r}{\sigma}$ the market price of risk
\end{result}
\begin{proof}
Note that $\lim_{\mu -r \to 0 or \sigma \to  \infty } ln(\dfrac{\mu -r}{\sigma ^2 \alpha}) = \lim_{x \to 0  } ln(x) = - \infty$  , \newline so $ \lim_{\theta \to 0} \tilde{A_t} =  \lim_{\theta \to 0} (  $ln$(\dfrac{\mu -r}{\sigma ^2 \alpha})+ (\mu-\dfrac{\sigma^2}{2})(t+1) ) -r(T-t)-$ln$(X_0) + r= - \infty$ \newline 
 In the case $ \sigma \to \infty$, we can apply Hôpital and get \newline $\lim_{\sigma \to \infty} \dfrac{\tilde{A_t}}{\sigma \sqrt{t+1}}=  \lim_{\sigma \to \infty} \dfrac{ \dfrac{-2}{\sigma} -\sigma(t+1)}{\sqrt{t+1}} = - \infty$ \newline
and therefore, $ \lim_{\theta \to 0}  \Phi(\dfrac{\tilde{A}_t}{\sqrt{t+1} \sigma}) = 0$.
\newline On the other hand, we have $\theta \to \infty$ if $\mu - r \to \infty$ or $\sigma \to 0$. Both can be easily calculated and gives the result (2.23).
\end{proof} 
The Results 2 and 4 seem to contradict the empirical observations, where we saw an increasing impact of the restriction on investment in stocks with increasing volatiliy and decreasing $\mu$ -r. However, this could also indicate that limits might not be a suitable concept to assess these effects. For example, if $\theta$ $\to$ 0, the investment in Result 4, $\hat{\pi}_t$, goes to zero and hence also the initial wealth, because of the assumption $\mathbb{P}[ X_0 < \hat{\pi}_0 X_0]$ = 1. But this case is exluded in (2.3). For $\sigma$ $\to$ 0, the limits seem to make sense: In the limit, the wealth process is deterministic, so it only depends on the initial conditions. Since we assumed $\mathbb{P}[X_0 > \hat{\pi}_0 X_0$] = 1  for Result 2 and  $\mathbb{P}[ X_0 < \hat{\pi}_0 X_0$] = 1  for Result 4, the limits seem reasonable. In particular, we see that the gap $\mu$-r needs to be sufficiently large. Otherwise, the deterministic amount to be invested in stocks grows at a faster rate than $X^{\pi}_t$  and surpasses it at some point.\newline  \newline
Also note, that the two scenarios do generally not converge at the same rate. For example, if we set $X_0$= $X_0 \pi_0$ to be able to make a comparison, we get for a 'realistic' setting ($\alpha$ = 0.001, r = 0.01, $\mu$ = 0.04, $\sigma$ = 0.1, T = 20) the probability of the first scenario to be 1.54 $\%$, whereas for the second scenario it is 10.37 $\%$.  This indicates that once the wealth falls under the optimal investment strategy and follows the modified strategy, it is more likely to stick with it and stay under the optimal investment curve. This is consistent with the observations we made, in particular the resulting distribution of terminal wealth being more concentrated at lower quantiles.  

\subsection{Quantiles of Optimal Terminal Wealth Distribution}
For a more detailed description of the CDF of $X^{\hat{\pi}}_T$, some empirical quantiles are listed below. To get these, a sample of 3000 paths and a stepwidth of 0.1 (for iteration of time of investment t) was used. \newline For comparison, the theoretical CDF of $X^{\hat{\pi}}_T$ $\sim$ $\mathcal{N}$( $X_0 e^{rT} + T\dfrac{\theta^2}{\alpha}$, $\dfrac{\theta^2}{\alpha^2}T$) is listed as well. Note, that these values also only depend on $X_0 \alpha$ (for the values below, $\alpha$ is set to be 0.001). The other parameters were fixed at r= 0.01 $\mu$= 0.03, $\sigma$= 0.1, T =20, $\alpha$= 0.001.
\begin{table}[H]
 \begin{minipage}{0.43\linewidth} \centering
 terminal wealth quantiles
\setlength{\tabcolsep}{0.8mm}
\renewcommand{\arraystretch}{1.2}
\begin{tabular}{|l|c|c|c|r|}
\hline
$X_0$&25$\%$&50$\%$&75$\%$&95$\%$\\ \hline 
10&182&818&1'404 & 2'300  \\
\rowcolor{LightMintGreen} 
$10^2$& 291 & 928 & 1'514 & 2'410\\
\rowcolor{LightMintGreen} 
 $10^3$ &1'403 & 1'990 & 2'612 & 3'490\\ 
\rowcolor{Lavender} 
$10^4$& 12'426 & 13'018 & 13'618 & 14'547\\
\rowcolor{Lavender} 
$10^5$ &122'297 & 122'934 & 123'520 & 124'416\\ \hline
 \end{tabular}  
 \end{minipage}
  \begin{minipage}{0.40\linewidth} \centering
return on initial wealth quantiles
\setlength{\tabcolsep}{0.8mm}
\renewcommand{\arraystretch}{1.2} 
\begin{tabular}{|c|c|c|r|}
\hline
  25$\%$& 50$\%$&75$\%$&95$\%$\\ \hline 
 1'815$\%$ & 8'179$\%$ & 14'037$\%$&23'001$\%$  \\
\rowcolor{LightMintGreen} 
 291$\%$&928$\%$&1'514$\%$&2'410$\%$\\
\rowcolor{LightMintGreen} 
 140$\%$ & 199$\%$ &261$\%$&349$\%$\\ 
\rowcolor{Lavender} 
124 $\%$ & 130$\%$ & 136$\%$ &145$\%$\\
\rowcolor{Lavender} 
122$\%$ & 123$\%$ & 124$\%$ & 124$\%$\\ \hline
 \end{tabular}  
 \end{minipage}
\newline
\newline
\newline
  \begin{minipage}{0.36\linewidth} \centering
theoretical terminal wealth quantiles
\setlength{\tabcolsep}{0.8mm}
\renewcommand{\arraystretch}{1.2}
\begin{tabular}{|l|c|c|c|r|}
\hline
 25$\%$ & 50$\%$ & 75$\%$ & 95$\%$ \\ \hline 
209 & 812 & 1'415 & 2'283 \\
\rowcolor{LightMintGreen} 
319 & 922 & 1'525  & 2'393\\
\rowcolor{LightMintGreen} 
1'418& 2'021& 2'625& 3'493\\ 
\rowcolor{Lavender} 
12'411 & 13'014 & 13'617 & 14'485\\
\rowcolor{Lavender} 
122'337&122'940&123'544&124'411\\ \hline
 \end{tabular}  
 \end{minipage}\begin{minipage}{0.5\linewidth} \centering
p.a. return on initial wealth quantiles (continously compounded)
\setlength{\tabcolsep}{0.8mm}
\renewcommand{\arraystretch}{1.2}
\begin{tabular}{|l|c|c|c|r|}
\hline
 25$\%$ & 50$\%$ & 75$\%$ & 95$\%$ \\ \hline 
14$\%$ & 22$\%$ & 25 $\%$& 27$\%$ \\
\rowcolor{LightMintGreen} 
5$\%$& 11$\%$ & 14$\%$  &16$\%$\\
\rowcolor{LightMintGreen} 
2$\%$& 3$\%$& 5$\%$&6$\%$\\ 
\rowcolor{Lavender} 
1$\%$ & 1$\%$ & 2$\%$& 2$\%$\\
\rowcolor{Lavender} 
1$\%$ &1$\%$ &1$\%$ &1$\%$ \\ \hline
 \end{tabular}  
 \end{minipage}\caption {Quantiles of $X^{\hat{\pi}}_T$ (varying $X_0$)} \end{table} 
 
The empiricial distributions of optimal terminal wealth for different inital wealth are illustrated below.
 \begin{figure}[H]
 \begin{minipage}{1\linewidth}
 \includegraphics[width=50mm]{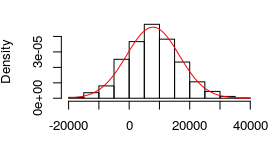}  \includegraphics[width=50mm]{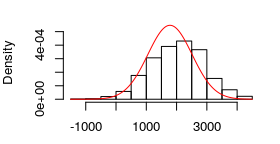}\includegraphics[width=50mm]{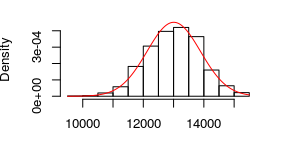} 
 
\caption{Histogram of terminal wealth for $X_0$ = 10; 1000; 10'000}
\end{minipage}
\end{figure} 
These are the quantiles for different volatilites (reffered to in Section 2.3.3):
\begin{table}[H]
\begin{center}   
\begin{tabular}{|l|c|c |c|c|c| }
\hline

$\mu$-r & $X_0$ & $\mathcal{Q}_{0.25}$($X_0$) & $\mathcal{Q}_{0.50}$($X_0$) &  $\mathcal{Q}_{0.75}$($X_0$) & $\mathcal{Q}_{0.95}$($X_0$)   \\ \hline 
5$\%$&256  &94$\%$&248$\%$&392$\%$&599$\%$  \\
 
3$\%$&154&45$\%$&199$\%$&343$\%$&549$\%$  \\
 
 1 $\%$&51&-3$\%$&150$\%$&294$\%$&500$\%$ \\ \hline

 \end{tabular} \subcaption{ $\sigma$ = 0.4}  \end{center}

\begin{center}   
\begin{tabular}{|l|c|c |c|c|c| }
\hline

$\mu$-r& $X_0$ & $\mathcal{Q}_{0.25}$($X_0$) &   $\mathcal{Q}_{0.50}$($X_0$) &  $\mathcal{Q}_{0.75}$($X_0$) &  $\mathcal{Q}_{0.95}$($X_0$)   \\ \hline 
5$\%$&4'094&211$\%$&245$\%$&285$\%$&340$\%$  \\
 
3$\%$&2'456&159$\%$&197$\%$&232$\%$&284$\%$  \\
 
 1 $\%$&819&110$\%$&148$\%$&183$\%$&235$\%$ \\ \hline

 \end{tabular} \subcaption{ $\sigma$ = 0.1} \end{center}
\caption{Quantiles of $X^{\hat{\pi}}_T$ as total return  (and varying $\mu$-r) }
\end{table}   
\begin{table}[H]
\begin{center}   
\begin{tabular}{|l|c|c|c |c|c|c| }
\hline

$\mu$& \textit{r} & Strategy &$\mathcal{Q}_{0.25}$($X_0$) &   $\mathcal{Q}_{0.50}$($X_0$) &  $\mathcal{Q}_{0.75}$($X_0$) &  $\mathcal{Q}_{0.95}$($X_0$)   \\ \hline 
1$\%$&0$\%$&$\hat{\pi}$& -2$\%$&309$\%$&602$\%$&1'024$\%$  \\
 
1$\%$&0$\%$&$\hat{\pi}_m$&82$\%$&112$\%$&149$\%$&228$\%$  \\\hline
0$\%$&-1$\%$&$\hat{\pi}$&-2$\%$&253$\%$&493$\%$&838$\%$  \\
 
0 $\%$&-1$\%$&$\hat{\pi}_m$&67$\%$&91$\%$&122$\%$&187$\%$ \\ \hline

 \end{tabular} \subcaption{ $\sigma$ = 0.1}  \end{center}
\caption{Quantiles of $X^{\hat{\pi}}_T$ and $X^{\hat{\pi}_m}_T$ as total return  (varying \textit{r}) }
\end{table}   

\subsection{Behaviour of $\hat{\pi}_l$ and $\hat{\pi}_u$}
In order to illustrate the qualitative behaviour of the optimal strategies under either an upper or a lower constraint for terminal wealth in more detail, we will focus on the differences to the optimal strategy given by $-\tilde{\pi}_c c(\tilde{X}^{\hat{\pi}}_t, t)$ and $\tilde{\pi}_p p(\tilde{X}^{\hat{\pi}}_t, t)$. \newline For an example case, we set $K_l$=800, $K_u$=1'250, and market conditions $r$ = 0.01,\newline $\sigma$ = 0.1, $\mu$=0.03. \newline First, we are interested in the behaviour of these strategies over time. We therefore fix $\tilde{X}^{\hat{\pi}}_t$ to be 1'000, keeping in mind that this value might not be realistic, since the initial shadow wealth $\tilde{X}_0$ is calculated in function of the initial wealth and can vary greatly. Second, we would expect $\tilde{X}^{\hat{\pi}}_t$ to change over time. Since it is stochastic, of course, the value is not known, but we would expect it to grow (it follows an optimal wealth process, see Proposition 3).  
\begin{figure} [H] 
\begin{minipage}{0.33\linewidth} \includegraphics[width=49.5mm]{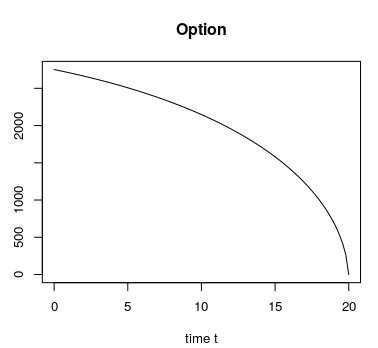} \end{minipage} \begin{minipage}{0.325\linewidth} \includegraphics[width=47mm]{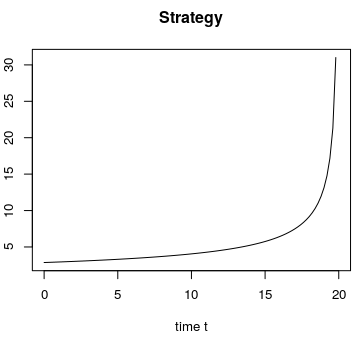}\end{minipage} \begin{minipage}{0.325\linewidth} \includegraphics[width=49mm]{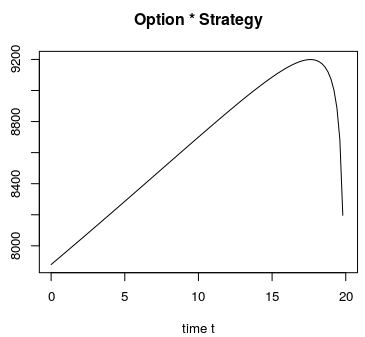}\end{minipage} 
 
\begin{minipage}{0.33\linewidth} \includegraphics[width=49mm]{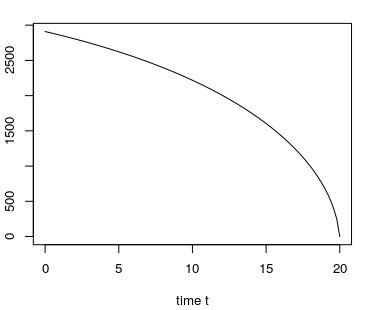} \end{minipage} \begin{minipage}{0.325\linewidth} \includegraphics[width=49mm]{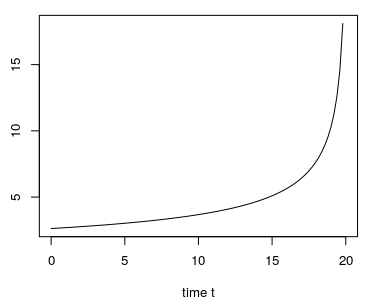}\end{minipage} \begin{minipage}{0.325\linewidth} \includegraphics[width=49mm]{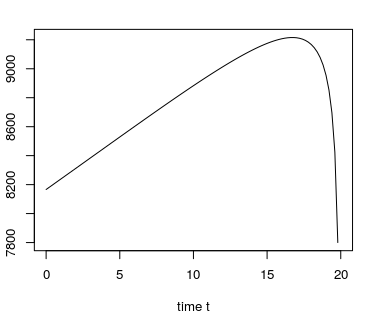}\end{minipage}

\caption{Behaviour of $\tilde{\pi}_p p(\tilde{X}^{\hat{\pi}}_t, t)$ (upper line) and $\tilde{\pi}_c c(\tilde{X}^{\hat{\pi}}_t, t)$ over time }
\end{figure}

Here, we call $p(\tilde{X}^{\hat{\pi}}_t,t)$ the 'Option'and its replicating strategy $\tilde{\pi}_p$ the 'Strategy', and equivalently for the call-option. As one can see in the diagrams on the right and in the formulas, the overall strategy is a product of these two factors. Consequently, for small t, the value of the option is more dominating, and with time getting closer to maturity the strategy gains more influence, resulting in a curve with a peak around $t$=17.\newline Generally, at $t$=$T$, the value of the option is identical to its payoff. In the case considered, $K_l$ < $\tilde{X}^{\hat{\pi}}_t$ < $K_u$ , so the options wouldn't be exercised and their payoff is zero. But it does not need to be like this, as for other strike prices and fixed current shadow wealth, the value of the option can indeed converge to a positive value.\newline
We will now fix the time at $t$=17 and look at the impact of the current shadow wealth $\tilde{X}^{\hat{\pi}}_{17}$.\newline
\begin{figure} [H] 
\centering{\includegraphics[width=130mm]{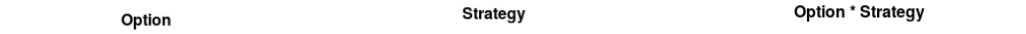}}
\begin{minipage}{0.33\linewidth} \includegraphics[width=48.5mm]{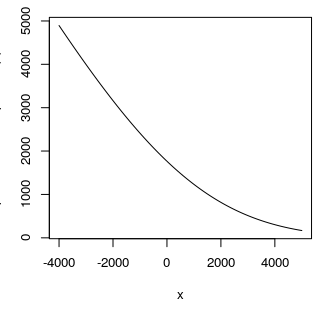} \end{minipage} \begin{minipage}{0.325\linewidth} \includegraphics[width=48.5mm]{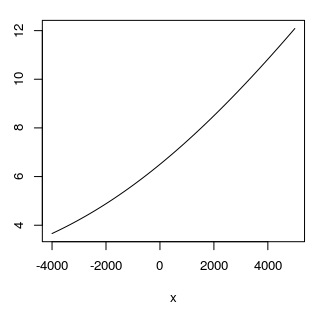}\end{minipage} \begin{minipage}{0.325\linewidth} \includegraphics[width=48.5mm]{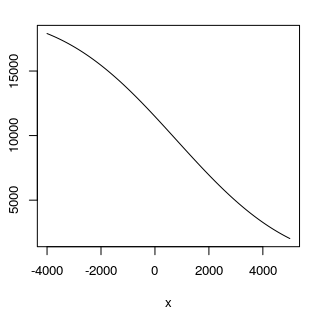}\end{minipage} 
 
\begin{minipage}{0.33\linewidth} \includegraphics[width=48.5mm]{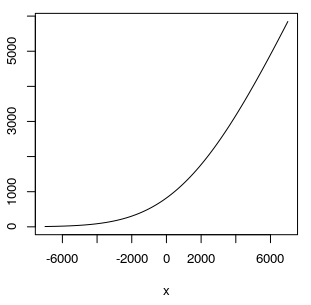} \end{minipage} \begin{minipage}{0.325\linewidth} \includegraphics[width=48.5mm]{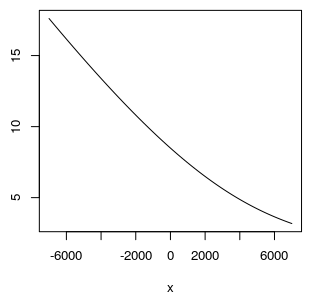}\end{minipage} \begin{minipage}{0.325\linewidth} \includegraphics[width=49mm]{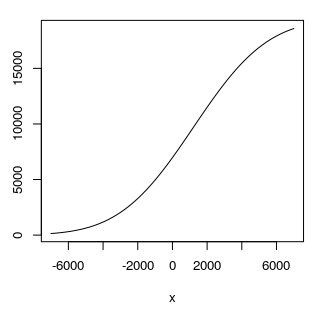}\end{minipage}

\caption{Behaviour of $\tilde{\pi}_p p(\tilde{X}^{\hat{\pi}}_{17}, 17)$ (upper line) and $\tilde{\pi}_c c(\tilde{X}^{\hat{\pi}}_{17}, 17)$ in function of current shadow wealth}
\end{figure}

On the right one can see the behaviour similar to the one shown in Section 3.2 and 4.2. 
It is the result of the 'inverse' behaviour of the replicating strategies and the options setting each other off. The convergences can also easily be derived from the formulas of $\tilde{\pi}_p$ and $\tilde{\pi}_c$. 

\begin{equation} \lim_{\tilde{X}^{\hat{\pi}}_{t} \to \infty} \tilde{\pi}_c (\tilde{X}^{\hat{\pi}}_{t}) = \lim_{\tilde{X}^{\hat{\pi}}_{t} \to - \infty} \tilde{\pi}_p (\tilde{X}^{\hat{\pi}}_{t}) = 0 \mathrm{for a fixed t} \in [0,T].  \end{equation}

\begin{proof}
By definition of $d_u(\tilde{X}^{\hat{\pi}}_{t}) = K_u- \tilde{X}^{\hat{\pi}}_{t} e^{r(T-t)}$ we have $d_u(\tilde{X}^{\hat{\pi}}_{t}) \to -\infty$ if $\tilde{X}^{\hat{\pi}}_{t} \to \infty$.\newline Hence, $\lim_{\tilde{X}^{\hat{\pi}}_{t} \to \infty} \tilde{\pi}_c (\tilde{X}^{\hat{\pi}}_{t}) = \lim_{d_u \to -\infty}  \dfrac{\Phi(-d_u)}{\sigma \sqrt{T-t}(\phi(d_u) -\Phi(-d_u)d_u)}$ = 0, \newline because $\lim_{d_u \to -\infty} \Phi(-d_u) = 1$ and $\lim_{d_u \to -\infty}  \phi(d_u) = 0$. From Section 4.2.1 it is known that  $\tilde{\pi}_c(d)= - \tilde{\pi}_p(d).$ \end{proof}

For the behaviour of the options, it can also easily be seen that \newline $\lim_{\tilde{X}^{\hat{\pi}}_{t} \to -\infty} p(\tilde{X}^{\hat{\pi}}_{t}, t)$  = $\lim_{\tilde{X}^{\hat{\pi}}_{t} \to -\infty} e^{-r(T-t)} \mathbb{E}_{\mathbb{Q}} [\max\{K_l-\tilde{X}^{\hat{\pi}}_T,0\} |\tilde{X}^{\hat{\pi}}_{t}] = \infty $ and \newline  $\lim_{\tilde{X}^{\hat{\pi}}_{t} \to \infty}  e^{-r(T-t)} \mathbb{E}_{\mathbb{Q}} [\max\{K_l-\tilde{X}^{\pi}_T,0\} | \tilde{X}^{\hat{\pi}}_{t}] = 0$.
\subsection{$K_l$-$K_u$-Strategy: Illustration of the impact of $K_l$}
These plots were produced to check how the impact of $K_l$ affects the strategy. Extremly negative investments are reduced for lower $K_l$, as it is mentioned in the text.
\begin{figure}[H] 

 \begin{minipage}{0.58\linewidth} \subcaption{$K_l$= 500}  
\includegraphics[width=75mm]{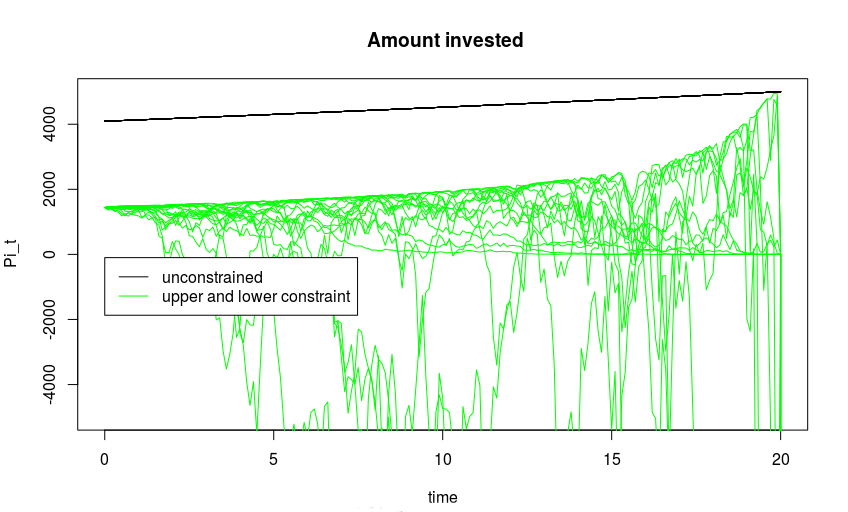} \end{minipage} \begin{minipage}{0.44\linewidth} \includegraphics[width=75mm]{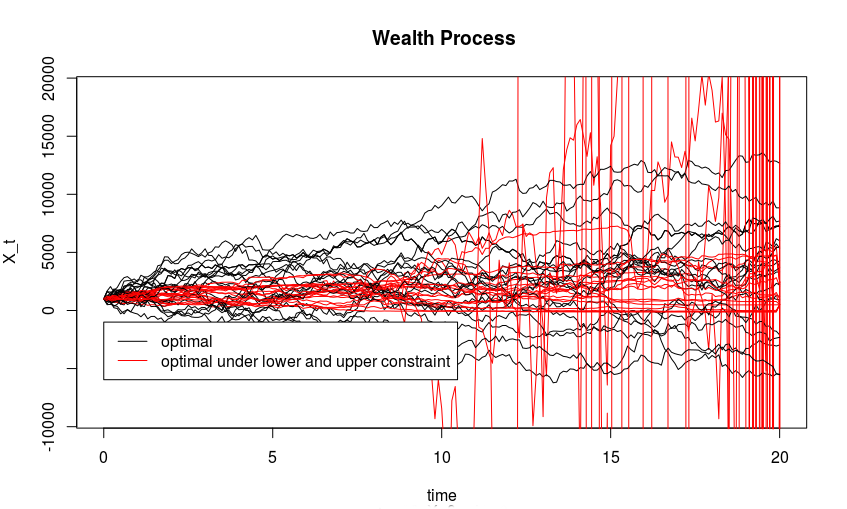}\end{minipage}
 \begin{minipage}{0.58\linewidth} \subcaption{$K_l$= -300} 
 \includegraphics[width=75mm]{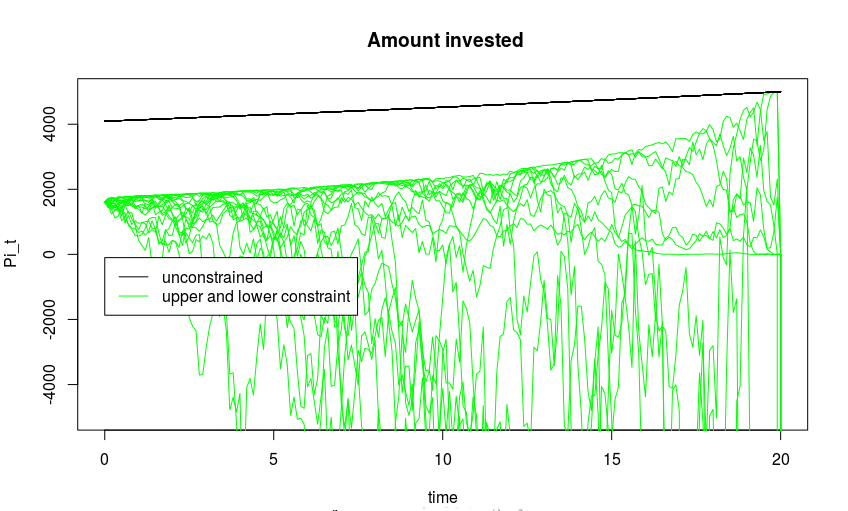} \end{minipage} \begin{minipage}{0.44\linewidth}\includegraphics[width=75mm]{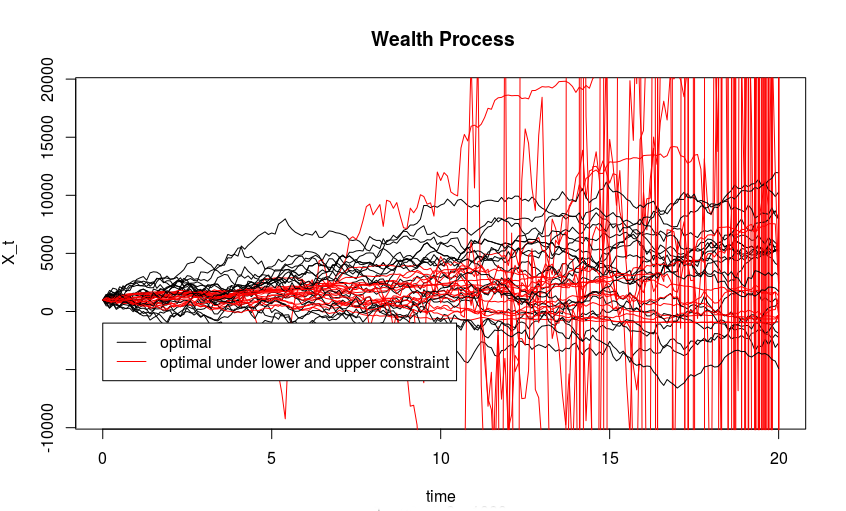} \end{minipage}   \newline
\caption{Amount invested over time and wealth process for $\hat{\pi}_{l,u}$ (different $K_l$)}
\end{figure} 
For an overview impression of the dimensions, find below displayed the performance of the strategies that had been developed so far. 
\begin{figure}[H]{ \begin{center}\includegraphics[width=69mm]{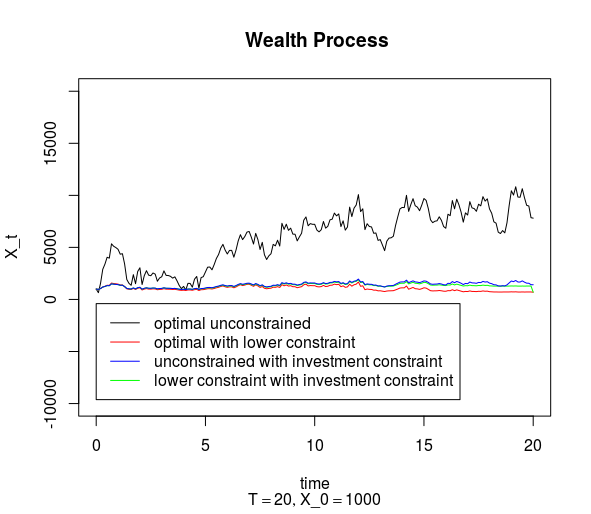} \end{center}}
\end{figure} 
\newpage

\section*{Appendix B}
\addcontentsline{toc}{section}{Appendix B}
\subsection{R-Code for Simulations}
\includegraphics{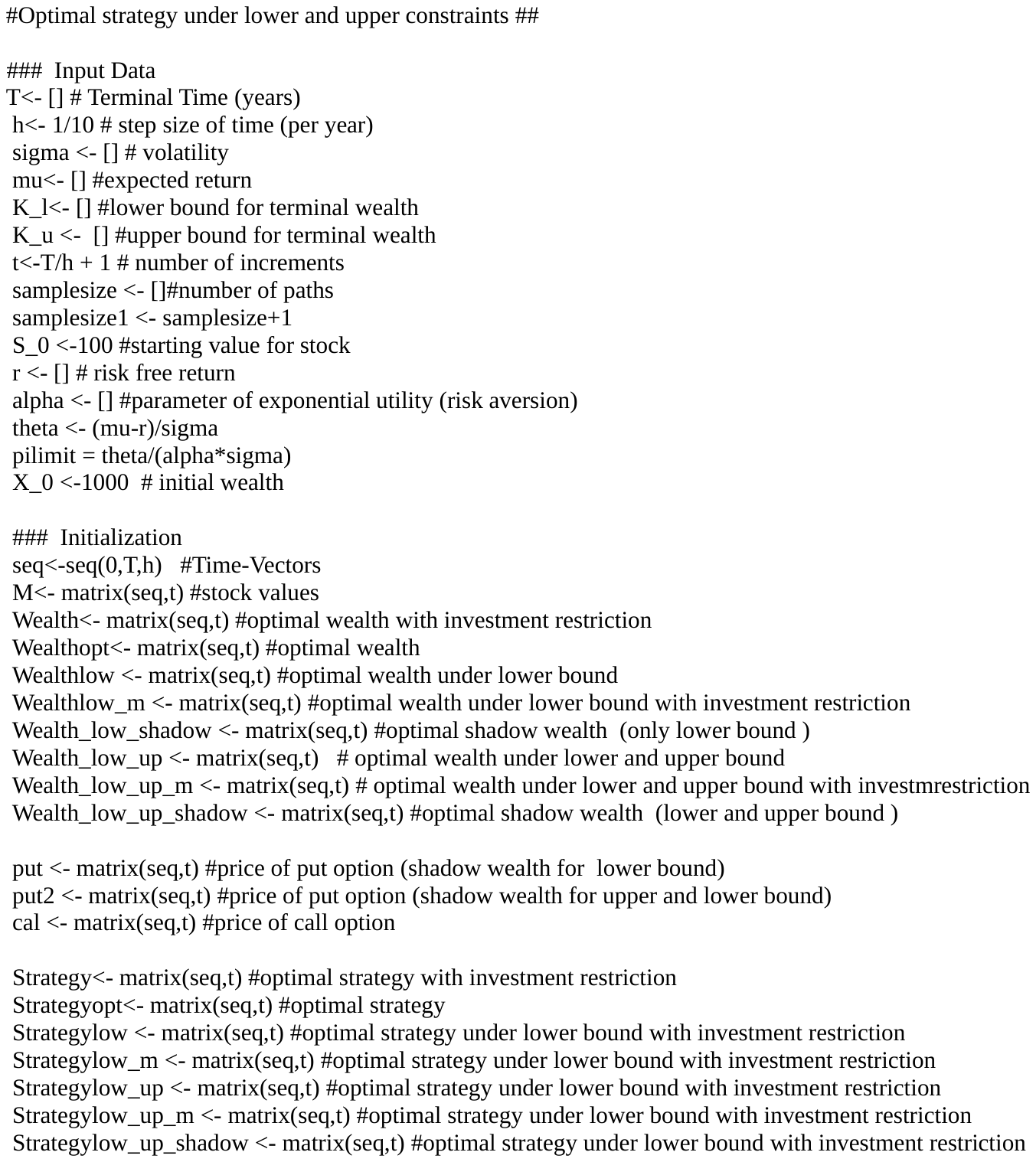}
\includegraphics{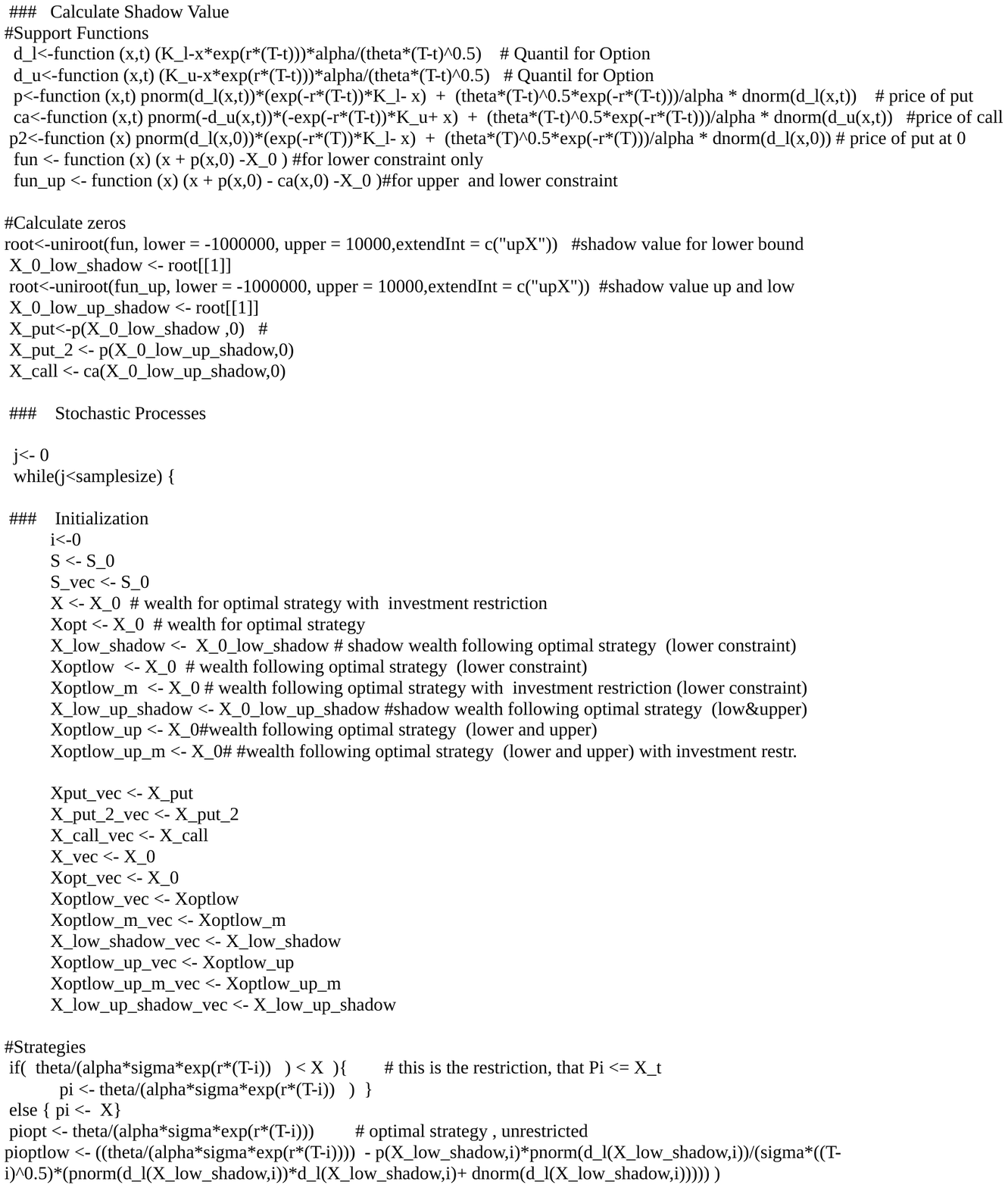}
\includegraphics{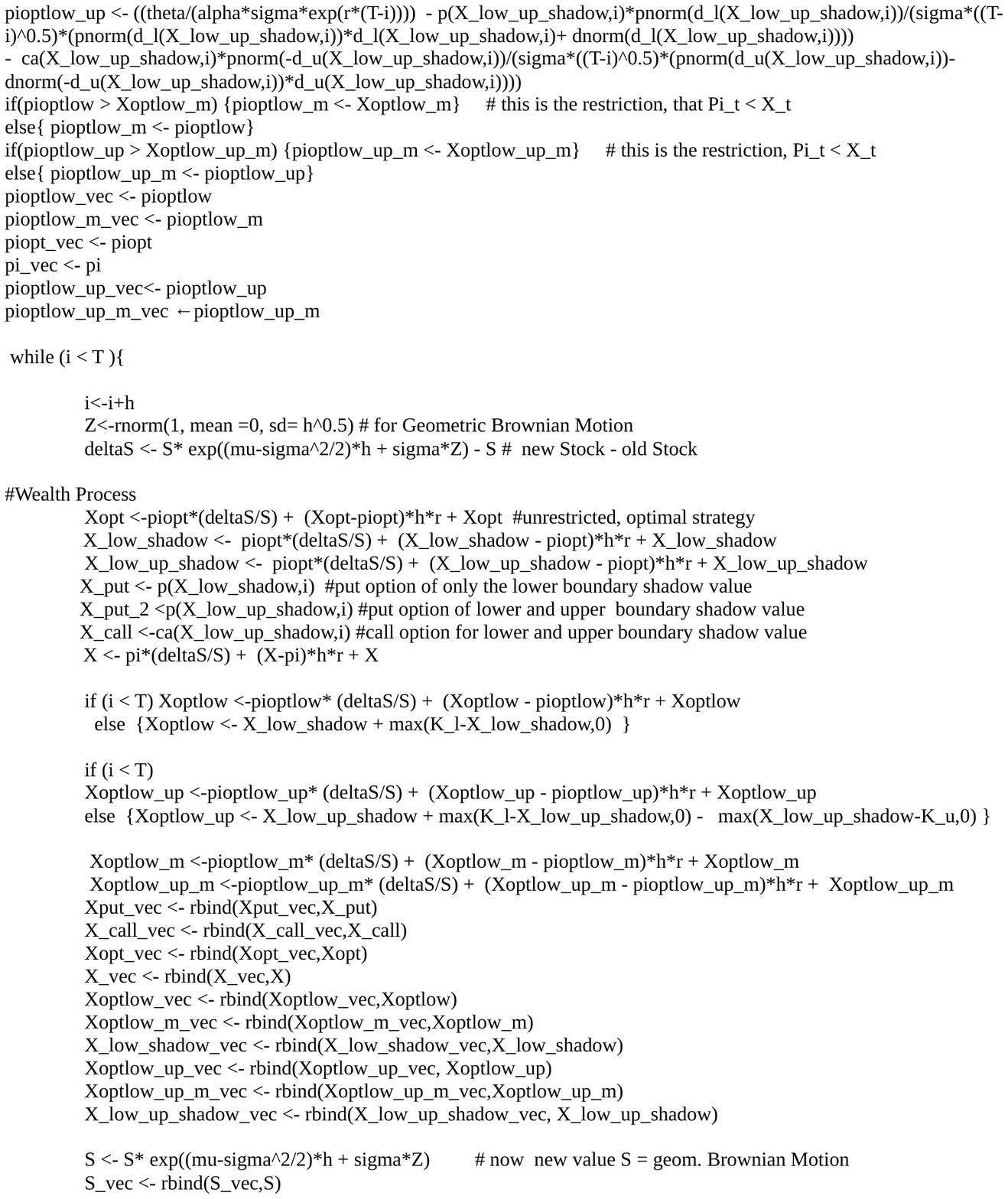}
\includegraphics{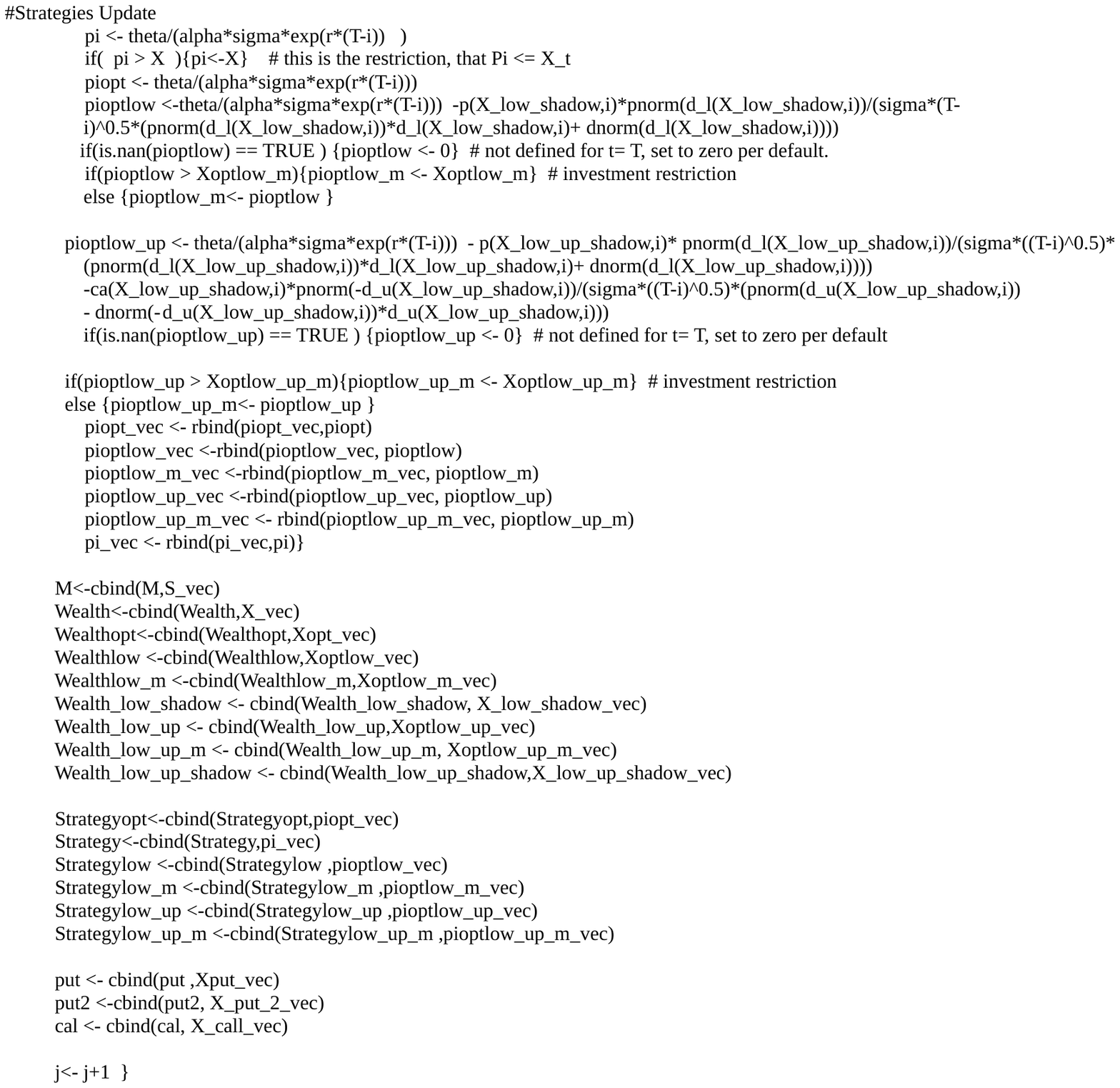}
\includegraphics{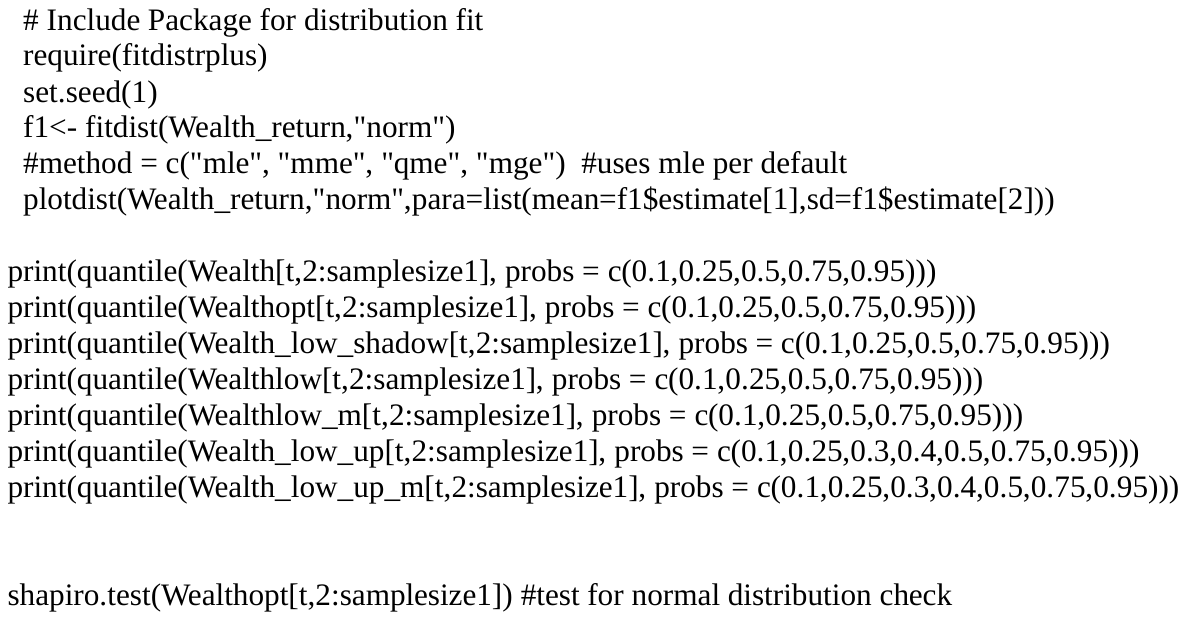}


\begingroup
\addcontentsline{toc}{section}{Bibliography}
\renewcommand*\refname{Bibliography}

\endgroup

\end{document}